\documentclass[11pt]{article}

\usepackage[margin=1in]{geometry}
\usepackage{graphicx}
\usepackage[colorlinks=true,linkcolor=blue,citecolor=blue,urlcolor=black]{hyperref} 
\usepackage{amsmath, amsthm, amssymb}
\usepackage{subfigure}
\usepackage{url}
\usepackage{xcolor}
\usepackage{tikz}
\usepackage{multicol}

\newcommand{\R}{\mathbb{R}}

\newcommand{\C}{\mathbb{C}}
\newcommand{\Z}{\mathbb{Z}}

\newcommand{\E}{\mathbb{E}}
\newcommand{\F}{\mathbb{F}}

\newcommand{\hsip}[2]{\langle #1,#2 \rangle}
\newcommand{\ip}[2]{\langle #1|#2 \rangle}
\newcommand{\proj}[1]{|#1\rangle \langle #1|}
\newcommand{\bracket}[3]{\langle #1|#2|#3 \rangle}

\newcommand{\mm}[1]{\begin{pmatrix} #1 \end{pmatrix}}

\DeclareMathOperator{\poly}{poly}
\DeclareMathOperator{\polylog}{polylog}

\DeclareMathOperator{\supp}{supp}
\DeclareMathOperator{\spann}{span}

\DeclareMathOperator{\QMA}{\mathsf{QMA}}

\newcommand{\be}{\begin{equation}}
\newcommand{\ee}{\end{equation}}
\newcommand{\bea}{\begin{eqnarray}}
\newcommand{\eea}{\end{eqnarray}}
\newcommand{\bes}{\begin{equation*}}
\newcommand{\ees}{\end{equation*}}
\newcommand{\beas}{\begin{eqnarray*}}
\newcommand{\eeas}{\end{eqnarray*}}

\newtheorem{thm}{Theorem}

\newtheorem{lem}[thm]{Lemma}
\newtheorem*{lem*}{Lemma}
\newtheorem{prop}[thm]{Proposition}

\newtheorem{question}{Question} 

\newtheoremstyle{definition}
   {}{}{}{0pt}{\bfseries}{.}{.5cm}
   {{\thmname{#1 }}{\thmnumber{#2}}{\thmnote{ (#3)}}}

\theoremstyle{definition}
\newtheorem{dfn}{Definition}

\def\01{\{0,1\}}
\newcommand{\eps}{\epsilon}
\newcommand{\Tr}{\mbox{\rm tr}}
\newcommand{\ADV}{\mbox{\rm ADV}}
\newcommand{\norm}[1]{\parallel #1 \parallel}
\newcommand{\ketbra}[2]{|#1\rangle\langle #2|}
\newcommand{\ceil}[1]{\lceil#1\rceil}
\newcommand{\Prop}{\mathcal{P}}
\newtheorem{claim}[thm]{Claim}

\input{Qcircuit.tex}

\begin{document}

\title{A Survey of Quantum Property Testing}
\author{Ashley Montanaro\thanks{Department of Computer Science, University of Bristol, UK; {\tt ashley@cs.bris.ac.uk}. Supported by an EPSRC Early Career Fellowship (EP/L021005/1).}
\and
Ronald de Wolf\thanks{CWI and University of Amsterdam, the Netherlands; {\tt rdewolf@cwi.nl}. Supported by a Vidi grant from the Netherlands Organization for Scientific Research (NWO) which ended in 2013, by ERC Consolidator Grant QPROGRESS, and by the European Commission IST STREP project Quantum Algorithms (QALGO) 600700.}}

\maketitle

\vspace{-22pt}

\begin{abstract}
The area of property testing tries to design algorithms that can efficiently handle very large amounts of data:
given a large object that either has a certain property or is somehow ``far'' from having that property, 
a tester should efficiently distinguish between these two cases.
In this survey we describe recent results obtained for \emph{quantum} property testing. 
This area naturally falls into three parts.
First, we may consider quantum testers for properties of classical objects.
We survey the main examples known where quantum testers can be much (sometimes exponentially) more efficient than classical testers.
Second, we may consider classical testers of quantum objects.
These arise for instance when one is trying to determine if quantum states or operations do what they are supposed to do, based only on classical input-output behavior.
Finally, we may also consider quantum testers for properties of quantum objects, such as states or operations.
We survey known bounds on testing various natural properties, such as whether two states are equal, whether a state is separable, whether two operations commute, etc. We also highlight connections to other areas of quantum information theory and mention a number of open questions.
\end{abstract}

\tableofcontents

\section{Introduction}

In the last two decades, the amounts of data that need to be handled have exploded: think of the massive amounts of data on the web, or the data warehouses of customer information collected by big companies.  In many cases algorithms need to decide whether this data has certain properties or not, without having sufficient time to trawl through all or even most of the data. 
Somehow we would like to detect the presence or absence of some global property by only making a few ``local'' checks.
The area of \emph{property testing} aims to design algorithms that can efficiently test whether some large object has a certain property,
under the assumption that the object either has the property or is somehow ``far'' from having that property.
An assumption like the latter is necessary for efficient property testing: deciding the property for objects that are ``just on the boundary'' typically requires looking at all or most of the object, which is exactly what we are trying to avoid here.
In general, different property testing settings can be captured by the following ``meta-definition'':

\begin{quote}
{\bf Property testing}\\
Let $\cal X$ be a set of objects and $d:{\cal X}\times{\cal X}\rightarrow[0,1]$ be a distance measure on $\cal X$.
A subset $\Prop\subseteq{\cal X}$ is called a \emph{property}. An object $x\in{\cal X}$ is \emph{$\eps$-far} from $\Prop$ if $d(x,y)\geq\eps$ for all $y\in \Prop$;
$x$ is \emph{$\eps$-close} to~$\Prop$ if there is a $y\in \Prop$ such that $d(x,y)\leq\eps$. 

An \emph{$\eps$-property tester} (sometimes abbreviated to $\eps$-tester) for~$\Prop$ is an algorithm that receives as input either an $x\in \Prop$ or an $x$ that is $\eps$-far from $\Prop$. In the former case, the algorithm accepts with probability at least $2/3$; in the latter case, the algorithm rejects with probability at least~$2/3$.
\end{quote}
Observe that, if an input is accepted by the property tester with high probability, then it must be $\eps$-close to $\Prop$. This is true for all inputs, including inputs neither in $\Prop$ nor $\eps$-far from $\Prop$. The value of~$2/3$ for the success probability is arbitrary and can equivalently be replaced with any other constant in $(1/2,1)$ since we can efficiently reduce the error probability by repeating the test a few times and taking the majority outcome.
We say that the tester has \emph{perfect completeness} if it accepts every state in $\Prop$ with certainty.
The distance parameter $\eps$ is usually taken to be some positive constant.
We will often just speak of a ``tester,'' leaving the value of $\eps$ implicit. 

Clearly, the above meta-definition leaves open many choices: what type of objects to consider, what property to test, what distance measure to use, what range of~$\eps$ to allow (the larger~$\eps$, the easier it should be to test~$\Prop$), and how to measure the complexity of the testing algorithm.
A lot of work in classical computer science has gone into the study of efficient testers for various properties, as well as proofs that certain properties are not efficiently testable, see for instance~\cite{blr:selftest,ggr,fis_sur,ron_sur,goldreich:prop}.
Typically,  ${\cal X}$ will be the set of all strings of length~$N$ over some finite alphabet, where we think of $N$ as being very large. The distance will usually be normalized Hamming distance $d(x,y)=|\{i: x_i\neq y_i\}|/N$, though also more sophisticated metrics such as ``edit distance'' have been used.  
The complexity of the tester is typically measured by the number of \emph{queries} it makes to entries of its input~$x$, 
and a tester is deemed efficient if its number of queries is much less than the length of the input~$N$, say $\polylog(N)$ 
or even some constant independent of~$N$.  This captures the goal that a tester is able to efficiently handle huge amounts of data.
The distance bound~$\eps$ is often taken to be a small fixed constant, but in some cases it is also interesting to quantify the dependence of the tester's complexity on $\eps$ as well as on~$N$. 
For example, a tester whose complexity is $\Theta(2^{2^{1/\eps}})$ might be considered to be of little use in practice.

As an initial (very simple) example, suppose our property $\Prop=\{0^N\}$ consists of only one object, the all-zero string, and we use normalized Hamming distance.  Our input~$x$ will either be in $\Prop$ (i.e., $x=0^N$) or $\eps$-far from $\Prop$ (i.e., $x$ has at least $\eps N$ 1-bits). An obvious tester would choose $k$ indices in the string at random, query them, and reject if and only if there is a 1 in at least one of those positions. This tester accepts $x=0^N$ with certainty (so it has perfect completeness), and fails to reject an input that is $\eps$-far from $\Prop$ with probability $(1-\eps)^k$.  Choosing $k=\Theta(1/\eps)$ gives a tester with small constant error probability, and this number of queries can be shown to be optimal.\footnote{
Note that the complexity of a property can differ much from that of its complement.  For example, $\overline{\Prop}=\01^N\backslash\{0^N\}$ is trivial to test if $\eps>1/N$: no string is $\eps$-far from $\overline{\Prop}$, so we might as well accept every input without querying anything.}

In this survey paper we will be concerned with \emph{quantum} property testing. There are several natural ways in which one can generalize property testing to the quantum world:
\begin{itemize}
\item Quantum testing of properties of classical objects. In this setting we would like to achieve provable quantum speed-ups over any possible classical algorithm, or to prove limitations on property testers, even if they are allowed to be quantum. By their very nature, efficient quantum query algorithms rely on extracting global information about the input, by querying in superposition; property testing is thus a plausible place to find significant quantum speed-ups.
A very simple example of such a speed-up is for the above-mentioned property $\Prop=\{0^N\}$: a tester based on Grover's search algorithm~\cite{grover:search} would use $O(1/\sqrt{\eps})$ queries, in contrast to the $\Theta(1/\eps)$ queries that classical testers need. 

\item Classical testing of properties of quantum objects. Here we imagine we are given a black-box device which is claimed to implement some quantum process, and we would like to test whether it does what is claimed. However, our access to the device is classical: all we can do is feed classical inputs to the device, and receive classical measurement outcomes.

\item Quantum testing of properties of quantum objects. In this most general scenario, we are given access to a quantum state or operation as a black box, and apply a quantum procedure to it to test whether it has some property.
\end{itemize}
We will discuss each of these settings in turn. We usually concentrate on describing the intuition behind prior work, without giving detailed proofs. Some of the results we present have not appeared in the literature before; these are largely based on combining, generalizing or improving existing works. Various open questions are pointed out throughout the survey.

A vast amount of work in quantum computing can be interpreted through the lens of property testing. Indeed, taken to extremes, any efficient quantum algorithm for a decision problem could be seen as an efficient property tester, 
and any measurement scheme that tries to learn properties of a quantum state or channel could be seen as a quantum property tester. We therefore concentrate on covering those algorithms which can clearly be understood as distinguishing objects with some property from those ``far'' from that property, and we make no attempt to be completely comprehensive. Also, our focus is on the computer-science aspects of the field, rather than work which primarily takes a physics perspective, such as the study of interaction-free measurement and the flourishing field of quantum metrology. Finally, we do not attempt to cover the (now very extensive) field of classical testers for classical properties. For much more on these, see the references given earlier. 

\subsection{Quantum testing of classical properties}

In the first part of this paper we will consider quantum testing of \emph{classical} properties.  Again, ${\cal X}$ will typically be the set of all strings of length~$N$ over some finite alphabet, the distance will be normalized Hamming distance, and the complexity of both quantum and classical property testers will be measured by the number of queries to the input~$x$.

One of our goals is to survey examples of quantum speed-up, i.e., describe properties where the complexity of quantum testers is substantially less than the complexity of classical testers. Most known quantum speed-ups for testing classical properties were derived from earlier improvements in query complexity: they rely on quantum algorithms such as those of (in chronological order) Bernstein and Vazirani~\cite{bernstein&vazirani:qcomplexity}, Simon~\cite{simon:power}, Shor~\cite{shor:factoring}, Grover~\cite{grover:search}, and Ambainis~\cite{ambainis:edj}. In Section~\ref{sec:claspropupperbounds} we describe these quantum property testers and the improvements they achieve over classical testers. Some of the properties considered are very natural, and some of the improvements achieved are quite significant.

In Section~\ref{sec:clasproplowerbounds} we describe some \emph{lower}-bound methods for quantum property testing, i.e., methods
to show query complexity lower bounds for quantum algorithms that want to test specific properties.  The main lower bounds in this area have
been obtained using the \emph{polynomial} method.  We also describe the \emph{adversary} method, which---when applied properly---proves
optimal lower bounds.  And we ask whether the recent classical property testing lower bounds of Blais et al.~\cite{bbmj}, based
on communication complexity, can be applied to quantum testers as well.

\subsection{Classical testing of quantum properties}

In the second part we will consider \emph{classical} testing of quantum properties.
At first sight, this scenario might make no sense---how could a classical algorithm,
without the ability to perform any quantum operations, be able to test quantum objects?
But suppose someone gives us a quantum state and claims it is an EPR-pair. Or someone builds a quantum device to implement a Hadamard gate, or to measure in a specific basis. How can we test that these quantum objects conform to their specifications?
These are questions often faced for instance by experimentalists who try to check that their quantum operations work as intended, or by parties who run quantum cryptographic hardware provided by an untrusted supplier.
We do not want to assume here that we already have the ability to implement some other quantum operations reliably,
because that would lead to an infinite regress: how did we establish that those other quantum objects are reliable?
Accordingly, we somehow would like to test the given quantum object while only trusting our \emph{classical} devices. Of course, in order to test a quantum object there has to be at least \emph{some}
interaction with the quantum object-to-be-tested.  In the testers we consider, the only quantum
involvement is with that object itself in a black-box fashion (whence the name ``self-testing''):
we can only observe its classical input-output behavior, but not its inner quantum workings.

This notion of quantum self-testing was introduced by Mayers and Yao~\cite{mayers&yao:impapp,mayers&yao:selftesting}, who described a procedure to test photon sources that are supposed to produce EPR-pairs. Since then quite a lot of work has been done on self-testing. We focus on two areas for self-testing: in Section~\ref{sec:stgates} we describe self-testing of universal sets of quantum \emph{gates} gates and in Section~\ref{sec:stprotocols} we describe the self-testing of \emph{protocols} for two or more parties, focusing on protocols for the so-called \emph{CHSH game}.  Self-testing of protocols has found many applications in the fast-growing area of \emph{device-independent} quantum cryptography, where parties want to run cryptographic protocols for things like key distribution or randomness generation, using quantum states or apparatuses (photon sources, measuring devices, etc.) that they do not fully trust.  Self-testing the states or apparatuses makes this possible in some cases. Device-independent cryptography is quite a large area and we will not cover it in this survey; see, e.g., \cite{bhk:crypto,colbeck:phd,ABGMPS:qkd,vazirani&vidick:dice,vazirani&vidick:qkd} for more about this area.

\subsection{Quantum testing of quantum properties}

In the final part of the paper we will consider cases where $\cal X$ is a set of quantum objects and our tester is also quantum, which is a setting of both theoretical and experimental interest.

As experimentalists control ever-larger quantum systems in the lab, the question of how to characterize and certify these systems becomes ever more pressing. Small quantum systems can be characterized via a procedure known as \emph{quantum state tomography}~\cite{paris04,nielsen&chuang:qc}. However, completely determining the state of a system of $n$ qubits necessarily requires exponentially many measurements in $n$. This is already a daunting task for fairly small experiments; for example, H\"affner et al.~\cite{haffner05} report tomography of a state of 8 ions requiring 656,100 experiments and a total measurement time of 10 hours. One way of reducing this complexity is to start with the assumption that the state is of a certain form (such as a low-rank mixed state~\cite{gross09b,flammia12} or a matrix product state~\cite{cramer10}), in which case the number of parameters required to be estimated can be dramatically reduced. The viewpoint of property testing suggests another approach: the direct determination of whether or not something produced in the lab has a particular property of interest, under the assumption that it either has the property or is far away from it.

One can view classical property testing algorithms in two ways: either as testing properties of data (such as graph isomorphism), or properties of functions (such as linearity). If one wishes to generalize property testing to the quantum realm, one is thus naturally led to two different generalizations: testing properties of quantum \emph{states}, and properties of quantum \emph{operations}. One can divide each of these further, according to whether the state is pure or mixed, and whether the operation is reversible or irreversible; this classification is illustrated in Table \ref{tab:taxonomy}. We discuss each of these possibilities in Sections~\ref{sec:states} and~\ref{sec:dynamics}. Within some of these categories there are natural generalizations of properties studied classically. For example, testing properties of mixed states is analogous to the classical idea of testing properties of probability distributions. Some quantum properties, however, have no simple classical analog (such as properties relating to entanglement).

\begin{table}[htb]
\begin{center}
\begin{tabular}{|c|c|c|}
\hline & {\bf Coherent} & {\bf Incoherent}\\
\hline {\bf Static} & Pure state (\textsection\ref{sec:pure}) & Mixed state (\textsection\ref{sec:mixed})\\
\hline {\bf Dynamic} & Unitary operator (\textsection\ref{sec:unitary}) & Quantum channel (\textsection\ref{sec:channel}) \\
\hline
\end{tabular}
\end{center}
\caption{The taxonomy of quantum properties}
\label{tab:taxonomy}
\end{table}

Classically, there are many connections known between property testing and computational complexity. In Section~\ref{sec:compcomp} we explore the link between quantum property testing and quantum computational complexity, including the use of property testers to prove results in computational complexity, and the use of computational complexity to prove limitations on property testers.

\section{Quantum testing of classical properties}\label{sec:qtestclprop}

\subsection{Preliminaries}

We will use $[m]$ to denote $\{1,\ldots,m\}$ and $\mathbb{Z}_m$ to denote $\{0,\ldots,m-1\}$ modulo~$m$.
When considering (quantum or classical) testers for classical objects, those classical objects are usually strings, ${\cal X}=[m]^N$, and the complexity of testers is measured by the number of \emph{queries} they make to their input~$x$.
In some cases we let $x$ correspond to a function $f:[N]\rightarrow[m]$, where $f(i)=x_i$, and~$i$ may be viewed as either an integer $i\in[N]$, or as its binary representation $i\in\01^{\ceil{\log N}}$. 

Here we briefly define the quantum query model, referring to~\cite{buhrman&wolf:dectreesurvey} for more details. We assume some basic familiarity with classical and quantum computing~\cite{nielsen&chuang:qc}.

Informally, a query allows us to ``read'' $x_i$ for any~$i$ of our choice.
Mathematically, to make this correspond to a quantum operation, it is modeled by the unitary map
$$
O_x:\ket{i}\ket{b}\mapsto\ket{i}\ket{b+x_i}.
$$
Here the first register has dimension~$N$ and the second has dimension~$m$.
The answer $x_i$ is added into this second register mod~$m$.
Part of the power of quantum query algorithms comes from their ability to apply a query to a \emph{superposition} of different $i$s, thus globally ``accessing'' the input~$x$ while using only one query. 

If $m=2$, then putting the state $\ket{-}=\frac{1}{\sqrt{2}}(\ket{0}-\ket{1})$
in the second register has the following effect:
$$
O_x:\ket{i}\ket{-}\mapsto\ket{i}\frac{1}{\sqrt{2}}(\ket{0+x_i}-\ket{1+x_i})=(-1)^{x_i}\ket{i}\ket{-}.
$$
We will sometimes call this a ``phase-query,'' because the answer~$x_i$ to the query is inserted in the state as a phase ($+1$ if $x_i=0$, and~$-1$ if $x_i=1$).

A $T$-query quantum algorithm is described by an initial state, say $\ket{0^k}$, and $T+1$ fixed $k$-qubit unitaries $U_0,\ldots,U_T$. The final state when the algorithm runs on input~$x$ is obtained by interleaving these unitaries with queries to~$x$ ($O_x$ may be tensored with the identity operation on the remaining workspace qubits),
$$
\ket{\psi_x}=U_TO_xU_{T-1}O_x\cdots O_xU_1O_xU_0\ket{0^k}.
$$
This final state depends on~$x$ via the $T$ queries.
A measurement of the final state will determine the classical output of the algorithm.

\subsection{Upper bounds}\label{sec:claspropupperbounds}

In this section we survey the main speed-ups that have been obtained using \emph{quantum} testers for \emph{classical} properties. Typically these apply pre-exisiting quantum algorithms to problems in property testing. Our distance measure will be normalized Hamming distance $d(x,y)=|\{i: x_i\neq y_i\}|/N$, unless explicitly stated otherwise.

\subsubsection{Using amplitude amplification}\label{sssec:amplampl}

A simple but very general way that quantum algorithms can speed up many classical property testers is via the powerful primitive of \emph{amplitude amplification}, which was introduced by Brassard et al.~\cite{bhmt:countingj} and can be seen as a generalization of Grover's quantum search algorithm~\cite{grover:search}. We assume we are given query access to some function $f$ (treated as a black box), and have a quantum algorithm which, with probability $p$, outputs $w$ such that $f(w)=1$. Then the result of Brassard et al.\ is that, for any $p>0$, we can find a $w$ such that $f(w)=1$ with $O(1/\sqrt{p})$ queries to $f$, with success probability at least $2/3$.

Amplitude amplification can be immediately applied to speed up classical property testers which have perfect completeness. Here we think of $w$ as the internal randomness of the algorithm, and $f(w)$ as the test which is applied to the unknown object, based on the random bits $w$. We let $f(w)=0$ if the test accepts, and $f(w)=1$ if the test rejects. Assuming that the test has perfect completeness, finding $w$ such that $f(w)=1$ is equivalent to determining whether we should reject. Given that the original test used $q$ queries to find such a $w$ with probability $p>0$, we therefore obtain a test which uses $O(q/\sqrt{p})$ queries, still has perfect completeness, and rejects with constant probability.

For example, consider the well-studied classical property of {\bf Linearity}~\cite{blr:selftest}. A function $f:\{0,1\}^n \rightarrow \{0,1\}$ is said to be \emph{linear} if $f(x \oplus y) = f(x) \oplus f(y)$, and \emph{affine} if $f(x \oplus y) = f(x) \oplus f(y) \oplus 1$, where $\oplus$ is addition modulo 2. (Linearity is equivalent to the condition $f(x) = \bigoplus_{i \in S} x_i$ for some $S \subseteq [n]$.) A simple and natural test for linearity is to pick $x,y \in \{0,1\}^n$ uniformly at random and accept if and only if $f(x) \oplus f(y) = f(x \oplus y)$. This test uses only 3 queries, has perfect completeness, and can be shown~\cite{bellare96} to reject functions~$f$ which are $\eps$-far from linear with probability at least~$\eps$. Applying amplitude amplification to this tester, we immediately get a quantum $\eps$-tester for {\bf Linearity} which uses $O(1/\sqrt{\eps})$ queries. Another simple example is {\bf Symmetry}, where $f:\{0,1\}^n \rightarrow \{0,1\}$ is said to be symmetric if $f(x)$ depends only on $|\{i:x_i=1\}|$. A classical tester for this property has been given by Majewski and Pippenger~\cite{majewski09}. The tester uses 2 queries, has perfect completeness and rejects functions which are $\eps$-far from symmetric with probability at least $\eps$. Therefore, we again obtain a quantum $\eps$-tester which uses $O(1/\sqrt{\eps})$ queries.

Hillery and Andersson~\cite{hillery11} gave different quantum testers for these two properties (though also based on amplitude amplification), each of which uses $O(\eps^{-2/3})$ queries, which is worse. More recently, Chakraborty and Maitra~\cite{chakraborty13} described a quantum algorithm for the closely related problem of testing whether a Boolean function is affine. Their algorithm also uses $O(1/\sqrt{\eps})$ queries and, although presented slightly differently, is also based on amplitude amplification.

\subsubsection{Using the Bernstein-Vazirani algorithm}\label{sec:bvtester}

One of the first quantum algorithms was the Bernstein-Vazirani algorithm~\cite{bernstein&vazirani:qcomplexity}. 
It efficiently decodes a given \emph{Hadamard codeword}. Let $N=2^n$, and identify $[N]$ with $\01^n$ so we can use the $n$-bit strings to index the numbers $1,\ldots,N$.\footnote{In many presentations of the Bernstein-Vazirani, Simon, and Grover algorithms, the input is taken to be a function $f:\01^n\rightarrow\01$ rather than a string $x\in\01^N$. With $N=2^n$, these two views are of course just notational variants of one another.}
Let $h:\01^n\rightarrow\01^N$ be the Hadamard encoding, defined by
$h(s)_i=s\cdot i$ mod~2; this is nothing more than identifying $s$ with the linear function $h(s)(i)=s \cdot i$ mod~2 and writing out its truth table.  Note that two distinct Hadamard codewords $h(s)$ and $h(s')$ are at normalized Hamming distance exactly~$1/2$.
Given input $h(s)$, the Bernstein-Vazirani algorithm recovers~$s$ with probability~1 using only one quantum query. 
In contrast, any classical algorithm needs $\Omega(\log N)$ queries for this.
The quantum algorithm works as follows:
\begin{enumerate}
\item Start with $\ket{0^n}$ and apply Hadamard gates to each qubit to form the uniform superposition $\frac{1}{\sqrt{N}}\sum_{i\in\01^n}\ket{i}$
\item Apply a phase-query to obtain $\frac{1}{\sqrt{N}}\sum_{i\in\01^n}(-1)^{x_i}\ket{i}$
\item Apply Hadamard transforms to each qubit and measure.
\end{enumerate}
If $x_i=s\cdot i$ for all~$i\in\01^n$, then it is easy to see that the measurement yields~$s$ with probability~1.

Buhrman et al.~\cite{bfnr:qpropj} showed this algorithm can be used to obtain an unbounded quantum speed-up for \emph{testing} most subsets of Hadamard codewords.

\begin{quote}
{\bf Bernstein-Vazirani property} for $A\subseteq\01^n$:\\ 
$\Prop^A_{BV}=\{x\in\01^N: \exists s\in A\mbox{ such that }x=h(s)\}$
\end{quote}

\begin{thm}[Buhrman et al.~\cite{bfnr:qpropj}]
For every $A\subseteq\01^n$ there is an $O(1/\sqrt{\eps})$-query quantum $\eps$-tester for $\Prop^A_{BV}$; in contrast, for a $1-o(1)$ fraction of all sets $A$, every classical 1/2-tester for $\Prop^A_{BV}$ needs $\Omega(\log N)$ queries.
\end{thm}

\begin{proof}
{\bf Quantum upper bound.}
We run the Bernstein-Vazirani algorithm on input~$x$, which takes one quantum query. 
The algorithm will output some~$s$, and if $x$ equals some $h(s)\in \Prop^A_{BV}$ then this will be the corresponding~$s$ with certainty.
Hence if $s\not\in A$ we can reject immediately.
If $s\in A$ then choose $i\in[N]$ at random, query~$x_i$, and test whether indeed $x_i=s\cdot i$.
If $x$ is $\eps$-far from $\Prop^A_{BV}$ then this test will fail with probability~$\eps$.
Using amplitude amplification, we can detect any $x$ that is $\eps$-far from $\Prop^A_{BV}$ with success probability at least~$2/3$
using $O(1/\sqrt{\eps})$ queries.

{\bf Classical lower bound.}
Choose the set $A\subseteq\01^n$ uniformly at random.
Consider the uniform input distribution over the set $H$ of all $N$ Hadamard codewords.
Note that the Hadamard codewords that are not in $\Prop^A_{BV}$ are 1/2-far from $\Prop^A_{BV}$,
because any two distinct Hadamard codewords have normalized Hamming distance exactly~1/2.
Hence if $\Prop^A_{BV}$ can be $1/2$-tested with $T$ queries, then there exists a decision tree (i.e., a deterministic query algorithm) that is correct on at least 2/3 of the $x\in H$.
Fix a deterministic decision tree ${\cal T}$ of depth~$T$.  For each $x\in H$, the probability (over the choice of~$A$) that $x\in\Prop^A_{BV}$ is 1/2, irrespective of the output that ${\cal T}$ gives on~$x$, so the probability that ${\cal T}$ correctly decides~$x$ is~1/2.  
Then the probability that ${\cal T}$ correctly decides at least~2/3 of the $x\in H$ is $2^{-\Omega(N)}$ by a Chernoff bound.  
The total number of deterministic decision trees of depth~$T$ is at most $2^{2^T}N^{2^T-1}$, because for each of the (at most) $2^T-1$ internal nodes we have to choose an index to query, and for each of the (at most) $2^T$ leaves we have to choose a binary output value.
Hence by the union bound, the probability (over the choice of~$A$) that there exists a depth-$T$ decision tree that correctly decides at least 2/3 of the $x\in H$ is at most
$$
2^{-\Omega(N)}\cdot 2^{2^T}N^{2^T-1}.
$$
For $T=(\log N)/2$ this quantity is negligibly small.
This shows that a $1-o(1)$ fraction of all possible sets~$A$, there is no classical tester for $\Prop^A_{BV}$ with $(\log N)/2$ or fewer queries.
\end{proof}

As Buhrman et al.~\cite{bfnr:qpropj} noted, the above classical lower bound is essentially optimal because for every property $\Prop\subseteq\01^N$ there exists an $\eps$-tester with $T=\lceil \ln(3|\Prop|)/\eps \rceil$ queries, as follows.  We just query the input $x\in\01^N$ at $T$ uniformly randomly chosen positions, and accept if and only if there is still at least one element $y\in\Prop$ that is consistent with all query outcomes. Clearly, if the input is in $\Prop$ this test will accept, so it has perfect completeness. If the input is $\eps$-far from $\Prop$, then the probability for a specific $y\in\Prop$ to ``survive'' $T$~queries is at most $(1-\eps)^T$. Hence by the union bound the probability that there is a $y\in\Prop$ surviving all $T$~queries is at most $|\Prop|\cdot (1-\eps)^T\leq |\Prop|\cdot e^{-\eps T}\leq 1/3$. 

\subsubsection{Testing juntas}
\label{sec:juntas}

Let $f:\01^n\rightarrow\{+1,-1\}$ be a Boolean function (such an $f$ can also be viewed as a string~$x$ of $N=2^n$ bits, with $x_i=f(i)$), and $J\subseteq[n]$ be the set of (indices of) variables on which $f$ depends.  If $|J|\leq k$ then $f$ is called a \emph{$k$-junta}.

\begin{quote}
{\bf $k$-junta property}: $\Prop_{k\text{-junta}}=\{f:\01^n\rightarrow\{+1,-1\} : f\mbox{ depends on at most $k$ variables}\}$
\end{quote}

The best known classical tester, due to Blais, uses $O(k\log k + k/\eps)$ queries~\cite{blais_juntas}, and the best known classical lower bound is $\Omega(k)$~\cite{chockler&gutfreund} (for fixed~$\eps$).

At{\i}c{\i} and Servedio~\cite{atici&servedio:testing} gave an elegant quantum $\eps$-property tester for $\Prop_{k\text{-junta}}$
using $O(k/\eps)$ quantum queries, slightly better than Blais's classical tester.%
\footnote{In fact, at the time~\cite{atici&servedio:testing} was written the best classical upper bound was only $O((k\log k)^2/\eps)$~\cite{fkrss:juntas}.}

\begin{thm}[essentially At{\i}c{\i} and Servedio~\cite{atici&servedio:testing}]
There is a quantum tester for $k$-juntas that uses $O(k/\sqrt{\eps})$ queries.
\end{thm}

``Essentially'' in the attribution of the above theorem refers to the fact that \cite{atici&servedio:testing} proves an $O(k/\eps)$ bound.
We observe here that the dependence on $\eps$ can easily be improved by a square root using amplitude amplification.

\begin{proof}
The basic quantum subroutine is the same as the Bernstein-Vazirani algorithm in Section~\ref{sec:bvtester}:
\begin{enumerate}
\item Start with $\ket{0^n}$ and apply Hadamard gates to each qubit to form the uniform superposition $\frac{1}{\sqrt{N}}\sum_{i\in\01^n}\ket{i}$
\item Apply a phase-query to obtain $\frac{1}{\sqrt{N}}\sum_{i\in\01^n}f(i)\ket{i}$
\item Apply Hadamard transforms to each qubit and measure.
\end{enumerate}
Let us analyze this subroutine by means of some Fourier analysis on the Boolean cube (see~\cite{odonnell:analysis,wolf:fouriersurvey} for background).
For every $s\in\01^n$, let 
$$
\widehat{f}(s)=\frac{1}{2^n}\sum_{i\in\01^n}f(i)(-1)^{i\cdot s}
$$
be the corresponding Fourier coefficient.
Going through the steps of the quantum subroutine, it is easy to see that the final state before the measurement is
$$
\sum_{s\in\01^n}\widehat{f}(s)\ket{s}.
$$
Accordingly, the final measurement will sample an~$s\in\01^n$ from the distribution given by the squared Fourier coefficients $\widehat{f}(s)^2$.
This procedure is known as \emph{Fourier Sampling}~\cite{bernstein&vazirani:qcomplexity}. It uses one query to~$f$.

Let $J$ be the set of variables on which the input~$f$ depends.  The goal of the tester is to decide whether $|J|\leq k$ or not.
Identifying sets $s\subseteq[n]$ with their characteristic vectors $s\in\01^n$, note that $\widehat{f}(s)\neq 0$ only if the support of $s$ lies within~$J$, so each Fourier Sample gives us a subset of~$J$. The tester will keep track of the union~$W$ of the supports seen so far. We will always have $W\subseteq J$, so if $f$ is a $k$-junta then $W$ will never have more than $k$ elements. On the other hand, below we show that if $f$ is $\eps$-far from any $k$-junta, then with high probability after $O(k/\sqrt{\eps})$ queries $W$ will end up having more than $k$ elements.

For a subset $W\subseteq[n]$ of size at most~$k$, define $g_W(i)=\sum_{s\subseteq W}\widehat{f}(s)(-1)^{i\cdot s}$. This function~$g_W$ need not be a Boolean function, but we can consider the Boolean function~$h_W$ that is the sign of~$g_W$. 
This~$h_W$ only depends on the variables in~$W$, so it is a $k$-junta and hence $\eps$-far from~$f$. Now we have
\begin{align*}
\eps & \leq \frac{1}{2^n}\sum_{i:f(i)\neq h_W(i)}1\\
& \leq \frac{1}{2^n}\sum_{i:f(i)\neq h_W(i)}(f(i)-g_W(i))^2\\
& \leq\E_{i\in\01^n}[(f(i)-g_W(i))^2]\\
& =\sum_{s}(\widehat{f}(s)-\widehat{g_W}(s))^2\\
& =\sum_{s\not\subseteq W}\widehat{f}(s)^2,
\end{align*}
where the first equality is Parseval's identity.
But this means that with probability at least~$\eps$, Fourier Sampling will output an~$s$ that is not fully contained in~$W$. Now we use amplitude amplification to find such an~$s$, using an expected number of $O(1/\sqrt{\eps})$ queries, and set $W:=W\cup s$ (so $W$'s size grows by at least one).
Repeating this at most $k+1$~times, after an expected number of $O(k/\sqrt{\eps})$ queries the set~$W$ (which was initially empty) will contain more than~$k$ variables and we can reject the input.
\end{proof}

Very recently Ambainis et al.~\cite{abrw:juntatesting} came up with a quantum $k$-junta tester that uses only $\widetilde{O}(\sqrt{k/\eps})$ queries.\footnote{The $\widetilde{O}(\cdot)$ notation hides logarithmic factors in~$k$.}  Unlike the tester of At{\i}c{\i} and Servedio, this actually beats the best known classical lower bound. The algorithm of~\cite{abrw:juntatesting} uses the adversary bound (see Section~\ref{ssecadversary} below), building upon quantum algorithms due to Belovs~\cite{belovs:juntalearning} for \emph{learning} the relevant variables of the junta. Their algorithm is substantially more complicated than the above, and we will not explain it here. They also give an implementation of their algorithm with \emph{time complexity} (i.e., number of quantum gates used) $\widetilde{O}(n\sqrt{k/\eps})$. They prove a quantum lower bound of $\Omega(k^{1/3})$ queries, leaving open the following:

\begin{question}
What is the quantum query complexity of testing juntas?
\end{question}

\subsubsection{Using Simon's algorithm}\label{sec:simontester}

The first exponential speed-up for quantum property testing was obtained by Buhrman et al.~\cite{bfnr:qpropj}.
It is inspired by Simon's algorithm~\cite{simon:power}, which was the first algorithm to have a provable exponential speed-up over classical algorithms in the black-box model and inspired Shor's factoring algorithm~\cite{shor:factoring} (which we will see in the next section).
Again let $N=2^n$ and identify $[N]$ with $\01^n$.
Consider an input $x\in[N]^N$ for which there exists an $s\in\01^n\backslash\{0^n\}$ such that $x_i=x_j$ if and only if ($j=i$ or $j=i\oplus s$).
Simon's algorithm finds~$s$ with high probability using $O(\log N)$ queries. The core of the algorithm is the following quantum subroutine:
\begin{enumerate}
\item Start with $\ket{0^n}\ket{0^n}$ and apply Hadamard transforms to the first $n$ qubits to form $\frac{1}{\sqrt{N}}\sum_{i\in\01^n}\ket{i}\ket{0^n}$
\item Apply a query to obtain $\frac{1}{\sqrt{N}}\sum_{i\in\01^n}\ket{i}\ket{x_i}$
\item Measure the second register. This yields some $z=x_i$ and collapses the first register to the two indices with value~$z$:
$\frac{1}{\sqrt{2}}(\ket{i}+\ket{i\oplus s})$
\item Apply Hadamard transforms to the first $n$ qubits and measure the state, obtaining some $y\in\01^n$.
\end{enumerate}
It is easy to calculate that the measured state is (up to phases) a uniform superposition over all $2^{n-1}$ strings $y\in\01^n$ that satisfy $s\cdot y=0$ (mod 2).
Each such~$y$ gives us a linear constraint (mod~2) on the bits of~$s$.
Repeating this subroutine $\Theta(n)$ times gives, with high probability, $n-1$ linearly independent $y^{(1)},\ldots,y^{(n-1)}$ all orthogonal to~$s$.
From these, $s$ can be calculated classically by Gaussian elimination.
Brassard and H\o yer~\cite{brassard&hoyer:simon} subsequently gave an \emph{exact} version of this algorithm, where each new~$y$ is produced by a modification of Simon's subroutine that uses $O(1)$ queries and is \emph{guaranteed} to be linearly independent from the previous ones (as long as such a linearly independent~$y$ exists).

This algorithm can be used to obtain a strong quantum speed-up for testing a specific property.

\begin{quote}
{\bf Simon property}:\\ 
$\Prop_{Simon}=\{x\in[N]^N : \exists s\in\01^n\backslash\{0^n\}\mbox{ such that }x_i=x_j\mbox{ if } j=i\oplus s\}$
\end{quote}
Note that, compared with Simon's original problem, the `if and only if' has been replaced with an `if.' Hence $x_i$ and $x_j$ can be equal even for distinct $i,j$ for which $j\neq i\oplus s$. However, also for such more general inputs, Simon's quantum subroutine (and the Brassard-H\o yer version thereof) only produces $y$ such that $s\cdot y=0$ (mod~2). The speed-up is as follows; for simplicity we state it for fixed $\eps=1/4$ rather than making the dependence on~$\eps$ explicit:

\begin{thm}[essentially Buhrman et al.~\cite{bfnr:qpropj}]
\label{thm:simon}
There is a quantum 1/4-property tester for the Simon property using $O(\log N)$ queries, while every classical 1/4-property tester needs $\Omega(\sqrt{N})$ queries.
\end{thm}

``Essentially'' in the attribution of the above theorem refers to the fact that Buhrman et al.~\cite{bfnr:qpropj} devised a property of \emph{binary} strings of length~$N$. In our presentation it will be more convenient to consider a property consisting of strings over alphabet~$[N]$. As remarked by Aaronson and Ambainis~\cite{aaronson11j}, Theorem~\ref{thm:simon} has an interesting consequence regarding the question of when we can hope to achieve exponential quantum speed-ups. In order to obtain a super-polynomial quantum speed-up for computing some function $f$ in the query complexity model, it is known that there has to be a \emph{promise} on the input, i.e., $f$ has to be a partial function~\cite{bbcmw:polynomialsj}. The Simon property indeed involves a promise on the input, namely that it is either in or far from $\Prop_{Simon}$; however, this promise is in some sense very weak, as the algorithm has to output the right answer on a $1-o(1)$ fraction of~$[N]^N$.

\begin{proof}
{\bf Quantum upper bound (sketch).}
We run the Brassard-H\o yer version of Simon's subroutine $n-1$ times. We then classically compute a non-zero string~$s$ that is orthogonal to all the $n-1$ strings~$y$ produced by these runs (there may be several such~$s$, in which case we just pick any). We then randomly choose $i\in[N]$, query $x_i$ and $x_{i\oplus s}$, and check if these two values are equal.  If $x\in\Prop_{Simon}$ then $s$ will have the property that $x_i=x_{i\oplus s}$ for all~$i$. On the other hand, if $x$ is 1/4-far from $\Prop_{Simon}$, then there exist at least $N/4$ $(i,i\oplus s)$-pairs such that $x_i\neq x_{i\oplus s}$ (for otherwise we could put $x$ into $\Prop_{Simon}$ by changing one value for each such pair, making fewer than $N/4$ changes in total). Hence in this case we reject with constant probability.  Testing a few different $(i,i\oplus s)$-pairs reduces the error probability to below~$1/3$.

{\bf Classical lower bound.}
Consider three distributions: ${\cal D}_1$ is uniform over $\Prop_{Simon}$, ${\cal D}_0$ is uniform over all $x\in [N]^N$ that are $1/4$-far from $\Prop_{Simon}$, and ${\cal U}$ is uniform over $[N]^N$.  We first show ${\cal D}_0$ and ${\cal U}$ are very close.

\begin{claim}
The total variation distance between ${\cal D}_0$ and ${\cal U}$ is $o(1)$.
\end{claim}

\begin{proof}
Let $S=\{y: y\mbox{ is not 1/4-far from }\Prop_{Simon}\}$ be the elements that are not in the support of ${\cal D}_0$.
We will upper bound the size of $S$.
Each element of $\Prop_{Simon}$ can be specified by giving an $s\in\01^n\backslash\{0^n\}$ and giving for each of the $N/2$ $(i,i\oplus s)$-pairs the value $x_i=x_{i\oplus s}$. Hence 
$$
|\Prop_{Simon}|\leq (N-1)N^{N/2}
$$
For each $x$, the number of $y$ that are $1/4$-close to $x$ is at most ${N\choose N/4}N^{N/4}$.
Hence the total number of elements $1/4$-close to $\Prop_{Simon}$ is
$$
|S|\leq (N-1)N^{N/2}{N\choose N/4}N^{N/4}=o(N^N).
$$
Since ${\cal U}$ is uniform over $[N]^N$ and ${\cal D}_0$ is uniform over $[N]^N\backslash S$,
the total variation distance between these two distributions is $O(|S|/N^N)=o(1)$.
\end{proof}

To finish the proof, below we slightly adapt the proof in~\cite{simon:power} to show that  
a $T$-query classical algorithm distinguishing distributions ${\cal D}_1$ and ${\cal U}$ has advantage of only $O(T^2/N)$ over random guessing.\footnote{The ``advantage'' of the algorithm is the difference between success and failure probabilities.}
Since ${\cal D}_0$ and ${\cal U}$ are $o(1)$-close, a $T$-query classical algorithm distinguishing distributions ${\cal D}_1$ and ${\cal D}_0$ has advantage $O(T^2/N)+o(1)$ over random guessing.
A classical tester for the Simon property can distinguish ${\cal D}_1$ and ${\cal D}_0$ with success probability at least~$2/3$, so it needs $T=\Omega(\sqrt{N})$ queries.
It remains to prove:

\begin{claim}
A $T$-query classical algorithm for distinguishing distributions ${\cal D}_1$ and ${\cal U}$ has advantage $O(T^2/N)$ over random guessing. 
\end{claim}

\begin{proof}
By the well-known Yao principle~\cite{yao:unified}, it suffices to prove the claim for an arbitrary \emph{deterministic} $T$-query algorithm. 
The proof will show that both under ${\cal D}_1$ and ${\cal U}$ the $T$ queries are likely to yield a uniformly random sequence of $T$ distinct values. Suppose the algorithm queries the indices $i_1,\ldots,i_T$ (this sequence may be adaptive, i.e., depend on $x$)
and gets outputs $x_{i_1},\ldots,x_{i_T}$. Call a sequence of queries $i_1,\ldots,i_T$
\emph{good} (for input~$x$) if it shows a collision, i.e., $x_{i_k}=x_{i_{\ell}}$ for some $k\neq\ell$. Call the sequence \emph{bad} (for~$x$) otherwise.
We will now show that the probability of a bad sequence is $O(T^2/N)$, both under input distribution~${\cal U}$ and under~${\cal D}_1$.

First, suppose the input~$x$ is distributed according to~${\cal U}$.
Then each output $x_{i_k}$ is uniformly distributed over $[N]$, independent of the other
entries of~$x$. The probability that $i_k$ and $i_\ell$ form a collision is exactly $1/N$, so
the expected number of collisions among the $T$ queries is ${T\choose 2}/N=O(T^2/N)$.
Hence, by Markov's inequality, the probability that $i_1,\ldots,i_T$ form a good sequence is $O(T^2/N)$.

Second, suppose the input $x$ is distributed according to~${\cal D}_1$.
Then there exists a nonzero $s\in\01^n$, unknown to the algorithm, such that $x_i=x_j$ whenever $j=i\oplus s$.
Initially, all such~$s$ are equally likely under~${\cal D}_1$
(the probability that there are two distinct such~$s$ for $x$ is negligibly small, and we will ignore this here).
If $i_1,\ldots,i_{k-1}$ is bad, then we have excluded ${k-1 \choose 2}$ of the $N-1$ possible values of $s$,
and all other values of $s$ are equally likely.
Let $i_k$ be the next query and $S=\{i_k\oplus i_j: j<k\}$.
This set $S$ has at most $k-1$ members, so the probability (under~${\cal D}_1$) that $S$ happens to contain the string~$s$ is at most $\frac{k-1}{N-1-{k-1\choose 2}}$.
If $S$ does not contain~$s$, then the only way to make the sequence good is if the uniformly random value~$x_{i_k}$ 
equals one of the $k-1$ earlier values, which has probability~$(k-1)/N$.
Hence the probability that the bad sequence $i_1,\ldots,i_{k-1}$ remains bad, after query $i_k$ is made, is very close to~1. More precisely:
\begin{eqnarray*}
\Pr[i_1,\ldots,i_T\mbox{ is bad}] & = &
\prod_{k=2}^T \Pr[i_1,\ldots,i_k\mbox{ is bad}: i_1,\ldots,i_{k-1}\mbox{ is bad}]\\
& \geq & \prod_{k=2}^T \left(1-\frac{k-1}{N-1-{k-1\choose 2}}-\frac{k-1}{N}\right)\\
& \geq & 1-\sum_{k=2}^T \left(\frac{k-1}{N-1-{k-1\choose 2}}+\frac{k-1}{N}\right).
\end{eqnarray*}
Here we used the fact that $(1-a)(1-b)\geq 1-(a+b)$ if $a,b\geq 0$.
The latter sum over~$k$ is $O(T^2/N)$, so the probability (under~${\cal D}_1$) that $i_1,\ldots,i_T$ form a good sequence is~$O(T^2/N)$.

In both cases (${\cal U}$ and ${\cal D}_1$), conditioned on seeing a bad sequence, the sequence of outputs is a uniformly random sequence of $T$ distinct values. Accordingly, the advantage (over random guessing) of the algorithm trying to distinguish these two distributions is upper bounded by the probability of seeing a good sequence, which is $O(T^2/N)$ in both cases.
\end{proof}
\end{proof}

\subsubsection{Using Shor's algorithm}

Probably the most famous quantum algorithm to date is Shor's polynomial-time algorithm for factoring integers~\cite{shor:factoring}.  
Its quantum core is an algorithm that can find the \emph{period} of a periodic sequence.
Chakraborty et al.~\cite{cfmw:qproptesting} used this to show that \emph{testing} periodicity exhibits a constant-versus-polynomial
quantum-classical separation.  Note that the Bernstein-Vazirani property (Section~\ref{sec:bvtester}) exhibits a constant-versus-logarithmic separation,
while the Simon property (Section~\ref{sec:simontester}) exhibits a logarithmic-versus-polynomial separation. 
Periodicity-testing thus exhibits a separation that is in some ways stronger than either of those.

\begin{quote}
{\bf Periodicity}: let $p$ be a prime number and $m$ an integer such that $m\geq p$.  A string $x\in[m]^N$ is \emph{1-1-$p$-periodic} if it satisfies that $x_i=x_j$ if and only if $i=j$ mod $p$
(equivalently, the elements in the sequence $x_0,\ldots,x_{p-1}$ are all unique, and after that the sequence repeats itself).
For integers $q$ and $r$ such that $q\leq r\leq\sqrt{N/2}$, define the property
$\Prop^{q,r}_{period}=\{x\in[m]^N: x \mbox{ is 1-1-$p$-periodic for some }p\in\{q,\ldots,r\}\}$
\end{quote} 

Note that \emph{for a given $p$} it is easy to test whether $x$ is $p$-periodic or far from it:
choose an $i\in[N]$ uniformly at random, and test whether $x_{i}=x_{i+kp}$ for a random positive integer $k$.
If $x$ is $p$-periodic then these values will be the same, but if $x$ is far from $p$-periodic then we will detect this with good probability.
However, $r-q+1$ different values of $p$ are possible in $\Prop^{q,r}_{period}$.  Below we will set $q=r/2$ so $r/2+1$ different values for the period are possible. This makes the property hard to test for classical testers. On the other hand, in the \emph{quantum} case the property can be tested efficiently.

\begin{thm}[Chakraborty et al.~\cite{cfmw:qproptesting}]\label{thperiodicity}
For every even integer $r\in[2,\sqrt{N})$ and constant distance~$\eps$, there is a quantum property tester for $\Prop^{r/2,r}_{period}$ using $O(1)$ queries, while every classical property tester for $\Prop^{r/2,r}_{period}$ makes $\Omega(\sqrt{r/\log r\log N})$ queries. 
In particular, for $r=\sqrt{N}$ testing can be done with $O(1)$ quantum queries but requires $\Omega(N^{1/4}/\log N)$ classical queries.
\end{thm} 

The quantum upper bound is obtained by a small modification of Shor's algorithm: 
use Shor to find the period~$p$ of input~$x$ (if there is such a period) and then test this purported period with another $O(1)$ queries.%
\footnote{These ingredients are already present in work of Hales and Hallgren~\cite{hales&hallgren:improvedfourier}, and in Hales's PhD thesis~\cite{hales:phd}. However, their results are not stated in the context of property testing, and no classical lower bounds are proved there.}
The classical lower bound is based on modifying proofs from Lachish and Newman~\cite{lachish&newman:periodicity},
who showed classical testing lower bounds for more general (and hence harder) periodicity-testing problems.

This quantum-classical separation is of the form $O(1)$ quantum queries vs $N^{\Omega(1)}$ classical queries, for a problem over a polynomial-sized alphabet (so each ``entry'' of the input takes only $O(\log N)$ bits).  How large can we make this separation? This was already asked by Buhrman et al.~\cite{bfnr:qpropj} in the following way:

\begin{question}
Is there a property of strings of length $N$ (over a moderately-sized alphabet) 
that can be tested with $O(1)$ quantum queries but needs $\Omega(N)$ classical queries?
\end{question}

A very recent result of Aaronson and Ambainis~\cite{aaronson&ambainis:forrelation} is relevant here: they showed that if a (total or partial) function on $x\in\01^N$ can be computed with bounded error probability using $k$ quantum queries, then the same function can be computed by a classical randomized algorithm using $O(N^{1-1/2k})$ queries. They also show that for $k=1$ this upper bound is tight up to a logarithmic factor, for a testing problem called ``Forrelation.'' In that problem, roughly, the input consists of two Boolean functions $f$ and~$g$, each on $\ell$-bit inputs so the total input length is $N=2\cdot 2^\ell$ bits, such that $g$ is either strongly or weakly correlated with the Fourier transform of~$f$ (i.e., $g(x)=\mbox{sign}(\widehat{f}(x))$ either for most $x$ or for roughly half of the~$x$). They show that this problem can be tested with one quantum query, whereas classical testers need $\Omega(\sqrt{N}/\log N)$ queries.\footnote{The lower bound improves an earlier $N^{1/4}$ bound of Aaronson~\cite{aaronson:bqpph}, which constituted the first $O(1)$ vs $N^{\Omega(1)}$ separation for quantum vs classical propery testing.}

Hence for binary alphabets the answer to the above question is negative: everything that can be tested with $k=O(1)$ quantum queries can be tested with $O(N^{1-1/2k})=o(N)$ classical queries.  This classical upper bound can be extended to small alphabets, but the question remains open for instance when the alphabet size is~$N$.

\subsubsection{Using quantum counting}\label{sec:usinggrover}

Grover's quantum search algorithm~\cite{grover:search} can be used to find the index~$i$ of a 1-bit in $x\in\01^N$ (i.e., $x_i=1$) with high probability, using $O(\sqrt{N})$ queries.  We will not describe the algorithm here, but just note that it can be modified to also \emph{estimate}, for given $S\subseteq[m]$, the number of occurrences of elements from $S$ in a string~$x\in[m]^N$, using a number of queries that is much less than would be needed classically.
More precisely, we have the following ``quantum approximate counting'' lemma, which follows from the work of Brassard et al.~\cite[Theorem~13]{bhmt:countingj}:

\begin{lem}\label{lem:quantumcouting}
There exists a constant $C$ such that for every set $S\subseteq [m]$ and every positive integer~$T$, there is a quantum algorithm that makes $T$ queries to input~$x\in[m]^N$ and, with probability at least~$2/3$, outputs an estimate $p'$ to $p = |\{i: x_i\in S\}|/N$ such that $|p'-p|\leq C(\sqrt{p}/T+1/T^2)$.
\end{lem} 

We now describe an application of quantum counting to property testing, namely to testing whether two probability distributions are equal or $\eps$-far from each other in total variation distance. 

\begin{quote}
{\bf Equal distributions property}: $\Prop_{distribution}=\{(p,p) : p\mbox{ is a distribution on }[m]\}$.
\end{quote}

Our distance measure on the set of pairs of distributions will be the sum of the total variation distances: $d((p,q),(p',q'))=\norm{p-p'}_{tvd} + \norm{q-q'}_{tvd}$, where the \emph{total variation distance} between two distributions is $\norm{p-p'}_{tvd}:=\frac{1}{2}\sum_j|p(j)-p'(j)|$).  Note that a pair of distributions $(p,q)$ will be $\eps$-far from $\Prop_{distribution}$ if and only if $\norm{p-q}_{tvd}\geq\eps$.

There are different ways in which the distributions could be ``given'' to the tester, but in this section each distribution will be given as an input $x\in[m]^N$.  This naturally induces a probability distribution ${\cal D}_x$ on $[m]$ according to the relative frequencies of the different elements:
$$
{\cal D}_x(j)=\frac{|\{i:x_i=j\}|}{N}.
$$
We can obtain a sample according to~${\cal D}_x$ by just querying $x$ on a uniformly random index~$i$. Assuming the distribution is given in this way is quite natural in the setting of property testing,  where our input is usually a very long string~$x$, much too long to inspect each of its elements. Note that ${\cal D}_x$ does not change if we permute the elements of~$x$; it just depends on the relative frequencies. Also note that Lemma~\ref{lem:quantumcouting} can be used to estimate the probability of $S\subseteq[m]$ under ${\cal D}_x$.

Suppose we are given two distributions ${\cal D}_x$ and ${\cal D}_y$ on $[m]$ (the distributions are given in the form of two inputs $x,y\in[m]^N$),
and we want to test whether these two distributions are equal or $\eps$-far in total variation distance.
Batu et al.~\cite{batu13} exhibited classical testers for this using $O((m/\eps)^{2/3}\log m)$ queries\footnote{All these classical bounds are stated in terms of number of samples rather than number of queries, 
but it is not hard to see that these two complexity measures are equivalent here.}, 
and Valiant \cite{valiant11} proved an almost matching lower bound of $\Omega(m^{2/3})$ for constant $\eps$. These bounds have both recently been improved by Chan et al.~\cite{chan13} to $\Theta(m^{2/3}/\eps^{4/3})$, which is tight for all $\eps\geq m^{-1/4}$.
Bravyi et al.~\cite{BHH10j} showed that quantum testers can do better in terms of their dependence on $m$:

\begin{thm}[Bravyi et al.~\cite{BHH10j}]
There is a quantum tester to test if two given distributions on $[m]$ are equal or $\eps$-far using $O(\sqrt{m}/\eps^8)$ queries.
\end{thm}

The dependence on $\eps$ can probably be improved with more technical effort.

\begin{proof}[Proof sketch]
Bravyi et al.~\cite{BHH10j} actually showed something stronger, namely that the total variation distance between two distributions can be estimated up to small additive error~$\eps$ using $O(\sqrt{m}/\eps^8)$ quantum queries; this clearly suffices for testing.  We sketch their idea here.
Consider the following random process:
\begin{quote}
Sample $j\in[m]$ according to ${\cal D}=\frac{1}{2}({\cal D}_x+{\cal D}_y)$.\\
Output $\displaystyle\frac{|{\cal D}_x(j)-{\cal D}_y(j)|}{{\cal D}_x(j)+{\cal D}_y(j)}$
\end{quote}
It is easy to see that the expected value of the output of this process is exactly the total variation distance between ${\cal D}_x$ and ${\cal D}_y$,
so it suffices to approximate that expected value. We sample $j$ according to ${\cal D}$ (which costs just one query), use the 
quantum algorithm of Lemma~\ref{lem:quantumcouting} with $S=\{j\}$ and $T=O(\sqrt{m}/\eps^6)$ queries to approximate both ${\cal D}_x(j)$ and ${\cal D}_y(j)$, and output the absolute difference between these two approximations divided by their sum. 
Bravyi et al.~\cite{BHH10j} show that repeating this $O(1/\eps^2)$ times and taking the average gives, with probability at least~$2/3$, an $\eps$-approximation of the expected value ${{\cal D}_x-{\cal D}_y}_{tvd}$ of the above random process.
\end{proof}

A second problem is where we fix one of the two distributions, say to the uniform distribution on~$[m]$ 
(assume $m$ divides $N$ so we can properly ``fit'' this distribution in $x\in[m]^N$).
Goldreich and Ron~\cite{gr} showed a classical testing lower bound of $\Omega(\sqrt{m})$ queries for this,
and Batu et al.~\cite{BFF+01} proved a nearly tight upper bound of $\widetilde{O}(\sqrt{m})$ queries.
Bravyi et al.~\cite{BHH10j}, and independently also Chakraborty et al.~\cite{cfmw:qproptesting}, 
showed that testing can be done more efficiently in the quantum case:

\begin{thm}[Bravyi et al.~\cite{BHH10j}, Chakraborty et al.~\cite{cfmw:qproptesting}]\label{th:qtestuniform}
There is a quantum tester to test if a given distribution on $[m]$ equals or is $\eps$-far from the uniform distribution on $[m]$,
using $O(m^{1/3}/\eps^2)$ quantum queries.
\end{thm}

\begin{proof}[Proof sketch]
Pick a set $R\subseteq[N]$ of $r=m^{1/3}$ indices uniformly at random, and query its elements. 
If ${\cal D}_x$ is uniform then it is very likely that all values $\{x_i\}_{i\in R}$ are distinct, 
so if there is some collision then we can reject immediately.
Otherwise, let $S=\{x_i: i\in R\}$ be the $r$ distinct results, and define $p=|\{i: x_i\in S\}|/N$.
If ${\cal D}_x$ is uniform then $p=r/m=1/m^{2/3}$, but \cite[Lemma~13]{cfmw:qproptesting} 
shows that if ${\cal D}_x$ is $\eps$-far from uniform then $p$ will be noticeably higher: 
there is a constant $c>0$ such that with high probability $p\geq (1+c\eps^2)r/m$.

Now we use the quantum algorithm of Lemma~\ref{lem:quantumcouting} with $T=4Cm^{1/3}/c\eps^2$ queries 
to obtain (with high probability) an estimate $p'$ of $p$ within additive error $|p'-p|\leq  C(\sqrt{p}/T + 1/T^2)$.
We accept if $p'\leq (1+c\eps^2/2)r/m$, and reject otherwise.
If $p=r/m=1/m^{2/3}$ then the additive error is at most 
$C(c\eps^2/4Cm^{2/3} + c^2\eps^4/16C^2m^{2/3})\leq \frac{c\eps^2}{2} \cdot \frac{r}{m}$, so then we will accept correctly.
If $p\geq (1+c\eps^2)r/m$ then it is easy to show that $p'\geq(1+c\eps^2/2)r/m$, so then we will reject correctly.
\end{proof}

Both Bravyi et al.~\cite{BHH10j} and Chakraborty et al.~\cite{cfmw:qproptesting} showed that $\Omega(m^{1/3})$ quantum queries are also necessary, so the above result is essentially tight; the lower bound follows from a reduction from the collision problem~\cite{aaronson&shi:collision}.
Bravyi et al.~\cite{BHH10j} also exhibited a quantum tester for whether two distributions are equal or of disjoint support (i.e., orthogonal), using $O(m^{1/3})$ quantum queries.
Chakraborty et al.~\cite{cfmw:qproptesting} extended Theorem~\ref{th:qtestuniform} to testing equality to \emph{any} fixed distribution (not just the uniform one), at the expense of a polylog factor in the number of queries. They in turn used equality-testing to obtain better quantum testers for graph isomorphism.


\subsubsection{Using Ambainis's algorithm}\label{sec:usingambainis}

Ambainis's element distinctness algorithm~\cite{ambainis:edj} acts on an input $x\in[m]^N$, and finds
(with high probability) a pair of distinct indices such that $x_i=x_j$ if such a pair exist, and reports ``no collision'' otherwise. 
It uses $O(N^{2/3})$ queries, which is optimal~\cite{aaronson&shi:collision}.
This algorithm spawned a large class of algorithms based on \emph{quantum walks} (see~\cite{santha:qrwsurvey} for a survey).

Ambainis et al.~\cite{acl:testing} use the element distinctness algorithm to give better quantum testers for certain \emph{graph properties}.
Graph properties have some amount of symmetry: they are invariant under relabelling of vertices. 
Problems with ``too much'' symmetry are known not to admit exponential quantum speed-up in the query complexity model~\cite{aaronson11j}, and the symmetry inherent to graph properties makes them an interesting test case for the question of how symmetric the problems can be for which we do obtain a significant quantum advantage.
Ambainis et al.~\cite{acl:testing} use the element distinctness algorithm
to give $\widetilde{O}(N^{1/3})$-query quantum testers for the properties of bipartiteness and being an expander in bounded-degree graphs. 
It is known that for classical testers, $\widetilde{\Theta}(\sqrt{N})$ queries are necessary and sufficient to test these properties~\cite{gr,gr:boundeddegree}.
Together with the graph isomorphism tester mentioned briefly at the end of Section~\ref{sec:usinggrover},
these are the only quantum results we are aware of for testing graph properties.
In contrast, graph properties have been one of the main areas of focus in classical property testing.

Let us describe the results of~\cite{acl:testing} a bit more precisely.
The object to be tested is an $N$-vertex graph~$G$ of degree~$d$, so each vertex has at most~$d$ neighbors.
We think of $d$ as a constant and will absorb the dependence of the bounds on~$d$ into the constant factor.  The input is given as an \emph{adjacency list}.
Formally, it corresponds to an $x\in([N]\cup\{*\})^{N\times d}$.  The entries of $x$ are indexed by a pair of a vertex $v\in[N]$ and a number $i\in[d]$,
and $x_{v,i}$ is the $i$th neighbor of vertex~$v$; $x_{v,i}=*$ in case $v$ has fewer than~$i$ neighbors. The \emph{distance} between two graphs given as adjacency lists is defined to be the minimal number of edges one most change in the first graph to obtain the second.

A graph is {\bf Bipartite} if its set of vertices can be partitioned into two disconnected sets,
and is an {\bf Expander} if there is a constant $c>0$ such that every set $S\subseteq[N]$ of at most $N/2$ vertices has at least $c|S|$ neighbors outside of~$S$.%
\footnote{Equivalently, if there is a constant gap between the first and second eigenvalue of $G$'s normalized adjacency matrix.  A crucial property of an expander is that the endpoint of a short ($O(\log N)$-step) random walk starting from any vertex is close to uniformly distributed over $[N]$.  We refer to~\cite{hlw:expander} for much more background on expander graphs and their many applications.}

\begin{thm}[Ambainis et al.~\cite{acl:testing}]
There exist quantum testers for {\bf Bipartite} and {\bf Expander} using $\widetilde{O}(N^{1/3})$ queries.
\end{thm}

\begin{proof}[Proof sketch]
At a high level, the optimal \emph{classical} testers for both properties look for collisions in a set of roughly $\sqrt{N}$ elements.  Using Ambainis's algorithm, this can be done in roughly $N^{1/3}$ queries. Let us see how this works for the two properties.

A bipartite graph has no odd cycles. In contrast, for a graph that is far from bipartite one can show the following.
Among roughly $\sqrt{N}$ short ($O(\log N)$-step) random walks from the same starting vertex~$v$, 
there is likely to be a pair that ``collides'' in the sense that one walk reaches a vertex~$w$ after an even number of steps and the other reaches the same vertex~$w$ after an odd number of steps. These two paths between $v$ and $w$ now form an odd cycle. 
Hence, fixing the randomness of such a classical tester, it suffices to detect such collisions in a string $x\in[m]^{c\sqrt{N}}$,
for some constant $c>0$, where the alphabet $[m]$ corresponds to short walks in the graph. 
A variant of Ambainis's algorithm can detect this in $O((c\sqrt{N})^{2/3})=O(N^{1/3})$ queries to~$x$. 
Each query to~$x$ corresponds to an $O(\log N)$-walk through the graph, 
so we use $O(N^{1/3}\log N)$ queries to the input graph in total.

In the case of expanders, a short random walk will quickly converge to the uniform distribution.
In contrast, for a graph that is far from any expander, such a walk will typically not be very close to uniform.  If we sample $k$ times from the uniform distribution
over some $s$-element set, the expected number of collisions is ${k\choose 2}/s$.  
In particular, for $k\approx\sqrt{2s}$ we expect to see one collision.
In contrast, $k$~samples from a non-uniform distribution give a higher expected number of collisions.
Hence if we do $c\sqrt{N}$ short random walks, for some constant~$c$, 
then the expected number of collisions among the $c\sqrt{N}$ endpoints is likely to be significantly smaller
for an expander than for a graph that is far from every expander.
Again we use a variation of Ambainis's algorithm, this time to approximately \emph{count} the number of collisions 
in an input $x\in[m]^{c\sqrt{N}}$, consisting of the endpoints of the $c\sqrt{N}$ random walks.
If this number is too high, we reject. This uses $\widetilde{O}(N^{1/3})$ queries to the graph. 
The technical details are non-trivial, but we will skip them here.
\end{proof}

Ambainis et al.\ also proved an $\widetilde{\Omega}(N^{1/4})$ quantum lower bound for testing expanders, using the polynomial lower bound method (see Section~\ref{sec:polynomialmethod}).
They were not able to show $N^{\Omega(1)}$ lower bounds for testing bipartiteness. This all leaves the following very interesting question open:

\begin{question}
Is there any graph property which admits an exponential quantum speed-up?
\end{question}

\subsubsection{Quantum speed-ups for testing group-theoretic properties}
Finally, a number of authors have considered quantum testers for properties of groups; we list these here without explaining them in detail.
\begin{itemize}
\item
Friedl et al.~\cite{fmsp:testhidden} give efficient quantum testers for the property of periodic functions on groups
(the testers are even \emph{time-efficient} for Abelian groups), as well as a few other group-theoretic properties. The testers are based on the use of the (Abelian and non-Abelian) quantum Fourier transform.
\item 
Friedl et al.~\cite{friedl05} exhibit an efficient ($\poly(\log N,1/\eps)$-query) \emph{classical} tester for the property of $N \times N$ multiplication tables corresponding to $N$-element Abelian groups, which is based on ``dequantizing'' a quantum tester. The distance used is the so-called ``edit distance.''
\item
Inui and Le Gall~\cite{inui&legall:testingj}, extending~\cite{friedl05}, exhibit an efficient ($\poly(\log N,1/\eps)$-query) quantum tester for the property of $N \times N$ multiplication tables corresponding to $N$-element \emph{solvable} groups. In this case, no efficient classical tester is known.
\item Le Gall and Yoshida~\cite{legall11j} give classical lower bounds on various group testing problems, which in particular demonstrate an exponential separation between the classical and quantum complexities of testing whether the input is an Abelian group generated by $k$ elements (where $k$ is fixed).
\end{itemize}

\subsection{Lower bounds}\label{sec:clasproplowerbounds}

Here we describe the main methods for obtaining \emph{lower bounds} on the number of queries that quantum property testers need. Most such lower bounds have been obtained using the so-called \emph{polynomial method}, but in principle the stronger \emph{adversary method} can give tight bounds for any property. At the end of this section we also describe an elegant approach for deriving \emph{classical} testing lower bounds from communication complexity, leaving its generalization to lower bounds on quantum testers as an open question.

\subsubsection{The polynomial method}
\label{sec:polynomialmethod}

The first lower bounds for quantum property testing were proven by Buhrman et al.~\cite{bfnr:qpropj}.
They were based on the polynomial method~\cite{bbcmw:polynomialsj}, which we now briefly explain.
The key property is:
\begin{quote}
The acceptance probability of a $T$-query quantum algorithm on input $x\in\01^N$ can be written as an $N$-variate multilinear polynomial $p(x)$ of degree~$\leq 2T$.
\end{quote}
This property can be generalized to non-Boolean inputs~$x$, but for simplicity we will assume $x\in\01^N$ in our presentation. 

Note that if we have a $T$-query quantum tester for some property $\Prop\subseteq\01^N$, then its acceptance probability $p$ is a degree-$2T$ polynomial~$p$
such that $p(x)\in[2/3,1]$ if $x\in \Prop$; $p(x)\in[0,1/3]$ if $x$ is far from~$\Prop$; and $p(x)\in[0,1]$ for all other~$x$.
The polynomial method derives lower bounds on the query complexity~$T$ from lower bounds on the minimal degree of such polynomials.

Our first application of this method is a result which is essentially from~\cite{bfnr:qpropj}. Informally, the result says the following: if we have a property $\Prop$ such that a (not necessarily uniform) random $x\in\Prop$ is indistinguishable from a random $N$-bit string if we only look at up to~$k$ bits, then the quantum query complexity of testing $\Prop$ is $\Omega(k)$.

\begin{thm}[Buhrman et al.~\cite{bfnr:qpropj}]
\label{thm:uniform}
Let $\Prop \subseteq \01^N$ be a property such that the number of elements $\eps$-close to~$\Prop$ is $< 2^{N-1}$. Let $\mathcal{D} = (p_z)$ be a distribution on $\01^N$ such that $p_z = 0$ for $z \notin \Prop$, and $\E_{\mathcal{D}} [z_{i_1} \dots z_{i_\ell}] = 2^{-\ell}$ for all choices of $\ell \le k$ distinct indices $i_1,\dots,i_\ell\in[N]$. Then every quantum $\eps$-property tester for $\Prop$ must make at least $(k+1)/2$ queries.
\end{thm}

\begin{proof}
Suppose there is a quantum algorithm which tests $\Prop$ using $T$ queries, where $T < (k+1)/2$. Then by the polynomial method, its acceptance probability is given by a polynomial $p(z)$ of degree at most $2T \le k$. Intuitively, the reason the theorem holds is that such a degree-$k$ polynomial cannot be correlated with a $k$-wise independent distribution. To make this precise, assume towards a contradiction that the algorithm has success probability at least~$2/3$ on every input~$z$ that is in or $\eps$-far from~$\Prop$. Then
\[ 
\E_{z \sim \mathcal{D}}[p(z)] \ge \frac{2}{3}
\]
and, letting $\Prop_{close}$ be the set of~$z$ that are $\eps$-close to $\Prop$, and $\mathcal{U}$ the uniform distribution over $\01^N$, we have 
\[ 
\E_{z \sim \mathcal{U}}[p(z)] \le \frac{|\Prop_{close}|}{2^N} + \frac{1}{3} \left(1-\frac{|\Prop_{close}|}{2^N}\right) < \frac{2}{3}, 
\]
Write $p(z) = \sum_{S \subseteq [N]} \alpha_S m_S(z)$, where $m_S$ is the monomial $\prod_{i \in S} z_i$. We have
\[ 
\E_{z \sim \mathcal{D}}[p(z)] = \sum_{S \subseteq [N]} \alpha_S \E_{z \sim \mathcal{D}}[m_S(z)] = \sum_{S \subseteq [N]} \alpha_S 2^{-|S|} = \sum_{S \subseteq [N]} \alpha_S \E_{z \sim \mathcal{U}}[m_S(z)] = \E_{z \sim \mathcal{U}} [p(z)]. 
\]
We have obtained a contradiction, which completes the proof.
\end{proof}

A variant of Theorem \ref{thm:uniform}, which generalizes the claim to an underlying set $[m]^N$ ($m > 2$) but does not consider the property testing promise, was independently shown by Kane and Kutin~\cite{kane09}. It is apparently quite hard to satisfy the uniformity constraint of Theorem \ref{thm:uniform}; however, it can sometimes be achieved. For example, consider any property which can be expressed as membership of a linear code $\mathcal{C} \subseteq \F_2^N$. Such a linear code is described as the set $\{Mz:z \in \{0,1\}^\ell\}$ for some $N \times \ell$ matrix $M$. A code has \emph{dual distance}~$d$ if every codeword $c'$ in the dual code $\mathcal{C}^{\perp} := \{ z: z\cdot c = 0, \forall~c \in \mathcal{C} \}$ satisfies $|c'| \ge d$.
As Alon et al.~\cite{alon05} observe, it is well-known in coding theory that if $\mathcal{C}$ has dual distance~$d$, then any subset of at most $d-1$ of the bits of $\mathcal{C}$ are uniformly distributed. As the (easy) proof does not seem easy to find in the recent literature, we include it here.

\begin{thm}\cite[Chapter 1, Theorem 10]{macwilliams83}
Let $\mathcal{C} \subseteq \{0,1\}^N$ be a code with dual distance $d$. Then every $k < d$ bits of codewords in $\mathcal{C}$ are uniformly distributed.
\end{thm}

\begin{proof}
Dual distance~$d$ implies that every set of $k \le d-1$ rows in the matrix~$M$ are linearly independent (otherwise such a linear combination would imply the existence of a Hamming weight $k < d$ vector~$z$ such that $Mz = 0^N$). So for each submatrix $M'$ formed by choosing $k$ rows from~$M$, all the rows of $M'$ are linearly independent, hence the output $M'z$ is uniformly distributed over $\{0,1\}^k$.
\end{proof}

Thus, if $\mathcal{C}$ has dual distance $d$, taking $\mathcal{D}$ to be uniform over $\mathcal{C}$ in Theorem~\ref{thm:uniform} gives an $\Omega(d)$ lower bound on the quantum query complexity of testing membership in $\mathcal{C}$. A natural example for which this result gives a tight lower bound is the Reed-Muller code $R(d,\ell)$. Each codeword of this code is a binary string of length $N=2^\ell$ obtained by evaluating a function $f:\{0,1\}^\ell \rightarrow \{0,1\}$, which can be written as a degree-$d$ polynomial in $\ell$ variables over $\F_2$, at every element $z \in \{0,1\}^\ell$. $R(d,\ell)$ is known to have dual distance $2^{d+1}$~\cite[Chapter 13]{macwilliams83}, so Theorem~\ref{thm:uniform} implies that any quantum algorithm testing the set of degree-$d$ polynomials in $\ell$ variables over $\F_2$ must make $\Omega(2^d)$ queries. In particular, this means that quantum algorithms obtain no asymptotic speed-up, in terms of their dependence on $d$, over the best classical algorithm for testing this property~\cite{alon05}. One can generalize this whole argument to derive quantum lower bounds for testing membership of various interesting properties corresponding to codes over $\F_q$, for $q > 2$; we omit the details. One example  of this approach outside of the property-testing setting is~\cite{kane09}, which proves bounds on the complexity of quantum interpolation of polynomials. Here the relevant code is the Reed-Solomon code.

Buhrman et al.\ also applied the polynomial method to show, by a counting argument, that \emph{most} properties do not have an efficient quantum property tester. Informally speaking, there are too many properties, and too few low-degree polynomials.

\begin{thm}[Buhrman et al.~\cite{bfnr:qpropj}]
Let $\Prop \subset \{0,1\}^N$ be chosen at random subject to $|\Prop| = 2^{N/20}$, and fix $\eps$ to be a small constant. Then, except with probability exponentially small in $N$, any quantum $\eps$-property tester for $\Prop$ must make $\Omega(N)$ queries.
\end{thm}

A more involved application of the polynomial method is the tight $\Omega(\log N)$ lower bound that Koiran et al.~\cite{knp:simon} proved
for the quantum query complexity of Simon's problem.  With a bit of work, their proof also works to show that the property tester presented in Section~\ref{sec:simontester} is essentially optimal.

Another highly non-trivial application of the polynomial method is the $\widetilde{\Omega}(N^{1/4})$ lower bound of Ambainis et al.~\cite{acl:testing}
for testing the property of a bounded-degree graph being an {\bf Expander} (see Section~\ref{sec:usingambainis}).
Their lower bound is inspired by the one for the collision problem~\cite{aaronson&shi:collision}, and at a high level works as follows.
They give an input distribution ${\cal D}_\ell$ over $N$-vertex $d$-regular graphs with $\ell$ components, obtained from $M$-vertex graphs that consist of $\ell$ equal-sized random parts ($M$ is slightly bigger than~$N$ and divisible by~$\ell$; its role in the proof is rather technical). They then show that the 
acceptance probability of a $T$-query quantum tester can be written as an $O(T\log T)$-degree bivariate polynomial $p(\ell,M)$ in $\ell$ and~$M$.
A random graph of $\ell=1$ components is very likely to be an expander, so $p(1,M)\approx 1$;
on the other hand, every graph with $\ell>1$ components will be far from an expander, so $p(\ell,M)\approx 0$ for integers $\ell>1$.
They then use results about polynomial approximation to show that such polynomials need degree $\Omega(N^{1/4})$.

\subsubsection{The adversary method}\label{ssecadversary}

The two main lower bound methods that we know for quantum query complexity are the above polynomial method, and the so-called \emph{adversary method}, introduced by Ambainis~\cite{ambainis:lowerboundsj}. For a long time this adversary method faced the so-called ``property testing barrier''~\cite{hls:madv}: for every $N$-bit partial Boolean function where all 0-inputs are at Hamming distance $\Omega(N)$ from all 1-inputs, the method can prove only a \emph{constant} lower bound on the query complexity. Note that all testing problems for classical properties with respect to Hamming distance fall in this regime, since the 0-inputs are required to be far from all 1-inputs (i.e., elements of the property).

However, H\o yer et al.~\cite{hls:madv} generalized Ambainis's method to something substantially stronger,
which can prove optimal bounds for quantum property testing.
We now describe their ``negative weights'' adversary bound.
Let $F:D\rightarrow\01$, with $D\subseteq[m]^N$, be a Boolean function.
An \emph{adversary matrix} $\Gamma$ for~$F$ is a real-valued matrix whose rows and columns are indexed by all $x\in D$,
satisfying that $\Gamma_{xy}=0$ whenever $f(x)=f(y)$.
Let $\Delta_j$ be the Boolean matrix whose rows and columns are indexed by all $x\in D$, such that $\Delta_j[x,y]=1$ if $x_j\neq y_j$, and $\Delta_j[x,y]=0$ otherwise. The (negative-weights) adversary bound for~$F$ is given by the following expression:
\begin{equation*}
\ADV^{\pm}(F)=\max_\Gamma\frac{\norm{\Gamma}}{\max_{j\in[N]}\norm{\Gamma\circ \Delta_j}},
\end{equation*}
where $\Gamma$ ranges over all adversary matrices for~$F$, `$\circ$' denotes entry-wise product of two matrices, and `$\norm{\cdot}$' denotes operator norm (largest singular value) of the matrix.\footnote{Crucially, the adversary matrix $\Gamma$ may have negative entries. Restricting it to non-negative entries gives one of the many equivalent formulations of Ambainis's earlier adversary method~\cite{spalek&szegedy:adversaryj}.}

H\o yer et al.~\cite{hls:madv} showed that this quantity is indeed a valid lower bound:
every quantum algorithm that computes $F$ with error probability $\leq\eps$ needs to make 
at least $\frac{1}{2}(1-\sqrt{\eps(1-\eps)})\ADV^{\pm}(F)$ queries.
Subsequently, Reichardt et al.~\cite{reichardt:tight,lmrss:stateconv} showed this lower bound is actually essentially tight: 
for every Boolean function~$F$ there is a quantum algorithm computing it with error $\leq 1/3$ using $O(\ADV^{\pm}(F))$ queries.
Since property testing is just a special case of this (the 1-inputs of~$F$ are all $x\in \Prop$, and the 0-inputs are all $x$ that are far from~$\Prop$),
in principle the adversary bound characterizes the quantum complexity of testing classical properties.
However, in practice it is often hard to actually calculate the value of $\ADV^{\pm}(F)$, and we are not aware of
good quantum property testing lower bounds that have been obtained using this method. 


\subsubsection{A communication complexity method?}
Recently, a very elegant lower bound method for classical property testing was developed by Blais et al.~\cite{bbmj}, based on \emph{communication complexity}.  In the basic setting of communication complexity~\cite{yao:distributive,kushilevitz&nisan:cc}, two parties (Alice with input~$x$ and Bob with input~$y$) try to compute a function $F(x,y)$ that depends on both of their inputs, using as little communication as possible. This is a very well-studied model with many applications, particularly for deriving lower bounds in other areas, such as circuits, data structures, streaming algorithms, and many others (for which see~\cite{kushilevitz&nisan:cc}).

Blais et al.~\cite{bbmj} showed for the first time how to derive property testing lower bounds from communication complexity.
Their idea is to convert a $T$-query property tester for some property~$\Prop$ into 
a protocol for some related communication problem~$F$, by showing that 1-inputs
$(x,y)$ for~$F$ somehow correspond to elements of~$\Prop$, 
while 0-inputs $(x,y)$ for $F$ correspond to elements that are far from $\Prop$.
The more efficient the tester, the less communication the protocol needs.
Communication complexity lower bounds for~$F$ then imply lower bounds on the complexity~$T$ of the tester.

This is best explained by means of an example.  A \emph{$k$-linear function} $f:\01^n\rightarrow\01$ is a linear function that depends on exactly $k$ of its input bits: there exists a weight-$k$ $x\in\01^n$ such that $f(i)=i\cdot x$ mod~2 for all $i\in\01^n$. 
Let $\Prop$ be the set of $k$-linear functions, and assume~$k$ is even. 
Suppose we have a randomized $T$-query tester~$\cal T$ for~$\Prop$.
We will show how such a tester induces an efficient communication protocol for the communication complexity problem
of deciding whether weight-$k/2$ strings $x\in\01^n$ and $y\in\01^n$ are \emph{disjoint} or not (i.e., whether $x\wedge y=0^n$).
Alice, who received input $x$, forms the function $f(i) = i\cdot x$ and Bob forms the function $g(i) = i\cdot y$. 
Consider the function $h(i) = i\cdot (x\oplus y)$. Since $|x \oplus y| = |x| + |y| - 2 |x \wedge y|$ and $|x|+|y|=k$, 
the function $h$ is a $(k-2|x\wedge y|)$-linear function. 
In particular, $h$ is a $k$-linear function if $x$ and $y$ are disjoint, and 1/2-far from any $k$-linear function if $x$ and $y$ intersect.
Now Alice and Bob use a shared random coin to jointly sample one of the deterministic testers that make up the property tester~$\cal T$.
Note that they can simulate a query~$i$ to $h$ by 2 bits of communication: Alice sends $i\cdot x$ to Bob and Bob sends $i\cdot y$ to Alice.
Hence a $T$-query tester for~$\Prop$ implies a $2T$-bit communication protocol for disjointness on weight-$k/2$ inputs $x$ and $y$.
Plugging in the known communication lower bound~\cite{ks:disj,razborov:disj} of $\Omega(k)$ bits for multi-round disjointness on weight-$k/2$ inputs implies that every
classical tester for $k$-linear functions needs $\Omega(k)$ queries, which is nearly tight (the best upper bound is $O(k\log k)$ due to Blais~\cite{blais_juntas}). Plugging in a better $\Omega(k\log k)$ lower bound for \emph{one-way} communication complexity 
gives $T=\Omega(k\log k)$ for \emph{non-adaptive} classical testers (i.e., testers where the next index to query is independent 
of the outcomes of the earlier queries), which is tight~\cite{dks:sparsedisj,bgmw:kparity}.

Can we use the same idea to prove lower bounds on \emph{quantum} testers?
In principle we can, but notice that the overhead when converting a \emph{quantum} tester into a communication protocol is much worse than in the classical case.
In the classical case, thanks to the fact that Alice and Bob can use shared randomness to fix a deterministic tester, they both know at each point in the protocol which query~$i$ will be made next.  Hence they only need to communicate the constant number of bits corresponding to the answer to that query, so the overall communication is $O(T)$.
In the quantum case, the queries can be made in superposition, so the conversion will have an overhead of $O(n)$ qubits of communication: each query will be ``simulated'' by an $n$-qubit message from Alice to Bob, and another such message from Bob to Alice. 
More precisely, suppose we let Alice run the $T$-query quantum tester for~$\Prop$. Whenever the tester wants to make a query to the function~$h$, its state will be in a superposition $\sum_{i\in\01^n}\alpha_i\ket{i}\ket{\phi_i}$ over all indices~$i$, possibly entangled with another register. To perform a phase-query to~$h$, Alice unitarily maps $\ket{i}\mapsto (-1)^{i\cdot x}$, sends the first $n$ qubits of the state to Bob, who unitarily maps $\ket{i}\mapsto (-1)^{i\cdot y}$ and sends back the $n$ qubits. This correctly implements a phase-query to~$h$ 
$$
\ket{i}\mapsto (-1)^{i\cdot x+i\cdot y}=(-1)^{h(i)},
$$
on Alice's state, at the expense of $2n$~qubits of communication.
Thus a $T$-query quantum tester induces a quantum protocol for disjointness that uses $2nT$ qubits of communication.  But the best communication lower bound one can hope for 
on communication complexity problems with $n$-bit inputs is $\Omega(n)$, which gives only a trivial $T=\Omega(1)$ lower bound!
This, however, is not due to a suboptimal reduction: for example, testing $k$-linear functions can be done with $O(1)$ quantum queries using the Bernstein-Vazirani algorithm, as in Section~\ref{sec:bvtester}.

\begin{question}
Can some modification of the ideas of Blais et al.~\cite{bbmj} be used to obtain non-trivial lower bounds on quantum testers?
\end{question}


\section{Classical testing of quantum properties}\label{sec:cltestqprop}

In this section we will survey what is known about classical testing of two kinds of quantum objects: implementations of basic unitary operations, and implementations of quantum protocols that win certain two-player games (most famously the ``CHSH game'') with high probability. 
Even though they are testing properties of \emph{quantum} objects, our testers will be classical in the sense that they will base their decision solely on classical data, in particular classical measurement outcomes from feeding classical inputs into the quantum objects.

Before we go there, let us mention that there is another way in which one can consider classical testing of quantum properties: by imagining that we are given classical access to a quantum object which is too large for an efficient classical description. For example, we might be given access to an unknown pure state $\ket{\psi}$ of $n$ qubits by being allowed to query arbitrary amplitudes in the computational basis at unit cost. This then becomes an entirely classical property testing problem. Some natural properties of quantum states in this context have indeed been studied classically; one example is the {\em Schmidt rank}. A bipartite state $\ket{\psi}$ is said to have Schmidt rank $r$ if it can be written as $\ket{\psi} = \sum_{i=1}^r \sqrt{\lambda_i} \ket{v_i}\ket{w_i}$ for orthonormal sets of states $\{\ket{v_i}\}$, $\{\ket{w_i}\}$ and non-negative $\lambda_i$; this is known as the Schmidt decomposition of $\ket{\psi}$. A tester for this property follows from work of Krauthgamer and Sasson~\cite{krauthgamer03}, who have given an efficient tester for low-rank matrices. Their algorithm distinguishes between the case that a $d\times d$ matrix $M$ is rank at most~$r$, and the case that at least an $\eps$-fraction of the entries in $M$ must be changed to reduce its rank to~$r$. Their algorithm queries only $O((r/\eps)^2)$ elements of the matrix. If we think of $M$ as the amplitudes of a bipartite pure quantum state $\ket{\psi} \in (\C^d)^{\otimes 2}$ (i.e., $M_{ij} = \bra{i}\ip{j}{\psi}$), this is equivalent to a tester for the property of $\ket{\psi}$ having Schmidt rank at most~$r$.

\subsection{Self-testing gates}
\label{sec:stgates}

When experimentalists try to implement a quantum computer in the usual circuit model, they will have to faithfully implement a number of basic quantum operations, called elementary \emph{gates}. Suppose we can implement some superoperator\footnote{Completely positive trace-preserving linear map, a.k.a.~``quantum channel.'' See Section~\ref{sec:channel} for more on these.}~$\bf G$. How can we test whether it indeed implements the gate it is supposed to implement? We are dealing here with the situation of classical testing of quantum properties, which means we can only ``trust'' classical states; we cannot assume that we have trusted machinery to faithfully prepare specific quantum states. What we \emph{can} do is faithfully prepare an initial computational basis state (i.e., a classical state), apply~$\bf G$ to it a number of times, measure the resulting state in the computational basis, and look at the classical outcomes.  

For example, say $\bf G$ is supposed to implement (conjugation by) the Hadamard gate~$H=\frac{1}{\sqrt{2}}\left(
\begin{array}{rr}
1 & 1\\
1 & -1
\end{array}
\right)$. If we prepare $\ket{0}$, apply $\bf G$ once and measure in the computational basis, the probability to see a~0 should be~1/2.  Similarly, if we prepare $\ket{0}$, apply $\bf G$ \emph{twice} and measure, the probability to see~0 should be~1.  These are examples of so-called \emph{experimental equations}. In general, an experimental equation specifies the probability of obtaining a certain outcome from an experiment that starts from a specific classical state and applies a specific sequence of the available superoperators.  A \emph{self-tester} for a set of gates repeatedly performs the experiments corresponding to a specific set of experimental equations, in order to verify that the probabilities of the specified outcomes are indeed (close to) what the equations claim.  A good self-tester will test experimental equations which (when approximately satisfied by $\bf G$), ``essentially'' tell us what $\bf G$ is, in a sense made precise below.

It should be noted that such experimental equations cannot \emph{fully} pin down a gate.  For example, if $\bf G$ is the Hadamard gate in a basis where $\ket{1}$ is replaced with $e^{i\phi}\ket{1}$, then no experiment as described above can detect this: $H$ and its cousin satisfy exactly the same experimental equations, and no self-tester is able to distinguish the two. Still, van Dam et al.~\cite{dmms:selftestj} showed that such experimental equations are surprisingly powerful and can essentially characterize many gate sets, including some universal sets.%
\footnote{A finite set of gates is \emph{universal} if every $n$-qubit unitary
can be approximated arbitrarily well (in the operator norm) by means of a circuit consisting of these gates.
We cannot hope to represent all unitaries \emph{exactly}, because the set of circuits over a finite (or even countable) set
of elementary gates is only countable, hence much smaller than the uncountable set of all unitaries.}
For concreteness we will focus below on a specific universal set, namely the one consisting of the Hadamard gate~$H$, the $\pi/4$-phase gate 
$T=\left(
\begin{array}{cc}
1 & 0\\
0 & e^{i\pi/4}
\end{array}
\right)$, and the controlled-NOT operation. This set has the added benefit that it supports fault-tolerant quantum computing: implementing these gates up to small error suffices for universal quantum computing.

Let us first define experimental equations a bit more precisely. Following van Dam et al.~\cite{dmms:selftestj}, we use $\Pr^c[\rho]$ to denote the probability that measuring the (pure or mixed) state~$\rho$ in the computational basis gives outcome~$c$. Then an experimental equation in one superoperator variable~$\bf G$ is of the form
$$
{\Pr}^c[{\bf G}^k(\ketbra{b}{b})]=r,
$$
for $b,c\in\01$, positive integer $k$, and $r\in[0,1]$.
Note that we assume here that we can apply exactly the same superoperator~${\bf G}$ more than once.
An experimental equation in two variables~$\bf F$ and~$\bf G$ is of the form
$$
{\Pr}^c[{\bf F}^{k_1}{\bf G}^{\ell_1}\cdots {\bf F}^{k_t}{\bf G}^{\ell_t}(\ketbra{b}{b})]=r,
$$
for $b,c\in\01$, integers $k_1,\ldots,k_t,\ell_1,\ldots,\ell_t$, and $r\in[0,1]$
(concatenation of superoperators here denotes composition).
We can similarly write experimental equations in more than two operators, and on systems of more than one qubit. Such experimental equations are all the things a self-tester can test.

Suppose one-qubit operators $\bf H$ and $\bf T$ are intended to be the Hadamard gate~$H$ and the $\pi/4$-phase gate~$T$, respectively, and two-qubit operator $\bf C$ is supposed to be CNOT (with slight abuse of notation we identify unitary gates with the corresponding superoperators here).
Let us see to what extent we can test this.
To start, the following experimental equations are clearly \emph{necessary} for $\bf H$:
\begin{align*}
{\Pr}^0[{\bf H}(\ketbra{0}{0})]& =1/2\\
{\Pr}^0[{\bf H}^2(\ketbra{0}{0})]& =1\\
{\Pr}^1[{\bf H}^2(\ketbra{1}{1})]& =1
\end{align*}
Van Dam et al.~\cite[Theorem~4.2]{dmms:selftestj} showed that these equations \emph{characterize} the Hadamard gate up to the one remaining degree of freedom that we already mentioned, in the following sense:
$\bf H$ satisfies the above three equations if and only if
there exists $\phi\in[0,2\pi)$ such that $\bf H$ equals (the superoperator corresponding to) $H_{\phi}$, which is the Hadamard gate where $\ket{1}$ is replaced with $e^{i\phi}\ket{1}$:
$$
H_{\phi}=\frac{1}{\sqrt{2}}\left(
\begin{array}{cc}
1 & e^{-i\phi}\\
e^{i\phi} & -1
\end{array}
\right).
$$
The unknown phase~$\phi$ cannot be ignored, because it might interact with the effects of other gates.

The following two experimental equations are clearly necessary for~$\bf T$:
\begin{align*}
{\Pr}^0[{\bf T}(\ketbra{0}{0})]& =1\\
{\Pr}^1[{\bf T}(\ketbra{1}{1})]& =1
\end{align*}
These two equations are far from sufficient for characterizing the $T$~gate; for example, every diagonal unitary will satisfy these two equations, as would the superoperator that fully decoheres a qubit in the computational basis. However, by introducing some additional equations involving both $\bf H$ and $\bf T$ we can do better:
\begin{align*}
{\Pr}^0[{\bf H}{\bf T}^8{\bf H}(\ketbra{0}{0})]& =1\\
{\Pr}^0[{\bf H}{\bf T}{\bf H}(\ketbra{0}{0})]& =\frac{1}{2}(1+\cos(\pi/4))
\end{align*}
Note that if ${\bf H}=H$, then both ${\bf T}=T$ and its inverse ${\bf T}=T^{-1}$ would satisfy the above equations; this is unfortunate, but will turn out below not to matter.
Van Dam et al.~\cite[Theorem~4.4]{dmms:selftestj} showed that a pair of superoperators $\bf H$ and $\bf T$ satisfy the above set of 7~equations if and only if there exists $\phi\in[0,2\pi)$ such that ${\bf H}=H_{\phi}$, and ${\bf T}$ corresponds to either $T$ or $T^{-1}$.

To complete our self-test, consider the superoperator $\bf C$.
The following experimental equations are clearly necessary for $\bf C$ to equal CNOT:
\begin{align*}
{\Pr}^{00}[{\bf C}(\ketbra{00}{00})]& =1\\
{\Pr}^{01}[{\bf C}(\ketbra{01}{01})]& =1\\
{\Pr}^{11}[{\bf C}(\ketbra{10}{10})]& =1\\
{\Pr}^{10}[{\bf C}(\ketbra{11}{11})]& =1
\end{align*}
These equations ensure that $\bf C$ implements the same permutation of basis states as the CNOT gate. This is still far from sufficient. We add the following experimental equations, which describe the desired interaction between CNOT and $H$:
\begin{align*}
{\Pr}^{00}[({\bf I\otimes H}){\bf C}({\bf I\otimes H})(\ketbra{00}{00})]& =1\\
{\Pr}^{10}[({\bf I\otimes H}){\bf C}({\bf I\otimes H})(\ketbra{10}{10})]& =1\\
{\Pr}^{00}[({\bf H\otimes I}){\bf C}^2({\bf H\otimes I})(\ketbra{00}{00})]& =1\\
{\Pr}^{01}[({\bf H\otimes I}){\bf C}^2({\bf H\otimes I})(\ketbra{01}{01})]& =1\\
{\Pr}^{00}[({\bf H\otimes H}){\bf C}({\bf H\otimes H})(\ketbra{00}{00})]& =1
\end{align*}
Van Dam et al.~\cite[Theorem~4.5]{dmms:selftestj} showed that
if superoperators $\bf H$, $\bf T$, $\bf C$ satisfy the above 16~experimental equations, then there exists $\phi\in[0,2\pi)$ such that:
$$
{\bf H}=H_{\phi}; {\bf T}=T \mbox{ or }{\bf T}=T^{-1}; {\bf C}=C_{\phi},
$$
where $C_\phi$ denotes (the superoperator corresponding to the) controlled-NOT gate with $\ket{1}$ replaced with $e^{i\phi}\ket{1}$.

Because our apparatuses are never perfect, we cannot hope to implement the elementary gates
exactly. Fortunately, thanks to quantum fault-tolerant computing
it suffices if we can implement them up to small error (in fact different applications
of the same superoperator can have different errors and need not all be identical).
Hence we also cannot expect the gates that we are testing to \emph{exactly} satisfy all of the above experimental equations.
Furthermore, even if they did satisfy these equations exactly, we would never be able
to perfectly test this with a finite number of experiments.
Accordingly, we would like the test consisting of these experimental equations to be \emph{robust}, in the sense that if $\bf H$, $\bf T$, and $\bf C$ \emph{approximately} satisfy these equations, then they will be \emph{close} to the gates they purport to be. We say that superoperators \emph{$\eps$-satisfy} a set of experimental equations if for each of the equations the left- and right-hand sides differ by at most~$\eps$. We measure closeness between superoperators in the norm induced by the trace norm:\footnote{This norm $\norm{G}_{\infty}$ is different from (and weaker than) the diamond norm defined later in Eq.~\eqref{def:diamondnorm}, which is also often used to measure distance between superoperators.}
$$
\norm{G}_{\infty}=\sup\{\norm{G(V)}_1:\norm{V}_1=1\},
$$
where the trace norm (Schatten 1-norm) is defined as $\|M\|_1 := \Tr(|M|)$.

Van Dam et al.~\cite[Theorem~6.5, last item]{dmms:selftestj} indeed showed that the above equations constitute a robust self-test:

\begin{thm}[van Dam et al.~\cite{dmms:selftestj}]\label{thm:selftestHTC}
There exists a constant $c$ such that for all $\eps>0$ the following holds.
If superoperators $\bf H$, $\bf T$, $\bf C$ $\eps$-satisfy the above 16~experimental equations, then there exists $\phi\in[0,2\pi)$ such that:
$$
\norm{{\bf H}-H_\phi}_\infty\leq c\sqrt{\eps};
\norm{{\bf T}-T}_\infty\leq c\sqrt{\eps}\mbox{ or }
\norm{{\bf T}-T^{-1}}_\infty\leq c\sqrt{\eps};
\norm{{\bf C}-C_{\phi}}_\infty\leq c\sqrt{\eps}.
$$
\end{thm}

Let us mention explicitly how this testing of sets of gates fits in the framework outlined in the introduction.  The universe now consists of all triples of superoperators $({\bf H},{\bf T},{\bf C})$. The property $\Prop$ consists of all triples for which there is a $\phi$ such that such that ${\bf H}=H_\phi$, ${\bf T}=T$ or  ${\bf T}=T^{-1}$, and ${\bf C}=C_{\phi}$. The distance measure would be 
$$
d(({\bf H},{\bf T},{\bf C}),({\bf H'},{\bf T'},{\bf C'}))=\max\left(\norm{{\bf H}-{\bf H'}}_\infty,\norm{{\bf T}-{\bf T'}}_\infty,\norm{{\bf C}-{\bf C'}}_\infty\right).
$$
One can derive a tester from Theorem~\ref{thm:selftestHTC} by running the experiments for each experimental equation $O(1/\eps)$ times, estimating the probabilities in their right-hand side up to additive error $c\sqrt{\eps}$, and accepting if and only if for each of the 16 equations, the estimate is $c\sqrt{\eps}$-close to what it should be.  This will accept (with high probability) every triple in $\Prop$, and reject (with high probability) every triple that is $2c\sqrt{\eps}$-far from $\Prop$.

Each triple $({\bf H},{\bf T},{\bf C})$ that passes the test is a universal (and fault-tolerant) set of elementary gates, so can in principle be used to realize any quantum circuit. The fact that we do not know $\phi$ is not important when implementing a circuit using this triple of gates: since $\phi$ cannot be detected by any experimental equations, it cannot affect the classical input-output behavior of a quantum circuit built from these superoperators. We also do not know whether $\bf T$ approximately equals $T$ or its inverse $T^{-1}$. Using Hadamard and CNOTs cannot help distinguish these two cases, because they only differ in a minus sign for the imaginary unit (something gates with real entries cannot pick up).  However, precisely because such a change is undetectable experimentally, we can just build our circuit assuming $\bf T$ is close to $T$; if it is close to $T^{-1}$ instead, that will incur no observable differences in the input-output behavior of our circuit, so for all intents and purposes we may just assume assume $\bf T$ is close to $T$.

In addition to the above result, van Dam et al.~\cite{dmms:selftestj} also showed a number of other families of gates to be robustly self-testable, and proved more general robustness results. In follow-up work, Magniez et al.~\cite{mmmo:selftesting} study self-testing of quantum circuits together with measurement apparatuses and sources of EPR-pairs, introducing notions of simulation and equivalence.

\subsection{Self-testing protocols}\label{sec:stprotocols}

In addition to quantum gates and circuits, a large area of application of quantum self-testing is in multi-party quantum \emph{protocols}.
Here typically two or more parties share an entangled state on which they operate locally. In the two-party case these are often EPR-pairs---or at least \emph{should} be EPR-pairs. Experimentalists often need to test that their apparatuses actually produce the required entangled state, or at least something close to it, and that the local operations and measurements act as required. Unless we somehow already have some other trusted quantum objects available, we are in the self-testing regime: like in the previous section, we can only trust preparations of classical states and measurements in the computational basis. We would like to test a quantum object by classically interacting with it, without making assumptions about the measurement apparatuses, the states used, or even the dimension of the Hilbert spaces that are involved.

Again, for concreteness we will focus on testing protocols for one specific example in the two-party setting%
\footnote{In the three-party setting, the most famous game is the GHZ game~\cite{ghz}. Colbeck~\cite{colbeck:phd} seems to have been the first to give a self-testing result for this.}, namely the famous CHSH game~\cite{chsh}. This is defined as follows.
\begin{quote}
{\bf CHSH game.}
Alice and Bob receive uniformly distributed inputs $x,y\in\01$, respectively. They output $a,b\in\01$, respectively. The players (equivalently, the protocol) \emph{win} the game if and only if the XOR of the outputs equals the AND of the inputs: $a\oplus b=xy$.
\end{quote}
Alice and Bob want to coordinate to maximize their probability\footnote{This probability is taken over the input distribution as well as over the internal randomness of the protocol.} of winning this game, without communication between them. It is known that classical protocols can win with probability $0.75$, but not more, even when they use shared randomness. In contrast, the following quantum protocol $P^*$ wins the game with probability $\cos(\pi/8)^2\approx 0.854$.%
\footnote{This ``Bell inequality violation'' has been confirmed by many experiments, albeit with a few remaining experimental ``loopholes,'' suggesting that Nature does not behave according to classical physics.
See the recent survey by Brunner et al.~\cite{bcpsw:bellnonlocality} for much more on such ``nonlocal'' behavior, where two spatially separated entangled players are correlated in ways that are impossible for classical players.}
It is defined in terms of the four single-qubit \emph{Pauli matrices}, which are
\begin{equation*}
I = \mm{1&0\\0&1},\;X=\mm{0&1\\1&0},\;Y=\mm{0&-i\\i&0},\;Z=\mm{1&0\\0&-1}. 
\end{equation*}
\begin{quote}
{\bf Standard protocol for CHSH.}
$P^*$ uses one EPR-pair $\ket{\phi^+}=\frac{1}{\sqrt{2}}(\ket{00}+\ket{11})$ as starting state. Depending on their inputs, Alice and Bob apply the following specific $\pm 1$-valued observables\footnote{A $\pm$-valued observable~$A$ can be written as the difference $A=P_+ - P_-$ of two orthogonal projections that satisfy $P_++P_-=I$. It corresponds to a projective measurement in a natural way, with outcome $+1$ corresponding to $P_+$ and outcome $-1$ corresponding to $P_-$. Note that such an $A$ is both Hermitian and unitary, and hence $A^2=I$.}: Alice measures $X$ if $x=0$, or $Z$ if $x=1$. She outputs~0 if her measurement yields~1, and she outputs 1 if it yields~$-1$. Bob measures the observable~$(X+Z)/\sqrt{2}$ if~$y=0$ and $(X-Z)/\sqrt{2}$ if~$y=1$, and outputs 0 or 1 accordingly. 
\end{quote}
Note that for $\pm 1$-valued observables $A$ and $B$, $\bra{\phi^+}A\otimes B\ket{\phi^+}=\Tr(AB^T)/2$ is the difference between the probability that the two output bits are equal and the probability that the outputs are different. If $xy=0$ a protocol tries to get this difference close to~1, and if $x=y=1$ it tries to get the difference close to~$-1$. In the above protocol $P^*$, the difference is $1/\sqrt{2}$ if $xy=0$, and $-1/\sqrt{2}$ if $x=y=1$, so the sum of these 4 terms (negating the last one) equals $2\sqrt{2}$. 
Tsirelson famously proved that this value of $2\sqrt{2}$ is optimal among all possible protocols~\cite{tsirelson80}, no matter how much entanglement they use;  hence the corresponding winning probability $\frac{1}{2}+\frac{1}{2\sqrt{2}}=\cos(\pi/8)^2$ is optimal as well.

\begin{thm}[Tsirelson~\cite{tsirelson80}]\label{thm:tsirelsonbound}
Suppose Alice and Bob run a protocol for CHSH that starts with a shared pure state $\ket{\psi}$, where Alice applies $\pm 1$-valued observables $A_0$ or $A_1$ depending on her input~$x$, and Bob applies $\pm 1$-valued observables $B_0$ or $B_1$ depending on~$y$. Then
$|\bra{\psi}(A_0B_0+A_0B_1+A_1B_0-A_1B_1)\ket{\psi}|\leq 2\sqrt{2}$.
\end{thm}

For simplicity we abbreviate $A\otimes B$ to $AB$ in the above statement as well as the rest of this section (and $A\otimes I$ to just $A$). The assumption that the starting state is pure and that Alice and Bob apply $\pm 1$-valued observables is without loss of generality, so Tsirelson's bound covers all possible quantum protocols.

\begin{proof}
Define $C=A_0B_0+A_0B_1+A_1B_0-A_1B_1$. Using that $A_x^2=B_y^2=I$, the square of~$C$ works out to
$$
C^2=4I+[A_0,A_1]\otimes[B_1,B_0],
$$
where $[A,B]=AB-BA$ denotes the \emph{commutator} of two operators. Note that if $\|A\|,\|B\|\leq 1$ then $\norm{[A,B]}\leq 2$.
Hence, using Cauchy-Schwarz we get: 
\begin{align*}
|\bra{\psi}C\ket{\psi}|^2 & \leq\bra{\psi}C^2\ket{\psi}=4+\bra{\psi}[A_0,A_1]\otimes[B_1,B_0]\ket{\psi}\\
& =4+\bra{\psi}([A_0,A_1]\otimes I)\cdot(I\otimes[B_1,B_0])\ket{\psi}\\
& \leq 4+\norm{([A_0,A_1]\otimes I)\ket{\psi}}\cdot\norm{(I\otimes[B_1,B_0])\ket{\psi}}\\
& \leq 4+\norm{[A_0,A_1]}\cdot\norm{[B_1,B_0]}\leq 4+2\cdot 2=8,
\end{align*}
which implies $|\bra{\psi}C\ket{\psi}|\leq 2 \sqrt{2}$.
\end{proof}

There are many different protocols that achieve the optimal value $2\sqrt{2}$ or something close to it. For example, applying a local basis change to $P^*$ results in a different protocol that still achieves the maximal value.  How much freedom do we have in such optimal or near-optimal protocols for the CHSH game? Surprisingly, this freedom is essentially limited to local basis transformations. Popescu and Rohrlich~\cite{popescu&rohrlich:generic} and Braunstein et al.~\cite{bmr:maxbellviolation} independently showed that any protocol that wins CHSH with maximal probability needs to start with an EPR-pair,
or something that can be turned into an EPR-pair (possibly in tensor product with another state shared between Alice and Bob) using local isometries.%
\footnote{The correct attribution of this result is not completely clear, see also the work of Summers and Werner~\cite{summers&werner:maxviolation} and Tsirelson~\cite[p.~11]{tsirelson:some}.}

However, as in the previous section, \emph{robustness} is important: we expect that if a protocol wins the CHSH game with \emph{close-to}-maximal probability, then its entangled state must be \emph{close} to an EPR-pair, and its measurement operators must be in some sense close to those of the standard protocol. Such a robust result was proved independently in~\cite{mys:robustselftest,miller&shi:selftestingxor}\footnote{The earlier work of Mayers and Yao~\cite{mayers&yao:impapp,mayers&yao:selftesting} that started the area of self-testing of quantum states also had a protocol for robustly self-testing EPR-pairs, albeit based on more than the CHSH game.}:

\begin{thm}[\cite{mys:robustselftest,miller&shi:selftestingxor}]\label{th:chshselftest}
Suppose Alice and Bob run a protocol for CHSH that starts with a shared pure state $\ket{\psi}$, where Alice applies $\pm 1$-valued observables $A_0$ or $A_1$ depending on her input~$x$, and Bob applies $\pm 1$-valued observables $B_0$ or $B_1$ depending on~$y$. Suppose the protocol wins CHSH with probability at least $\cos(\pi/8)^2-\eps$. Define new operators for Alice and Bob, respectively:
$$
\begin{array}{ll}
X'_A=A_0, & Z'_A=A_1,\\
X'_B=\displaystyle\frac{B_0+B_1}{\sqrt{2}}, & Z'_B=\displaystyle\frac{B_0-B_1}{\sqrt{2}}.
\end{array}
$$
Then there exists a local isometry $\Phi=\Phi_A\otimes\Phi_B$ and a pure state $\ket{junk}$ shared between Alice and Bob, such that for all $M,N\in\{I,X,Z\}$ we have
$$
\norm{\Phi(M'_AN'_B\ket{\psi})-\ket{junk}\otimes M_A N_B\ket{\phi^+}}\,=O(\sqrt{\eps}),
$$
where, e.g., if $M=X$ the notation $M'_A$ denotes the operator $X'_A$.
\end{thm}

In words, up to a local basis change and small errors depending on~$\eps$, $\ket{\psi}$ behaves like an EPR-pair and $X'_A,Z'_A,X'_B,Z'_B$ behave like the standard Pauli operators $X$ and $Z$ for Alice and Bob, respectively, applied to that EPR-pair. Note that this also implies that $A_0,A_1,B_0,B_1$ behave like the observables of the standard protocol~$P^*$.
We give the proof of~\cite{mys:robustselftest} here for the special case where $\eps=0$.  This allows us to describe the main ideas, without going into the technical but straightforward details needed to keep track of the errors and approximations.

\begin{proof}[Proof for $\eps=0$]
Consider the proof of Tsirelson's bound (Theorem~\ref{thm:tsirelsonbound}). If a protocol achieves the maximum value $2\sqrt{2}$, then the inequalities in the proof must be equalities. This implies that $\norm{([A_0,A_1]\otimes I)\ket{\psi}}=2$ and hence (omitting the identities on Bob's space as before):
$$
A_0A_1\ket{\psi}=-A_1A_0\ket{\psi}.
$$
So $A_0$ and $A_1$ anti-commute on $\ket{\psi}$.
Similarly $|\bra{\psi}[B_1,B_0]\ket{\psi}|=2$, and hence $B_0$ and $B_1$ anti-commute on $\ket{\psi}$ as well:
$$
B_0B_1\ket{\psi}=-B_1B_0\ket{\psi}.
$$

We list some properties of the operators $X'_A, Z'_A, X'_B, Z'_B$ that were defined in the statement of the theorem. All are clearly Hermitian.  On Alice's side, $X'_A$ and $Z'_A$ are unitary because $A_0$ and $A_1$ are. They anti-commute on $\ket{\psi}$ because $A_0$ and $A_1$ do. On Bob's side, $X'_B$ and $Z'_B$ anti-commute. We cannot assume $X'_B$ and $Z'_B$ are unitary. However, since $(X'_B)^2=I+(B_0B_1+B_1B_0)/2$ and $B_0$ and $B_1$ anti-commute on $\ket{\psi}$, we have $(X'_B)^2\ket{\psi}=\ket{\psi}$. Hence $\norm{X'_B\ket{\psi}}^2=\bra{\psi}(X'_B)^2\ket{\psi}=1$, so $X'_B$ preserves the norm of $\ket{\psi}$. Similarly, $Z'_B$ preserves the norm of $\ket{\psi}$, as does $X'_BZ'_B$. 

We now want to show that $X'_AX'_B\ket{\psi}=\ket{\psi}$.
First,
\begin{equation}\label{eq:ABC}
\bra{\psi}A_0(B_0+B_1)\ket{\psi}+\bra{\psi}A_1(B_0-B_1)\ket{\psi}=\bra{\psi}C\ket{\psi}=2\sqrt{2}.
\end{equation}
Second, by squaring the operator $A_0(B_0+B_1)$ and using anti-commutativity of $B_0$ and $B_1$ on $\ket{\psi}$ we can show $\bra{\psi}A_0(B_0+B_1)\ket{\psi}\leq\sqrt{2}$, and similarly $\bra{\psi}A_1(B_0-B_1)\ket{\psi}\leq\sqrt{2}$. Combining with Eq.~\ref{eq:ABC}, it follows that both terms \emph{equal} $\sqrt{2}$. Then we have
$$
\bra{\psi}X'_AX'_B\ket{\psi}=\frac{1}{\sqrt{2}}\bra{\psi}A_0(B_0+B_1)\ket{\psi}=1,
$$
hence $X'_AX'_B\ket{\psi}=\ket{\psi}$. Since $X'_A$ is unitary and Hermitian it is self-inverse, which implies $X'_A\ket{\psi}=X'_B\ket{\psi}$.
A similar argument shows $Z'_AZ'_B\ket{\psi}=\ket{\psi}$ and $Z'_A\ket{\psi}=Z'_B\ket{\psi}$.

We now need to show that, after a local isometry, $\ket{\psi}$ behaves like an EPR-pair (tensored with some ``junk'' state) and $X'_A, Z'_A, X'_B, Z'_B$ behave like $X_A, Z_A, X_B, Z_B$.  Consider the dimension-increasing map on states $\ket{\phi}$ (in the same space as $\ket{\psi}$) that is described by Figure~\ref{fig:chshisometry}. It adds one auxiliary qubit for Alice (at the top line of the figure) and one for Bob (at the bottom), both initially~$\ket{0}$. Because all operators involved preserve norm on all states involved, this can be extended to a local isometry $\Phi=\Phi_A\otimes\Phi_B$.  

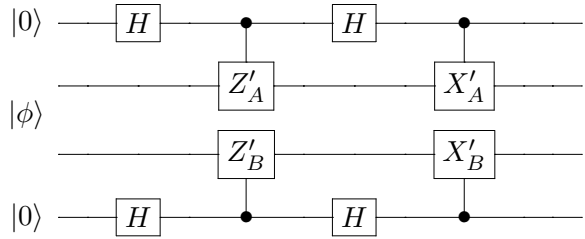
\begin{figure}[hbt]
\[
\Qcircuit @C=1em @R=.7em {
\lstick{\ket{0}} & \qw & \gate{H} & \qw & \ctrl{1} & \qw & \gate{H} & \qw & \ctrl{1} & \qw & \qw & \qw \\
\dstick{\ket{\phi}\hspace*{2.2em}} & \qw & \qw & \qw & \gate{Z'_A} & \qw & \qw & \qw & \gate{X'_A} & \qw & \qw & \qw \\
\lstick{} & \qw & \qw  &\qw & \gate{Z'_B} & \qw & \qw & \qw & \gate{X'_B} & \qw & \qw & \qw \\
\lstick{\ket{0}} & \qw  & \gate{H} & \qw & \ctrl{-1} & \qw & \gate{H} & \qw & \ctrl{-1} & \qw & \qw & \qw 
}
\]
\caption{Isometry for transforming a perfect CHSH protocol to the standard one.}
\label{fig:chshisometry}
\end{figure}

For convenience we will write the two auxiliary qubits on the right of the state, the first for Alice and the second for Bob. Let $M,N\in\{I,X,Z\}$. Following the state through the different steps of Figure~\ref{fig:chshisometry}, a straightforward calculation shows
\begin{align}\label{eq:isometry}
\Phi(M'_AN'_B\ket{\psi})& = \frac{1}{4}(I+Z'_A)(I+Z'_B)M'_AN'_B\ket{\psi}\ket{00}\nonumber\\
                & +  \frac{1}{4}X'_B(I+Z'_A)(I-Z'_B)M'_AN'_B\ket{\psi}\ket{01}\nonumber\\
                & +  \frac{1}{4}X'_A(I-Z'_A)(I+Z'_B)M'_AN'_B\ket{\psi}\ket{10}\nonumber\\
                & +  \frac{1}{4}X'_AX'_B(I-Z'_A)(I-Z'_B)M'_AN'_B\ket{\psi}\ket{11}.
\end{align}
First consider the case where $M=N=I$. Then the second term vanishes, because $I\ket{\psi}=Z'_AZ'_B\ket{\psi}$ and $Z'_A\ket{\psi}=Z'_B\ket{\psi}$.
Similarly the third term vanishes.
The fourth term equals the first (except in the last two qubits) because $X'_AX'_B(I-Z'_A)(I-Z'_B)\ket{\psi}=(I+Z'_A)(I+Z'_B)X'_AX'_B\ket{\psi}$ by anti-commutativity, and $X'_AX'_B\ket{\psi}=\ket{\psi}$.
Hence we end up with
$$
\Phi(\ket{\psi})  = \left(\frac{1}{2\sqrt{2}}(I+Z'_A)(I+Z'_B)\ket{\psi}\right)\otimes\frac{1}{\sqrt{2}}(\ket{00}+\ket{11})=\ket{junk}\otimes\ket{\phi^+},
$$
where we defined $\ket{junk}:=\frac{1}{2\sqrt{2}}(I+Z'_A)(I+Z'_B)\ket{\psi}$.

If $MN=XX$ then the same proof applies, because $X'_AX'_B\ket{\psi}=\ket{\psi}$ and $X_AX_B\ket{\phi^+}=\ket{\phi^+}$. The same holds if $MN=ZZ$. 

Now consider the case $MN=XZ$. Looking at Eq.~(\ref{eq:isometry}), the first term vanishes because 
$$
(I+Z'_A)(I+Z'_B)X'_AZ'_B\ket{\psi}=X'_A(I-Z'_A)(I+Z'_B)\ket{\psi}=0, 
$$
using the anti-commutativity of $X'_A$ and $Z'_A$, and the fact that $(I+Z'_B)Z'_B\ket{\psi}=(I+Z'_B)\ket{\psi}$ (because $(Z'_B)^2\ket{\psi}=I\ket{\psi}$). Similarly the fourth term vanishes. For the second term we use
\begin{align*}
X'_B(I+Z'_A)(I-Z'_B)X'_AZ'_B\ket{\psi} & =(I+Z'_A)(I+Z'_B)X'_AX'_BZ'_B\ket{\psi}\\
                                      & =-(I+Z'_A)(I+Z'_B)X'_AX'_B\ket{\psi} =-(I+Z'_A)(I+Z'_B)\ket{\psi},
\end{align*}
where we used $X'_BZ'_B\ket{\psi}=-Z'_BX'_B\ket{\psi}$, $X'_AX'_B\ket{\psi}=\ket{\psi}$, and $(I+Z'_B)Z'_B\ket{\psi}=(I+Z'_B)\ket{\psi}$. We similarly analyze the third term. We end up with
\begin{align*}
\Phi(X'_AZ'_B\ket{\psi})&  = -\frac{1}{4}(I+Z'_A)(I+Z'_B)\ket{\psi}\ket{01}+\frac{1}{4}(I+Z'_A)(I+Z'_B)\ket{\psi}\ket{10}\\ 
                       & = \frac{1}{2\sqrt{2}}(I+Z'_A)(I+Z'_B)\ket{\psi}\otimes\frac{1}{\sqrt{2}}(\ket{10}-\ket{01})=\ket{junk}\otimes X_AZ_B\ket{\phi^+}.
\end{align*}
For the other five possible $M,N$ pairs, a similar calculation (starting from Eq.~(\ref{eq:isometry}) and using the known commutation and anti-commutation properties) works to establish the desired property: 
$$
\Phi(M'_AN'_B\ket{\psi})=\ket{junk}\otimes M_A N_B\ket{\phi^+}.
$$
\end{proof}

Accordingly, we can use this robust result to test whether a given protocol behaves essentially like $P^*$, based only on classical-input output behavior: run it multiple times on uniformly distributed classical input bits, observe the classical output bits, and see if the winning probability is close to the optimal value $\cos(\pi/8)^2$. If so, then (up to local change of basis) the state must be close to an EPR-pair tensored with some other ``junk'' state, and the behavior of the measurements must be close to the ones of the standard CHSH protocol~$P^*$.  

There has been a lot more work along these lines. McKague et al.~\cite{mys:robustselftest} give a more general framework for bipartite robust self-testing that subsumes the CHSH inequality, the Mayers-Yao self-test (simplifying~\cite{mmmo:selftesting}), as well as others. Yang and Navascu\'es~\cite{yang&navascues} give robust self-tests for any entangled two-qubit states, not just maximally entangled ones; the noise-resistance was further improved in~\cite{yvbsn:openingblackbox}. McKague~\cite{mckague:selftestgraph} and Miller and Shi~\cite{miller&shi:selftestingxor} give results about self-testing of states shared by more than two parties. 

In some applications one needs to have many states that all behave like EPR-pairs, not just the one EPR-pair that is needed for an optimal protocol for CHSH. Recently, Reichardt et al.~\cite{ruv:leash} proved a subtle robustness result for playing \emph{many instances} of CHSH.  Roughly, their result says: if a quantum protocol wins a fraction of nearly $\cos(\pi/8)^2$ of a sequence of $k$ given instances of the CHSH game, then most blocks of $m=k^{\Omega(1)}$ instances have the property that they start ``essentially'' (again, up to local operations and small differences like in Theorem~\ref{th:chshselftest}) from $m$~EPR-pairs and run $m$ independent instances of the standard protocol~$P^*$. With significant additional work it is possible to use this result to devise methods that allow a classical system to ``command'' an untrusted quantum system, in the sense of forcing that quantum system to either use essentially the states and operations you want it to use, or be detected if it deviates too much from those states and operations. Such control enables various kinds of device-independent quantum cryptography, as well as the ability to offload general quantum computation to untrusted devices.


\section{Quantum testing of quantum properties: States}\label{sec:states}

In the third part of this survey we discuss quantum testers for quantum properties. The first decision we have to take in this setting is how the quantum object which we wish to test is presented to us. The two options are a quantum presentation (i.e., we are given access to the object as a black box, which can be used in a quantum algorithm), or a classical presentation (i.e., we are given an efficient classical description of the object, such as a quantum circuit). We concentrate on the former option (Sections~\ref{sec:states}--\ref{sec:dynamics}), as this seems to be the most natural generalization of ideas from classical property testing. However, in Section~\ref{sec:compcomp} we also discuss the latter option, which turns out to be important in quantum computational complexity.

Our focus in this part of the survey is on quantum tests for quantum properties which generalize the idea of classical property testing. That is, tests which are designed to distinguish quantum states (or operations) with some property from those far from having that property, given access to the state (or operation) as a black box. We also mention here two related and well-studied areas elsewhere in quantum information theory. The first is quantum state discrimination, which can be seen as a quantum generalization of classical hypothesis testing. The archetypal problem in this setting is as follows: given the ability to create copies of an unknown quantum state $\rho$ picked from a known set $S$ of quantum states, identify $\rho$ with minimal probability of error.
Some authors use the term ``quantum hypothesis testing'' for this problem~\cite{chefles00}; others reserve this term for the case $|S|=2$, where precise results have been obtained relating the optimal error probability to the number of copies of $\rho$ consumed, and trade-offs between different kinds of error have been determined~\cite{audenaert08}. See the surveys~\cite{barnett09,chefles00} for detailed reviews of quantum state discrimination. The second area is the question of directly estimating some quantity of interest about a completely unknown quantum state $\rho$, given access to multiple copies of the state, without performing full tomography. Results of this form include direct estimation of the spectrum of $\rho$~\cite{keyl01}, estimation of polynomials in the entries of $\rho$~\cite{brun04}, and estimation of quantities related to entanglement (e.g.,~\cite{guhne09}).

We begin our discussion of quantum properties by considering properties of quantum states, first pure states and then mixed states.



\subsection{Pure states}
\label{sec:pure}

A pure state $\ket{\psi}$ of a $d$-dimensional quantum system is described by a $d$-dimensional complex unit vector (technically, a ray; that is, $e^{i \theta} \ket{\psi}$ is equivalent to $\ket{\psi}$ for all real $\theta$). A property of $d$-dimensional pure quantum states is therefore a set $\Prop \subseteq \C^d$. One can naturally generalize this to properties of pairs of quantum states, where $\Prop \subseteq \C^d \times \C^d$, etc.

There is a natural measure of distance between quantum states $\ket{\psi}$ and $\ket{\phi}$: the trace distance
\be \label{eq:tracedistance} D(\ket{\psi},\ket{\phi}) := \frac{1}{2} \|\proj{\psi} - \proj{\phi}\|_1 = \sqrt{1-|\ip{\psi}{\phi}|^2}. \ee
Here, as in Section~\ref{sec:stgates}, $\|\cdot\|_1$ is the trace norm (Schatten 1-norm) $\|M\|_1 := \Tr(|M|)$. Given a state promised to be either $\ket{\psi}$ or $\ket{\phi}$, with equal probability of each, the optimal probability of determining via a measurement which state we have is exactly $(1+D(\ket{\psi},\ket{\phi}))/2$~\cite{helstrom76,nielsen&chuang:qc}. We therefore say that $\ket{\psi}$ is \emph{$\eps$-close} to having property $\Prop$ if
\[ D(\ket{\psi},\Prop) := \inf_{\ket{\phi} \in \Prop} D(\ket{\psi},\ket{\phi}) \le \eps, \]
and similarly that $\ket{\psi}$ is \emph{$\eps$-far} from having property $\Prop$ if $D(\ket{\psi},\Prop) \ge \eps$. If $\ket{\psi}$ is $\eps$-close to having property $\Prop$, there is no hope of certifying that $\ket{\psi} \notin \Prop$ with worst-case bias larger than $\eps$, given access to only one copy of $\ket{\psi}$.

The complexity of algorithms for testing pure quantum states is measured by the number of copies of the test state $\ket{\psi}$ required to distinguish between the two cases that (a) $\ket{\psi} \in \Prop$, or (b) $\ket{\psi}$ is $\eps$-far away from having property $\Prop$.
We therefore say that $\Prop$ can be \emph{$\eps$-tested with $q$ copies} if there exists a quantum algorithm which uses $q$ copies of the input state to distinguish between these two cases, and fails with probability at most $1/3$ on any input. As with classical property testers, we say that a tester has \emph{perfect completeness} if it accepts every state in $\Prop$ with certainty.  Crucially, we look for algorithms where the number of copies used scales only in terms of $\eps$, and there is no dependence on the dimension~$d$, making this a fair analog of the classical concept. If we cannot find such an algorithm, we attempt to minimize the dependence on~$d$.

On the other hand, if we do not care about the dependence on $d$, \emph{any} (even infinite) property $\Prop \subseteq \C^d$ can be tested using $O(d/\eps^2)$ copies of the input state $\ket{\psi}$; it suffices to obtain an estimate $\ket{\psi'}$ such that $D(\ket{\psi'},\ket{\psi}) < \eps/2$, and accept if and only if $D(\ket{\psi'},\Prop) \le \eps/2$. In order to produce such an estimate one can use a procedure known as \emph{quantum state estimation}, which needs $O(d/\eps^2)$ copies of $\ket{\psi}$ to achieve the required accuracy with success probability at least $2/3$~\cite{bruss99}.


\subsubsection{Equality}
\label{sec:equality}

The first property we consider is extremely basic, but a useful building block for more complicated protocols: whether the input state is equal to some fixed state. We say that a state $\ket{\psi}$ satisfies the {\bf Equality to $\ket{\phi}$} property if $\ket{\psi} = e^{i \theta} \ket{\phi}$ for some real $\theta$, so $\Prop=\{e^{i \theta} \ket{\phi}:\theta\in\R\}$; it is necessary to allow an arbitrary phase~$\theta$ in the definition of this property, as $\ket{\psi}$ cannot be distinguished from $e^{i \theta} \ket{\psi}$ by any measurement. A natural test for {\bf Equality to $\ket{\phi}$} is simply to perform the measurement $\{ \proj{\phi}, I-\proj{\phi} \}$ on $\ket{\psi}$, and accept if and only if the first outcome is obtained. The probability of acceptance is precisely $|\ip{\psi}{\phi}|^2$, so if $\ket{\psi}$ satisfies the property, the test accepts with certainty. On the other hand, if $D(\ket{\psi},\ket{\phi})=\eps$, the test rejects with probability $1-|\ip{\psi}{\phi}|^2=\eps^2$. Via repetition, we find that for any~$\ket{\phi}$, {\bf Equality to $\ket{\phi}$} can be tested with $O(1/\eps^2)$ copies.

A matching lower bound follows from considering the special case where the input state is promised to be either $\ket{\phi}$ or some state $\ket{\phi'}$ such that $D(\ket{\phi},\ket{\phi'})=\eps$, with equal probability of each. Then any test which uses $k$ copies to test whether the input is equal to $\ket{\phi}$ is equivalent to a procedure which discriminates between $\ket{\phi}^{\otimes k}$ and $\ket{\phi'}^{\otimes k}$, which has success probability upper-bounded by $(1+D(\ket{\psi}^{\otimes k},\ket{\phi}^{\otimes k}))/2$. Using the definition (\ref{eq:tracedistance}) of the trace distance, we require $k = \Omega(1/\eps^2)$ to achieve success probability $2/3$. This same argument in fact shows that \emph{any} non-trivial property of pure states requires $\Omega(1/\eps^2)$ copies to be tested.

We remark that testing equality to a fixed state immediately generalizes to the problem of testing whether $\ket{\psi} \in \C^d$ is contained in a known subspace $S \subseteq \C^d$. Here the prescription is to perform the measurement $\{\Pi_S,I-\Pi_S\}$ $O(1/\eps^2)$ times, where $\Pi_S$ is the projector onto $S$, and accept if and only if the first outcome is obtained every time. For example, this allows the property {\bf Permutation Invariance} to be tested efficiently, where $\ket{\psi} \in (\C^d)^{\otimes n}$ satisfies the property if it is invariant under any permutation of the $n$ subsystems. As $\ket{\psi}$ is permutation-invariant if and only if it is contained in the symmetric subspace of $(\C^d)^{\otimes n}$, projecting onto this subspace gives an efficient test for this property. This procedure, which is known as \emph{symmetrization}, has been studied in the context of quantum fault-tolerance and can be performed efficiently~\cite{barenco97}; see Section~\ref{sec:unitinv} below for a description of how this can be achieved via the powerful primitive of generalized phase estimation.

Another immediate generalization of {\bf Equality to $\ket{\phi}$} is the question of testing whether two \emph{unknown} states are the same. We say that a pair of states $\ket{\psi}$, $\ket{\phi}$ satisfies the {\bf Equality} property if $\ket{\phi} = e^{i \theta} \ket{\psi}$ for some real $\theta$, so now the property is $\Prop=\{(\ket{\psi},e^{i \theta} \ket{\psi}): \ket{\phi}\text{ is a pure state},\theta\in\R\}$. In order to test this property, we will use a simple but important procedure known as the \emph{swap test}. This was used by Buhrman et al.~\cite{buhrman01} to demonstrate an exponential separation between the quantum and classical models of simultaneous message passing (SMP) communication complexity, and has since become a standard tool in quantum algorithm design. In the test, we take two (possibly mixed\footnote{See Section~\ref{sec:mixed} for more about mixed states and a formal definition.}) states $\rho$, $\sigma$ as input and attach an ancilla qubit in state $\ket{0}$. We then apply a Hadamard gate to the ancilla, followed by a controlled-SWAP gate (controlled on the ancilla), and another Hadamard gate. We then measure the ancilla qubit and accept if the answer is 0. This procedure is illustrated by the circuit in Figure~\ref{fig:swaptest}.
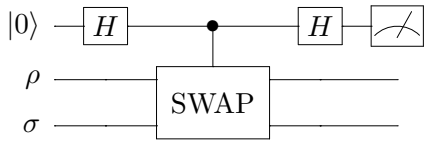
\begin{figure}[hbt]
\[
\Qcircuit @C=1em @R=.7em {
\lstick{\ket{0}} & \gate{H} & \ctrl{1} & \gate{H} & \meter \\
\lstick{\rho} & \qw & \multigate{1}{\text{SWAP}} & \qw & \qw\\
\lstick{\sigma} & \qw & \ghost{\text{SWAP}} & \qw & \qw
}
\]
\caption{The swap test.}
\label{fig:swaptest}
\end{figure}

One can show~\cite{buhrman01,kobayashi03} that the swap test accepts with probability
\[ \dfrac{1}{2} + \dfrac{1}{2}\Tr(\rho\,\sigma), \]
which for pure states $\ket{\psi}$, $\ket{\phi}$ is equal to $(1 + |\ip{\psi}{\phi}|^2)/2 = 1-D(\ket{\psi},\ket{\phi})^2/2$. In particular, if this test is applied to two pure states which satisfy the {\bf Equality} property then the test accepts with certainty. On the other hand, if the states are $\eps$-far away from equal, then by definition
\[ \inf_{\ket{\xi}} D(\ket{\psi}\ket{\phi},\ket{\xi}^{\otimes 2}) \ge \eps. \]
But
\[ \inf_{\ket{\xi}} D(\ket{\psi}\ket{\phi},\ket{\xi}^{\otimes 2}) = \sqrt{1-\sup_{\ket{\xi}} |\ip{\psi}{\xi}\ip{\phi}{\xi}|^2} \le \sqrt{1- |\ip{\psi}{\phi}|^2} = D(\ket{\psi},\ket{\phi}), \]
where the inequality follows by taking $\ket{\xi}=\ket{\phi}$.
Thus the test rejects with probability at least $\eps^2/2$, so $O(1/\eps^2)$ repetitions suffice to detect states $\eps$-far away from equal with constant probability; in other words, {\bf Equality} can be tested with $O(1/\eps^2)$ copies. The swap test is in fact optimal among all testers for this property which have perfect completeness and use one copy of each of the input states. To see this, observe that the swap test is precisely the operation of projecting onto the symmetric subspace of $(\C^d)^{\otimes 2}$. Any tester which accepts every pair of equal states $\ket{\psi}^{\otimes 2}$ must accept every state in this subspace, so the swap test is the most refined test of this type. One can generalize this to prove that the swap test is also optimal among tests which are allowed two-sided error, in the sense that it achieves the largest possible gap between the acceptance probabilities in equal and orthogonal instances~\cite{kada08}.

The property of {\bf Equality} can be generalized further, to the question of testing whether $n$ pure states $\ket{\psi_1},\dots,\ket{\psi_n}$ are all equal. The natural tester for this property, generalizing the swap test, is to project onto the symmetric subspace of $(\C^d)^{\otimes n}$, i.e., to perform symmetrization~\cite{barenco97}. Kada et al.~\cite{kada08} have studied this procedure under the name of the \emph{permutation test}, and show that the test accepts $n$-tuples where at least one pair of states is orthogonal with probability at most $1/n$, and that this is optimal among tests with perfect completeness. No explicit bounds appear to be known on this tester's parameters if the promise is relaxed, for example to specify that at least one pair of states has overlap at most $\eps$. Kada et al.\ also study a related tester, called the \emph{circle test}, and prove that this tester is also optimal for prime $n$~\cite{kada08}. This procedure is somewhat simpler as it only involves taking a quantum Fourier transform over $\Z_n$, rather than $S_n$.


\subsubsection{Productness}
\label{sec:productness}

A pure state $\ket{\psi} \in (\C^d)^{\otimes n}$ of $n$ $d$-dimensional subsystems is said to be \emph{product} (i.e., satisfy the {\bf Product} property) if it can be written as a tensor product $\ket{\psi} = \ket{\psi_1} \ket{\psi_2} \dots \ket{\psi_n}$ for some local states $\ket{\psi_1},\dots,\ket{\psi_n} \in \C^d$. A state which is not product is called \emph{entangled}. Entanglement is a ubiquitous phenomenon in quantum information theory (see for example~\cite{horodecki09} for an extensive review), so the property of being a product state is an obvious target to test.

Given just one copy of $\ket{\psi}$, our ability to test whether it is product is very limited. Indeed, as every quantum state can be written as a linear combination of product states, any tester which accepts all product states with certainty must accept all states with certainty. However, if we are given two copies of $\ket{\psi}$, there are non-trivial tests we can perform. In particular, consider the following procedure, which was first discussed by Mintert et al.~\cite{mintert05} and is called the \emph{product test}~\cite{harrow13}: apply the swap test across each corresponding pair of subsystems of $\ket{\psi}^{\otimes 2}$, and accept if and only if all of the tests accept. The overall procedure is illustrated in Figure \ref{fig:prodtest}.

\begin{figure}[hbt]
\begin{center}
\begin{tikzpicture}[scale=1.25]

\filldraw[fill=gray!10,rounded corners] (-0.6,-0.4) rectangle (4.6,0.4);
\filldraw[fill=gray!10,rounded corners] (-0.6,0.6) rectangle (4.6,1.4);

\draw[rounded corners] (-0.4,-0.6) rectangle (0.4,1.6);
\draw[rounded corners] (0.6,-0.6) rectangle (1.4,1.6);
\draw[rounded corners] (1.6,-0.6) rectangle (2.4,1.6);
\draw[rounded corners] (3.6,-0.6) rectangle (4.4,1.6);

\foreach \x in {1,...,3} {
  \filldraw[fill=gray!75] (\x-1,0) node {\x} circle (0.25);
  \filldraw[fill=gray!50] (\x-1,1) node {\x} circle (0.25);
}
\node at (3,0) {\Huge ...};
\node at (3,1) {\Huge ...};
\filldraw[fill=gray!75] (4,0) node {$n$} circle (0.25);
\filldraw[fill=gray!50] (4,1) node {$n$} circle (0.25);
\node[anchor=east] at (-0.8,1) {$\ket{\psi_1}$};
\node[anchor=east] at (-0.8,0) {$\ket{\psi_2}$};
\end{tikzpicture}

\caption{Schematic of the product test applied to an $n$-partite state $\ket{\psi}$. The swap test (vertical boxes) is applied to the $n$ pairs of corresponding subsystems of two copies of $\ket{\psi}$ (horizontal boxes).}

\label{fig:prodtest}
\end{center}
\end{figure}
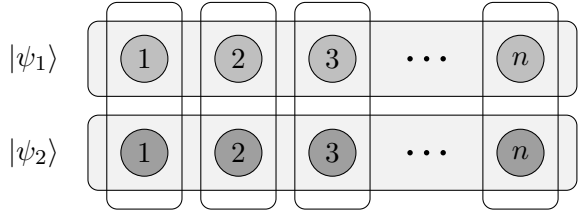

If $\ket{\psi}$ is indeed product, then all of the swap tests will accept. On the other hand, if $\ket{\psi}$ is far from product, the intuition is that the entanglement in $\ket{\psi}$ will cause at least some of the tests to reject with fairly high probability. This intuition can be formalized to give the following result.

\begin{thm}[Harrow and Montanaro~\cite{harrow13}]
\label{thm:producttest}
If $\ket{\psi}$ is $\eps$-far from product, the product test rejects with probability $\Omega(\eps^2)$.
\end{thm}

Thus the property of productness can be tested with $O(1/\eps^2)$ copies. We will not give the full, and somewhat technical, proof of Theorem~\ref{thm:producttest} here, but merely sketch the proof technique; see~\cite{harrow13} for details. 

\begin{proof}[Proof sketch]
Let $P_{\text{test}}(\ket{\psi})$ denote the probability of the product test accepting when applied to two copies of $\ket{\psi}$, and let the distance of $\ket{\psi}$ from the nearest product state be $\eps$. The proof is split into two parts, depending on whether $\eps$ is low or high. For $S \subseteq [n]$, let $\psi_S$ be the mixed state obtained by tracing out (discarding) the qubits not in $S$. Then the starting point is the observation that
\be \label{eq:avgpurity} P_{\text{test}}(\ket{\psi}) = \frac{1}{2^n} \sum_{S \subseteq [n]} \Tr(\psi_S^2). \ee
The quantity $\Tr(\psi_S^2)$ measures the purity of the reduced state $\psi_S$, which can be seen as a measure of the entanglement of $\ket{\psi}$ across the bipartition $(S,S^c)$; if $\ket{\psi}$ were product across this bipartition, $\psi_S$ would be pure and $\Tr(\psi_S^2)$ would equal 1. By (\ref{eq:avgpurity}), the probability that the test passes is equal to the average purity of the reduced state obtained by a random bipartition of the $n$ systems. Writing $\ket{\psi} = \sqrt{1-\eps^2} \ket{0^n} + \eps \ket{\phi}$ (without loss of generality), for some product state $\ket{0^n}$ and arbitrary orthogonal state $\ket{\phi}$, Eq.~(\ref{eq:avgpurity}) allows an explicit expression for $\Tr(\psi_S^2)$ in terms of $\eps$ and $\ket{\phi}$ to be obtained. Expanding $\ket{\phi} = \sum_{x \in \{0,\dots,d-1\}^n} \alpha_x \ket{x}$ and summing over $S$, we get an expression containing terms of the form $\sum_{x \in \{0,\dots,d-1\}^n} |\alpha_x|^2 c^{|x|}$ for some $c<1$, where $|x|:= |\{i:x_i \neq 0\}|$. In order to obtain a non-trivial bound from this, the final step of the first part of the proof is to use the fact that $\ket{0^n}$ is the closest product state to $\ket{\psi}$ to argue that $\ket{\phi}$ cannot have any amplitude on basis states $\ket{x}$ such that $|x| \le 1$. A bound is eventually obtained that is applicable when $\eps$ is small, namely that
\[ P_{\text{test}}(\ket{\psi}) \le 1 - \eps^2 + \eps^3 + \eps^4. \]
In the case where $\eps$ is large, this does not yet give a useful upper bound, so the second part of the proof finds a constant upper bound on $P_{\text{test}}(\ket{\psi})$. This quantity can be shown to be upper bounded by the probability that a relaxed test, for being product across some partition of the $n$ subsystems into $k \le n$ parties, passes. If $\ket{\psi}$ is far from product across the $n$ subsystems, the proof shows that one can find a partition into $k$ parties (for some $k \le n$) such that the distance from the closest product state (with respect to this partition) falls into the regime where the first part of the proof works. The eventual result is that if $\eps^2 \ge 11/32 > 0.343$, then $P_{\text{test}}(\ket{\psi}) \le 501/512 < 0.979$; combining these two bounds completes the proof.
\end{proof}

We mention two implications of Theorem \ref{thm:producttest}. First, by the characterization (\ref{eq:avgpurity}), the content of Theorem \ref{thm:producttest} can be understood as: if a pure state of $n$ systems is still fairly pure on average after discarding a random subset of the systems, it must in fact have been close to a product state in the first place. In the classical property testing literature, one of the motivations for analysing tests for combinatorial properties is to obtain some insight into the structure of the property being tested; Theorem \ref{thm:producttest} can be seen as achieving something similar in a quantum setting.

Second, by allowing one to efficiently certify productness given two copies of $\ket{\psi}$, the product test can be used to show that quantum Merlin-Arthur proof systems with multiple provers can be simulated efficiently by two provers, or in complexity-theoretic terminology that $\QMA(k)=\QMA(2)$~\cite{harrow13}.
Roughly speaking, to simulate a $k$-Merlin protocol, one can simply ask two Merlins to provide identical copies of the $k$-Merlin proofs, and perform the product test to ensure that they are indeed product states. Since the product test uses only two copies of the state, two Merlins suffice.
Via a previous result of Aaronson et al.~\cite{aaronson09} giving a multiple-prover quantum proof system for 3-SAT, this in turn allows one to prove hardness of various tasks in quantum information theory, conditioned on the hardness of 3-SAT~\cite{harrow13}. This is again analogous to the classical literature, where efficient property testers are used as components in hardness-of-approximation results.

Although the product test itself is natural, the detailed proof of Theorem~\ref{thm:producttest} given in~\cite{harrow13} is a lengthy case analysis which does not provide much intuition and gives suboptimal constants. For example, the lower bound obtained on the probability of the product test rejecting does not increase monotonically with $\eps$, which presumably should be the case for an optimal bound. We therefore highlight the following open question.

\begin{question}
Can the analysis of the product test be improved?
\end{question}


\subsubsection{Arbitrary finite sets}

The following algorithm of Wang~\cite{wang11} gives a tester for \emph{any} finite property $\Prop \subset \C^d$ (this is similar to the result for any finite \emph{classical} property mentioned at the end of Section~\ref{sec:bvtester}). The tester cannot necessarily be implemented time-efficiently in general. Given access to copies of an input state $\ket{\psi}$, the tester proceeds as follows:
\begin{enumerate}
\item Create the state $\ket{\psi}^{\otimes T}$, for some $T$ to be determined.
\item Let $S = \spann \{ \ket{\phi}^{\otimes T}:\ket{\phi} \in \Prop\}$. Perform the measurement $\{\Pi_S,I-\Pi_S\}$, where $\Pi_S$ is the projector onto $S$, and accept if the first outcome is obtained. Otherwise, reject.
\end{enumerate}

\begin{thm}[Wang~\cite{wang11}]
\label{thm:arbsets}
Let $\Prop \subset \C^d$ be such that $\min_{\ket{\phi}\neq\ket{\phi'}\in\Prop} D(\ket{\phi},\ket{\phi'}) = \delta$. Then it suffices to take $T = O(\log |\Prop| \max\{ \eps^{-2}, \delta^{-2} \})$ to obtain a tester which accepts every state in $\Prop$ with certainty, and rejects every state $\ket{\psi}$ such that $D(\ket{\psi},\Prop)\ge\eps$ with probability at least $2/3$.
\end{thm}

Observe that the dependence on $|\Prop|$ is only logarithmic. The intuition behind Theorem \ref{thm:arbsets} is that, if all the states in $\Prop$ have large pairwise distances, $\{\ket{\phi}^{\otimes T}\}$ is an approximately orthonormal basis for $S$, so if $\ket{\psi}$ is $\eps$-far from $\Prop$, the probability of incorrectly accepting is
\[ \bra{\psi}^{\otimes T} \Pi_S \ket{\psi}^{\otimes T} \approx \sum_{\ket{\phi} \in \Prop} |\ip{\psi}{\phi}|^{2T} \le |\Prop| (1-\eps^2)^T, \]
which is sufficiently small when $T=O((\log|\Prop|)/\eps^2)$.
Wang describes an application of Theorem~\ref{thm:arbsets} to testing the set of permutations of $n$ qubits using $O((n \log n) / \eps^2)$ copies~\cite{wang11}. However, the dependence of the theorem on $\delta$ seems to limit its applicability. It is an interesting question whether this dependence can be improved or removed, either by better analysis of the above tester or by designing a new tester.

\begin{question}
Does there exist a tester for arbitrary finite properties $\Prop \subset \C^d$ which uses $\polylog |\Prop|$ copies, and whose parameters have no dependence on $\min_{\ket{\phi}\neq\ket{\phi'}\in\Prop} D(\ket{\phi},\ket{\phi'})$?
\end{question}

The above tester is a general algorithm for testing any property $\Prop$. For some properties $\Prop$ it is possible to prove better bounds on the performance of this algorithm than Theorem \ref{thm:arbsets} would give, or prove bounds with fewer preconditions. For example, the product test is a particular case of this algorithm (with $T=2$), and Theorem \ref{thm:producttest} gives non-trivial bounds on its performance, even though it is applied to the infinite set of product states. We also remark that an alternative algorithm to the above tester would be to produce $\ket{\psi}^{\otimes T}$, and for each $\ket{\phi} \in \Prop$ in turn, perform the measurement $\{\proj{\phi}^{\otimes T},I-\proj{\phi}^{\otimes T}\}$, and accept if and only if the first outcome is obtained from any measurement. This algorithm would achieve similar scaling in terms of $\eps$ and $\delta$ (as can be shown using the ``quantum union bound'' of Aaronson~\cite{aaronson06}, or the ``gentle measurement lemma'' of Winter~\cite{winter99} and Ogawa and Nagaoka~\cite{ogawa02}) but would not have perfect completeness.


\subsubsection{Open questions}

There are a number of interesting sets of pure states for which an efficient tester is not known. One such set is the \emph{stabilizer states}. Recall that the Pauli matrices on one qubit are defined to be the set
\begin{equation*}
I = \mm{1&0\\0&1},\;X=\mm{0&1\\1&0},\;Y=\mm{0&-i\\i&0},\;Z=\mm{1&0\\0&-1}. 
\end{equation*}
They form a basis for the space of single-qubit linear operators, and by tensoring form a basis for the space of linear operators on $n$ qubits; for $s\in \{I,X,Y,Z\}^n$, we write $\sigma_s$ for the corresponding operator on $n$ qubits. We call each such tensor product operator a ($n$-qubit) Pauli matrix, and use $\Prop_n$ to denote the set of all $n$-qubit Pauli matrices, together with phases $\pm1$, $\pm i$, which forms a group under multiplication.

A state $\ket{\psi}$ of $n$ qubits is said to be a stabilizer state if there exists a maximal Abelian subgroup $G$ of $\Prop_n$ such that $U\ket{\psi} = \ket{\psi}$ for all $U \in G$. Stabilizer states are important in the study of quantum error-correction~\cite{gottesman99} and measurement-based quantum computation~\cite{raussendorf03}, as well as many other areas of quantum information. It is known that, given access to copies of an unknown stabilizer state $\ket{\psi}$ of $n$ qubits, $\ket{\psi}$ can be learned with $O(n)$ copies~\cite{aaronson08}; there is a matching $\Omega(n)$ lower bound following from an information-theoretic argument~\cite{holevo73}. However, it might be possible to \emph{test} whether $\ket{\psi}$ is a stabilizer state using far fewer copies.

\begin{question}
Is there a tester for the property of being a stabilizer state whose parameters do not depend on the number of qubits $n$?
\end{question}

Other sets of pure states for which it would be interesting to have an efficient tester are matrix product states (see, e.g., \cite{perezgarcia07}) and states of low Schmidt rank, or with low complexity with respect to some other entanglement measure~\cite{guhne09}. See Section~\ref{sec:mixed} below for evidence for a lower bound on the complexity of testing the Schmidt rank.

Another interesting, and as yet largely unexplored, direction for future research is testing properties of quantum states in a distributed setting. Here we imagine that two parties, Alice and Bob, each hold part of one copy of a large unknown state $\ket{\psi}$. Their goal is to determine whether $\ket{\psi}$ satisfies some property while exchanging only a small number of qubits; in particular, Alice cannot just send her half of the state to Bob. Our normal complexity measure ``number of copies consumed'' is thus replaced with ``number of qubits sent.'' Aharonov et al.~\cite{aharonov14} recently showed that the $d$-dimensional maximally entangled state $\frac{1}{\sqrt{d}} \sum_{i=1}^d \ket{i}\ket{i}$ can be tested up to accuracy $\eps$ by communicating only $O(\log 1/\eps)$ qubits. There are many other properties where the question of existence of communication-efficient testers remains open.


\subsection{Mixed states}
\label{sec:mixed}

A \emph{mixed} state $\rho$ is a convex combination of pure states. Mixed states are described by density matrices, which are positive semidefinite matrices with unit trace; we let $\mathcal{B}(\C^d)$ denote the set of $d$-dimensional density matrices. The concept of property testing can easily be generalized from pure states to mixed states. We retain the same, natural distance measure
\[ D(\rho,\sigma) := \frac{1}{2} \| \rho - \sigma \|_1, \]
which is called the trace distance between $\rho$ and $\sigma$. Note that for classical probability distributions (i.e., diagonal density matrices) this is just the total variation distance. As before, say that $\rho$ is \emph{$\eps$-far} from having property $\Prop \subseteq \mathcal{B}(\C^d)$ if
\[ D(\rho, \Prop) := \inf_{\sigma \in \Prop} D(\rho,\sigma) \ge \eps, \]
and \emph{$\eps$-close} to having property $\Prop$ if $D(\rho, \Prop) \le \eps$. Another important distance measure for mixed states is the \emph{fidelity}, which is defined as $F(\rho,\sigma) := \|\sqrt{\rho}\sqrt{\sigma}\|_1$, where $\sqrt{\rho}$ denotes the positive semidefinite square root of the operator $\rho$. For any mixed state $\rho$ and pure state $\ket{\psi}$, $F(\rho,\proj{\psi}) = \sqrt{\bracket{\psi}{\rho}{\psi}}$. The fidelity and trace distance are related by the inequalities~\cite[Eq.~9.110]{nielsen&chuang:qc}
\be \label{eq:fvsd} 1 - F(\rho,\sigma) \le D(\rho,\sigma) \le \sqrt{1- F(\rho,\sigma)^2}. \ee
In a mixed-state property testing scenario, we are given $k$ copies of $\rho$, for some unknown $\rho$, and asked to perform a measurement on $\rho^{\otimes k}$ to determine whether $\rho \in \Prop$, or $\rho$ is $\eps$-far away from $\Prop$. 

Similarly to the case of pure states, \emph{any} property $\Prop \subseteq \mathcal{B}(\C^d)$ can be tested with $O(d^4/\eps^2)$ copies. To distinguish between the two cases that $\rho \in \Prop$ or $\rho$ is $\eps$-far from $\Prop$, it suffices to use an estimate $\widetilde{\rho}$ such that $D(\widetilde{\rho},\rho) < \eps/2$, and accept if and only if $D(\widetilde{\rho},\Prop) \le \eps/2$. Producing such an estimate can be achieved using quantum state tomography~\cite{paris04,nielsen&chuang:qc}; in order to achieve the required accuracy with success probability $2/3$, $O(d^4/\eps^2)$ copies suffice~\cite{flammia12}. If one knows in advance that $\rho$ is rank $r$, compressed sensing techniques allow this bound to be improved to $O((rd/\eps)^2 \log d)$~\cite{flammia12}.

Some properties of mixed states can be tested significantly more efficiently than this general upper bound. A simple example is the property {\bf Purity}, where $\rho$ satisfies the property if and only if it is a pure state. A natural way to test purity is to apply the swap test (Figure~\ref{fig:swaptest}) to two copies of $\rho$. This accepts with probability $(1+\Tr(\rho^2))/2$, which is equal to 1 if and only if $\rho$ is pure. On the other hand, if we let $\rho = \sum_i \lambda_i \proj{\psi_i}$ be the eigendecomposition of $\rho$, where eigenvalues are listed in non-increasing order, a closest pure state to $\rho$ is $\ket{\psi_1}$. If $\rho$ is $\eps$-far away from pure, then $\lambda_1 \le 1 - \eps$. Note that 
\[
\Tr(\rho^2)=\sum_i\lambda_i^2 \leq \max_i \lambda_i \sum_j \lambda_j = \lambda_1.
\] 
Thus the test accepts with probability at most $1 - \eps/2$, implying that {\bf Purity} can be tested with $O(1/\eps)$ copies of $\rho$.

On the other hand, consider the ``dual'' property of {\bf Mixedness}, where $\rho \in \mathcal{B}(\C^d)$ satisfies the property if and only if it is the maximally mixed state $I/d$. A strong lower bound has been shown by Childs et al.~\cite{childs07c} on the number of copies required to test this property.

\begin{thm}[Childs et al.~\cite{childs07c}]
\label{thm:qcoll}
Let $d$ and $r$ be integers such that $r$ strictly divides $d$. Any algorithm which distinguishes, with probability of success at least $2/3$, between the two cases that $\rho = I/d$ or $\rho$ is maximally mixed on a uniformly random subspace of dimension $r$ must use $\Omega(r)$ copies of $\rho$. Further, there exists an algorithm which solves this problem using $O(r)$ copies.
\end{thm}

Childs et al.\ call the problem which they consider the \emph{quantum collision problem}. To see how their result can be applied to {\bf Mixedness}, consider the space of $n$ qubits, whose dimension is~$d=2^n$. As a state $\rho$ which is maximally mixed on a dimension-$r$ subspace of $\C^{2^n}$ satisfies $D(\rho,I/2^n) = 1-r/2^n$, taking $r=2^{n-1}$ implies that any algorithm distinguishing between the cases that $\rho=I/2^n$ and $\rho$ is $1/2$-far from $I/2^n$ must use $\Omega(2^n)$ copies of $\rho$.\footnote{Very recently, O'Donnell and Wright~\cite{odonnell&wright:spectrum} strengthened this result, among other things they obtained a tight dependence on~$\eps$.}  This result also puts strong lower bounds on a number of other property testing problems which one might wish to solve. For example, consider the following three properties:
\begin{itemize}
\item {\bf Equality} of pairs of mixed states, where the pair $(\rho, \sigma)$ satisfies the property if $\rho = \sigma$. This can be seen as the quantum generalization of the classical question of testing whether two probability distributions on $d$ elements are equal or $\eps$-far from equal (with respect to the total variation distance), given access to samples from the distributions. A sublinear tester for the classical problem has been given by Batu et al.~\cite{batu13}, and recently improved by Chan et al.~\cite{chan13}; for constant $\eps$ the tester uses $O(d^{2/3})$ samples. By fixing $\sigma=I/d$, the result of~\cite{childs07c} implies that the quantum generalization of this problem is more difficult: it requires at least $\Omega(d)$ ``samples'' (i.e., copies of the states).

\item Whether a mixed state $\rho$ has {\bf rank at most $r$}. Theorem \ref{thm:qcoll} immediately implies that this requires $\Omega(r)$ copies of $\rho$, which has an interesting implication for testing pure states. Recall that a bipartite state $\ket{\psi}$ on systems $AB$ is said to have Schmidt rank $r$ if it can be written as $\ket{\psi} = \sum_{i=1}^r \sqrt{\lambda_i} \ket{v_i}\ket{w_i}$ for pairwise orthonormal sets of states $\{\ket{v_i}\}$, $\{\ket{w_i}\}$ and non-negative $\lambda_i$. If one looks only at the $A$ subsystem, the rank of the reduced state is precisely the Schmidt rank of $\ket{\psi}$. Therefore, Theorem \ref{thm:qcoll} implies that any algorithm which tests whether a pure state $\ket{\psi}$ has Schmidt rank $r$ by producing $k$ copies of $\ket{\psi}$ and acting only on the first subsystems $A_1,\dots,A_k$ of $\ket{\psi}^{\otimes k}$ must satisfy $k = \Omega(r)$. This bound does not apply immediately to general algorithms acting on both the $A$ and $B$ subsystems, leaving the complexity of testing the Schmidt rank open.

\item {\bf Separability} of mixed states. A bipartite quantum state $\rho \in \mathcal{B}((\C^d)^{\otimes 2})$ is said to be \emph{separable} if it can be written as a convex combination of product states, and is said to be \emph{entangled} otherwise. Given a classical description of a $d$-dimensional mixed state as input, determining separability up to accuracy which is inversely polynomial in $d$ is known to be $\mathsf{NP}$-hard~\cite{gurvits03,gharibian10}, and there is some evidence for intractability of the problem even up to constant accuracy~\cite{harrow13}. This does not preclude the existence of a tester for separability which is efficient in terms of the number of copies of the input state $\rho$ used; however, Theorem \ref{thm:qcoll} can be used to show that such a tester cannot exist.

The idea is to show that the maximally mixed state on a random subspace of dimension $r$ is far from separable, if $r$ is picked suitably. This can be achieved by combining some previously known results. The \emph{entanglement of formation} of a bipartite state $\rho$ on systems $AB$, is defined by
\[ E_F(\rho) = \min_{\sum_i p_i \proj{\psi_i} = \rho} \sum_i p_i\,S(\Tr_B(\proj{\psi_i})), \]
where $S(\rho) = -\Tr(\rho \log_2 \rho)$ is the von Neumann entropy. Of course, if $\rho$ is separable, $E_F(\rho)=0$. Let $\rho$ be the maximally mixed state on a random subspace of $\C^d \otimes \C^d$ of dimension $r = \lfloor cd^2 \rfloor$, for some fixed $c\in(0,1)$. Hayden et al.~\cite{hayden06} have shown that, for small enough $c>0$, there exists a universal constant $C > 0$ such that $E_F(\rho) \ge C \log_2 d$, except with probability exponentially small in~$d$. Also, Nielsen~\cite{nielsen00b} has shown a continuity property for the entanglement of formation:
\[ E_F(\rho) - E_F(\sigma) \le 18 (\log_2 d) \sqrt{1-F(\rho,\sigma)} + 2(\log_2 e)/e. \]
Combining these two properties, and relating the fidelity to the trace distance using (\ref{eq:fvsd}), we have that $\rho$ is distance $\Omega(1)$ from the set of separable states with high probability. On the other hand, the maximally mixed state $I/d^2$ is clearly separable. Therefore, any tester which distinguishes separable states from states a constant distance from any separable state can be used to distinguish the maximally mixed state from a random dimension-$r$ subspace; by Theorem~\ref{thm:qcoll}, this task requires $\Omega(r)=\Omega(d^2)$ copies of the input state.
\end{itemize}

We remark that the theory of \emph{entanglement witnesses} takes an alternative approach to the direct detection of entanglement (see for example \cite{guhne09,horodecki09} for extensive reviews). An entanglement witness for a state $\rho$ is an observable corresponding to a hyperplane separating $\rho$ from the convex set of separable states; measuring the observable allows one to certify that $\rho$ is entangled. Each such witness will only be useful for certain entangled states, however, so this approach does not provide a means of certifying entanglement of a completely unknown state $\rho$.


There is a gap between the best known lower and upper bounds for testing the above three properties. We therefore highlight the following open question:

\begin{question}
What is the complexity of testing {\bf Equality}, {\bf Separability}, and {\bf Rank at most $r$}?
\end{question}


\subsubsection{Testing equality to a fixed pure state}

We have seen that testing whether $\rho \in \mathcal{B}(\C^d)$ is the maximally mixed state $I/d$ can require $\Omega(d)$ copies of $\rho$. By contrast, testing whether $\rho$ is a fixed \emph{pure} state $\proj{\psi}$ is easy: the obvious test is to perform the measurement $\{\proj{\psi}, I-\proj{\psi}\}$, and to accept if the first outcome is returned. The probability of acceptance is $\bracket{\psi}{\rho}{\psi}$, which is upper bounded by $1-D(\rho,\proj{\psi})^2$ by (\ref{eq:fvsd}), so this property can be tested with $O(1/\eps^2)$ copies of $\rho$.

However, there is a more interesting related question of relevance to experimentalists. Imagine we have some experimental apparatus which is claimed to produce a state $\ket{\phi}$ of $n$ qubits, and we would like to certify this fact. In this setting, the above test does not seem to make sense; being able to measure $\ket{\phi}$ is essentially precisely what we wish to certify! We further imagine that $n$ is too large for full state tomography to be efficient. In order to solve this self-certification problem, we would therefore like a procedure which makes a small number of measurements, can easily be implemented experimentally, and certifies that the state produced is approximately equal to $\ket{\phi}$. This question has been considered by da Silva et al.~\cite{dasilva11}, and independently Flammia and Liu~\cite{flammia11}, who show that certain states $\ket{\phi}$ can be certified using significantly fewer copies of $\ket{\phi}$ than would be required for full tomography, and indeed that any state $\ket{\phi}$ can be certified using quadratically fewer copies ($O(2^n)$ rather than $O(2^{2n})$). The measurements used are also simple: Pauli measurements.

The Pauli matrices $\{\sigma_s\}$ on $n$ qubits form a basis for the space of $n$-qubit linear operators and satisfy $\Tr(\sigma_s \sigma_t) = 2^n \delta_{st}$. So any state $\rho \in \mathcal{B}(\C^{2^n})$ can be expanded as
\[ \rho = \sum_{s \in \{I,X,Y,Z\}^n} \widehat{\rho}_s \sigma_s \]
for some real coefficients $\widehat{\rho}_s = \Tr(\rho \sigma_s)/2^n$. Writing $\phi := \proj{\phi}$ for conciseness, the squared fidelity between $\ket{\phi}$ and $\rho$ is
\[ \bracket{\phi}{\rho}{\phi} = \Tr(\rho \phi)= 2^n \sum_{s \in \{I,X,Y,Z\}^n} \widehat{\rho}_s \widehat{\phi}_s. \]
The works~\cite{dasilva11,flammia11} propose the following scheme. First, pick $s \in \{I,X,Y,Z\}^n$ with probability $2^n \widehat{\phi}_s^2$; orthonormality of the Pauli matrices implies that this is indeed a valid probability distribution. Then repeatedly measure copies of $\rho$ in the eigenbasis of $\sigma_s$, and take the average of the eigenvalues corresponding to the measurement results to produce an estimate $\widetilde{\rho}_s$ of $2^n\widehat{\rho}_s = \Tr(\rho \sigma_s)$.  Finally, output $\widetilde{\rho}_s / \widehat{\phi}_s$ as our guess for the squared fidelity. The expectation of $\widetilde{\rho}_s$ is precisely $\Tr(\rho \sigma_s)$, and if we assume that this estimate is exact (i.e., $\widetilde{\rho}_s = \Tr(\rho \sigma_s)$), the expected value of the output is
\[ \sum_{s \in \{I,X,Y,Z\}^n} (2^n \widehat{\phi}_s^2) \frac{\widehat{\rho}_s}{\widehat{\phi}_s} = \Tr(\rho \phi). \]
Of course, in general we cannot produce an exact estimate without using an infinite number of copies of $\rho$. However, to estimate the fidelity up to constant additive error with constant success probability, it suffices to use a finite number of copies. The number of copies required turns out to depend on the quantity $\min_{s,\widehat{\phi}_s \neq 0} |\widehat{\phi}_s|$; for certain classes of states $\ket{\phi}$ (such as stabilizer states), the number of copies used does not depend on $n$.


\subsubsection{Unitarily invariant properties}
\label{sec:unitinv}

Generalizing the properties {\bf Purity} and {\bf Mixedness}, one can consider properties $\Prop$ of mixed quantum states which are unitarily invariant, in the following sense: If $\rho \in \Prop$, then $(U \rho U^\dag) \in \Prop$ for all $U \in U(d)$, where $U(d)$ denotes the unitary group in $d$ dimensions. Observe that this implies that, if $\rho$ is $\eps$-far from $\Prop$, then so is $U \rho U^\dag$, for all $\eps$ and all $U \in U(d)$. For any $\rho$, $D(\rho,\Prop)$ must necessarily be a symmetric function of the spectrum of $\rho$.

We can see unitarily invariant properties as quantum analogs of symmetric properties of classical probability distributions. Quite recently, it has been shown that a particular ``canonical'' classical tester is close to optimal for all such symmetric properties which satisfy certain continuity constraints~\cite{valiant11}. This has allowed strong bounds to be proven on the complexity of testing properties such as equality of probability distributions, and distinguishing high-entropy from low-entropy distributions. We now discuss an analogous ``canonical tester'' for unitarily invariant properties.

In order to take advantage of the unitary symmetry, one can use a concept known as \emph{Schur-Weyl duality}. We will only briefly summarize this beautiful theory here, and sketch the consequences for property testing; for much more detailed introductions, see the theses~\cite{christandl06,harrow05}. Schur-Weyl duality implies that any linear operator $M$ on $(\C^d)^{\otimes k}$ which commutes with permutations of the $k$ subsystems, and also with local unitaries on each subsystem (i.e., $U^{\otimes k} M (U^{-1})^{\otimes k} = M$ for all $U \in U(d)$) can be written as $M = \sum_{\lambda \vdash k} \alpha_\lambda P_\lambda$ for some coefficients $\alpha_\lambda$ and projectors $P_\lambda$, where the sum is over partitions $\lambda$ of $k$ (e.g., the partitions of 4 are $(4)$, $(3,1)$, $(2,2)$, $(2,1,1)$, $(1,1,1,1)$). Each partition $\lambda$ corresponds to an irreducible representation (irrep) of~$S_k$, the symmetric group on $k$ elements;  one important irrep is the trivial irrep $(k)$ which maps $\pi \mapsto 1$ for all $\pi \in S_k$. The operators $P_\lambda$ are defined by
\[ P_\lambda := \frac{d_\lambda}{k!} \sum_{\pi \in S_k} \chi_\lambda(\pi) U_\pi. \]
In the above expression, $d_\lambda$ is the dimension of the corresponding irrep $V_\lambda$ of $S_k$, which associates a $d_\lambda$-dimensional square matrix with each permutation $\pi \in S_k$. Then $\chi_\lambda$ is the corresponding character $\Tr(V_\lambda)$ and $U_\pi$ is the operator which acts by permuting $k$ $d$-dimensional systems according to $\pi$:
\[ U_\pi \ket{i_1}\dots\ket{i_k} = \ket{i_{\pi^{-1}(1)}} \dots \ket{i_{\pi^{-1}(k)}}. \]
One can show that each operator $P_\lambda$ is indeed a projector, that $P_\lambda P_\mu = \delta_{\lambda \mu} P_\lambda$, and that $\sum_{\lambda \vdash k} P_\lambda = I$. These operators therefore define a measurement (POVM), and performing this measurement is known as \emph{weak Schur sampling}~\cite{childs07c}. This can be implemented efficiently via a procedure which is known as generalized phase estimation~\cite{harrow05,childs07c} and generalizes the swap test~\cite{buhrman01} (cf.\ Section~\ref{sec:equality}) and symmetrization~\cite{barenco97}. Generalized phase estimation is based on the quantum Fourier transform (QFT) over $S_k$~\cite{beals97}, which is a unitary operation that performs a change of basis from $\{ \ket{\pi}: \pi \in S_k\}$ to $\{\ket{\lambda,i,j}:\lambda \vdash k,1 \le i,j \le d_\lambda\}$. It follows from basic representation theory that this makes sense, i.e., that $\sum_{\lambda\vdash k} d_\lambda^2 = k!$.

The generalized phase estimation procedure proceeds as follows:
\begin{enumerate}
\item Start with a quantum state $\sigma \in \mathcal{B}((\C^d)^{\otimes k})$.
\item Prepend a $k!$-dimensional ancilla register whose basis states correspond to triples $\ket{\lambda,i,j}$, initialized in the state $\ket{(k),1,1}$ corresponding to the trivial irrep.
\item Apply the inverse quantum Fourier transform over $S_k$ to the ancilla to produce the state $\frac{1}{\sqrt{k!}} \sum_{\pi \in S_k} \ket{\pi}$ (see, e.g.,~\cite{beals97} for an explanation of this).
\item Apply the controlled permutation operation $\sum_{\pi \in S_k} \proj{\pi} \otimes U_\pi$, controlled on the ancilla.
\item Apply the quantum Fourier transform over $S_k$ to the ancilla and measure it, receiving outcome $(\lambda,i,j)$.
\item Output $\lambda$.
\end{enumerate}
One can show~\cite{bacon06,harrow05} that, on input $\sigma$, generalized phase estimation does indeed output $\lambda$ with probability $\Tr(P_\lambda \sigma)$.\footnote{Some works describe the procedure as instead starting with a QFT and finishing with an inverse QFT~\cite{childs07c,montanaro09c}, but this does not appear correct as the QFT should map from the group algebra of $S_k$ to the space of irreps of $S_k$~\cite{beals97}.}

It turns out that any test for a unitarily invariant property can, essentially, be taken to consist of performing weak Schur sampling and classically post-processing the results. 

\begin{lem}
Let $\Prop \subseteq \mathcal{B}(\C^d)$ be a unitarily invariant property. Assume there exists a tester which uses $k$ copies of the input state $\rho$, and accepts all states $\rho \in \Prop$ with probability at least $1-\delta$, but accepts all states which are $\eps$-far from $\Prop$ with probability at most $1-f(\eps)$ for $\eps>0$. Then there exists a tester with the same parameters which consists of performing weak Schur sampling on $\rho^{\otimes k}$ and classically postprocessing the results.
\end{lem}

\begin{proof}
Let $M$ be the measurement operator corresponding to the tester accepting, and for each $\eps$, let $\rho_\eps$ be a state which is distance $\eps$ from $\Prop$ and achieves the worst-case probability of acceptance (so $\rho_0$ is a state in $\Prop$ with the lowest probability of acceptance, and for $\eps >0$, $\rho_\eps$ is a state with the highest probability of acceptance such that $D(\rho_\eps,\Prop)=\eps$). Then, by the permutation invariance of~$\rho_\eps^{\otimes k}$, we have
\[ 
\Tr(M \rho_\eps^{\otimes k}) = \frac{1}{k!} \sum_{\pi \in S_k} \Tr(M U_\pi \rho_\eps^{\otimes k} U_\pi^{-1} =: \Tr(\overline{M} \rho_\eps^{\otimes k}), 
\]
where we define $\overline{M} = \frac{1}{k!} \sum_{\pi \in S_k} U_\pi M U_\pi^{-1}$, and by the unitary invariance of $\Prop$,
\[
\Tr(\overline{M} \rho_0^{\otimes k}) \le \int \Tr(\overline{M} (U \rho_0 U^{-1})^{\otimes k})dU = \Tr\left( \int U^{\otimes k} \overline{M} (U^{-1})^{\otimes k}  dU\right) \rho_0 =: \Tr(\overline{\overline{M}} \rho_0),
\]
where the integral is taken according to Haar measure on $U(d)$, and similarly $\Tr(\overline{M} \rho_\eps^{\otimes k}) \ge \Tr(\overline{\overline{M}} \rho_\eps^{\otimes k})$ for $\eps >0$. Therefore, it suffices to implement $\overline{\overline{M}}$ to achieve the same parameters as $M$. But $\overline{\overline{M}}$ commutes with local unitaries and permutations of the $k$ systems, so by Schur-Weyl duality we can write $\overline{\overline{M}} = \sum_\lambda \alpha_\lambda P_\lambda$ for some coefficients $\alpha_\lambda$; as $\overline{\overline{M}}$ is a measurement operator, for each $\lambda$ it holds that $0 \le \alpha_\lambda \le 1$. So we can implement $\overline{\overline{M}}$ by performing weak Schur sampling, obtaining outcome~$\lambda$, and then accepting with probability $\alpha_\lambda$.
\end{proof}

Further, one can write down the probability of obtaining each outcome $\lambda$ as follows: if the input state $\rho$ has eigenvalues $(x_1,\dots,x_d)$, then
\[ \Tr(P_\lambda \rho^{\otimes k}) = d_{\lambda} s_{\lambda}(x_1,\dots,x_d), \]
where $s_\lambda$ is a Schur polynomial (see, e.g., \cite{audenaert06} for a discussion). In principle, this allows one to calculate the parameters of the optimal test for any unitarily invariant property; in practice, the calculations required are somewhat daunting. Nevertheless, a careful analysis of the output distributions resulting from weak Schur sampling was the approach taken by Childs et al.~\cite{childs07c} to prove their bounds on the quantum collision problem. Indeed, their approach is an example of how one can prove lower bounds on quantum property testers more generally: first use symmetry arguments to prove that the optimal test must be of a certain form, then analyse the optimal test directly.


\section{Quantum testing of quantum properties: Dynamics}\label{sec:dynamics}

\subsection{Unitary operators}
\label{sec:unitary}

In this section, we will consider quantum property testing of quantum dynamics, beginning with unitary dynamics. We will imagine we are given black-box access to a unitary operator $U$, and we want to test if $U$ either has a certain property or is far from having it, by applying $U$ a small number of times. This setting is more complicated than that of testing properties of quantum states in that, rather than simply performing a measurement on a number of copies of a state, we can consider more involved protocols based on the use of $U$ in a sequential, adaptive fashion.

There are a number of choices one needs to make when defining this model---in particular, what distance measure to use, and whether or not to allow applications of controlled-$U$ and/or $U^{-1}$ as part of the model. In Sections~\ref{sec:distmeasures} and~\ref{sec:controlledu} we will discuss the effect of these choices. Next, we will discuss a useful correspondence between quantum states and unitaries---the \emph{Choi-Jamio\l kowski isomorphism}---which allows one to apply many of the algorithms developed for testing quantum states to unitaries. Finally, in sections \ref{sec:paulitest}-\ref{sec:unitaryprops} we will describe several known results on testing various properties of unitary operators.

We continue to let $U(d)$ denote the unitary group in $d$ dimensions, and let $M(d)$ denote the set of $d \times d$ matrices. A property of unitary operators is simply a (discrete or continuous) subset $\Prop \subseteq U(d)$. 


\subsubsection{Distance measures}
\label{sec:distmeasures}

As compared with the case of pure states, it is less obvious which distance measure between unitary operators is the right one to choose to obtain interesting property testing results. For quantum states, the distinguishability of any two states is controlled by their trace distance. A natural way to generalize this to unitary operations would be to maximize the distinguishability of the output states over all input states\footnote{One might wonder whether distinguishability could be improved further by allowing the unknown unitary operator to act on part of an entangled state; it turns out that this is not the case~\cite{watrous08a}.}, to produce
\[ D^{\text{max}}(U,V) := \max_{\ket{\psi}} D(U\ket{\psi},V\ket{\psi}) = \max_{\ket{\psi}} \sqrt{1-|\bracket{\psi}{U^\dag V}{\psi}|^2}. \]
Unfortunately, there are extremely simple properties which are hard to test with respect to this distance measure. One such example is the {\bf Identity} property: does an input unitary $U$ satisfy $U = e^{i \theta} I$? (Note that, as with the case of pure state properties, we allow an arbitrary phase $\theta$ in the definition, as $U$ cannot be distinguished from $e^{i\theta} U$.) Consider the family of $n$-dimensional unitary operators $U_i$, $i \in [n]$, where $U_i \ket{j} = (-1)^{\delta_{ij}} \ket{j}$. Each of these has maximal distance from~$I$, according to the distance measure $D^{\text{max}}$. However, a quantum algorithm which uses the input operator $U$ $k$ times and distinguishes between the case where $U$ is equal to the identity, and the case where $U=U_i$ for some $i$, would imply a quantum algorithm which computes the OR function of $n$ input bits, promised to have Hamming weight at most 1, using $O(k)$ queries. As this problem is known to require $\Omega(\sqrt{n})$ quantum queries~\cite{bennett97}, it follows that $k=\Omega(\sqrt{n})$. This is a lower bound on the complexity of identity-testing in an oracular setting; we discuss a lower bound based on computational complexity arguments in Section~\ref{sec:compcomp}.

It is perhaps not surprising that $D^{\text{max}}$ is not the right measure of distance to choose for property testing problems, as it is a ``best-case'' rather than ``average-case'' measure. A suitable such alternative measure can be defined as follows. For any $d$-dimensional operators $A,B\in M(d)$, let $\hsip{A}{B}$ denote the \emph{normalized Hilbert-Schmidt} inner product
\[ \hsip{A}{B} := \frac{1}{d} \Tr(A^\dag B) = \frac{1}{d} \sum_{i,j} A_{ij}^* B_{ij}. \]
Assume that $\hsip{A}{A} = \hsip{B}{B}=1$ (a property satisfied, for example, if $A$ and $B$ are unitary). Then the distance between $A$ and $B$ is given by
\[ D(A,B) := \sqrt{1-|\hsip{A}{B}|^2}. \]
For $\Prop \subseteq U(d)$, we analogously define
\[ D(U,\Prop) := \inf_{V \in \Prop} D(U,V). \]
Note the close analogy to the distance between pure states (\ref{eq:tracedistance}). Indeed, we use the same notation as for the distance $D(\ket{\psi},\ket{\phi})$ to highlight the fact that the distance for unitaries is naturally induced by the distance for states. The distance measure $D(A,B)$ seems to have been first explicitly introduced by Low~\cite{low09}; Wang~\cite{wang11} has defined a closely related alternative measure as $D'(A,B) := \sqrt{1-|\hsip{A}{B}|}$. As $D(A,B)/\sqrt{2} \le D'(A,B) \le D(A,B)$, the two measures are essentially interchangeable. For any operators $A$ and $B$ such that $\hsip{A}{A}=\hsip{B}{B}=1$, $D(A,B)$ has the following properties.
\begin{itemize}
\item $0 \le D(A,B) \le 1$, with $D(A,B)=0$ if and only if $A = e^{i \phi} B$ for some overall phase $\phi$. As there exist $A \neq B$ with $D(A,B) = 0$, this implies that $D(\cdot,\cdot)$ is not a metric, but only a ``pseudometric.'' Further, $D(A,B) = D(WA,WB) = D(AW,BW)$ for any unitary $W$.

\item $D(A,B)$ can alternatively be defined as
\[ D(A,B) = \frac{1}{\sqrt{2}}\|A \otimes A^\dag - B \otimes B^\dag \|_2, \]
where $\|\cdot\|_2$ is the normalized Schatten 2-norm $\|M\|_2 = \sqrt{\frac{1}{d} \sum_{i,j=1}^d |M_{ij}|^2}$~\cite{low09}. Observe that this representation shows that $D(\cdot,\cdot)$ satisfies the triangle inequality.

\item We have $\|M\|_2^2 = \hsip{M}{M}$. Therefore, $\|A-B\|_2^2 = \hsip{A-B}{A-B} = 2-2\text{Re} \hsip{A}{B}$. This implies that $D(A,B) \le \|A-B\|_2$, via the elementary inequality $2 \text{Re}\,z \le |z|^2 + 1$, valid for any $z \in \C$.
\end{itemize}

The following justifies the claim that $D(\cdot,\cdot)$ is indeed an ``average-case'' measure of distance.

\begin{prop}
Fix $d$-dimensional unitary operators $U$ and $V$. Then
\[ \int d\psi\,D(U\ket{\psi},V\ket{\psi})^2 = \frac{d}{d+1} D(U,V)^2, \]
where the integral is taken according to Haar measure on pure states $\ket{\psi} \in \C^d$.
\end{prop}

\begin{proof}
We have
\beas
\int d\psi\,D(U\ket{\psi},V\ket{\psi})^2 &=& 1 - \int d\psi\,|\bracket{\psi}{U^\dag V}{\psi}|^2\\
&=& 1 - \int d\psi \Tr[ (U^\dag V \otimes V^\dag U) \proj{\psi}^{\otimes 2} ]\\
&=& 1 - \Tr\left[ (U^\dag V \otimes V^\dag U) \left(\frac{I + F}{d(d+1)} \right)\right]\\
&=& \frac{d}{d+1} \left(1 - \left| \frac{\Tr(U^\dag V)}{d} \right|^2 \right)\\
&=& \frac{d}{d+1} D(U,V)^2.
\eeas
In the third equality we use the fact that $\int \proj{\psi}^{\otimes 2} d\psi = (I+F)/(d(d+1))$, where $F$ is the flip (or swap) operator which interchanges two $d$-dimensional systems. The fourth equality follows from the facts that, for any $d$-dimensional operators $A$, $B$, $\Tr(A \otimes B) = \Tr(A)\Tr(B)$ and $\Tr((A \otimes B) F) = \Tr(AB)$.
\end{proof}

The quantity $\int d\psi\,|\bracket{\psi}{U^\dag V}{\psi}|^2$ appearing in the proof was previously introduced by Ac\'in~\cite{acin01} as an average-case variant of the fidelity. We will see in Section~\ref{sec:statesunitaries} below a number of properties, including the {\bf Identity} property, which can be tested efficiently with respect to the distance measure $D(\cdot,\cdot)$.


\subsubsection{Controlled and inverse unitaries}
\label{sec:controlledu}

As well as being given access to a unitary operator $U$, we may be given access to the inverse $U^{-1}$ and/or the controlled unitary c-$U$, or in other words the operator $\proj{0} \otimes I + \proj{1} \otimes U$. This may be a reasonable assumption if we would like to apply our property testing algorithm to a unitary operator given in the form of a quantum circuit; on the other hand, it may not be reasonable in an adversarial scenario where we only assume access to $U$ as a black box.

For any $U$, $V$ we have $\hsip{\text{c-}U}{\text{c-}V} = (1+\hsip{U}{V})/2$, implying
\beas
D(\text{c-}U,\text{c-}V) &=& \sqrt{1-\left|\frac{1+\hsip{U}{V}}{2}\right|^2}\\
&=& \frac{1}{2} \sqrt{3 - 2 \text{Re} \hsip{U}{V} - |\hsip{U}{V}|^2 }\\
&=& \frac{1}{2} \sqrt{\|U-V\|_2^2 + D(U,V)^2}.
\eeas
Recalling that $D(U,V) \le \|U-V\|_2$, we therefore have the inequalities
\be \label{eq:cu} \|U-V\|_2/2 \le D(\text{c-}U,\text{c-}V) \le \|U-V\|_2/\sqrt{2}. \ee
Thus, given access to controlled unitaries, one can hope to design tests which are sensitive to the 2-norm distance $\|U-V\|_2$. For example, if we are allowed access to controlled unitaries we can distinguish $U$ from $-U$ (see the next section for how this can be done), whereas this is impossible given access to $U$ alone.

Being given access to $U^{-1}$ can also be powerful. In particular, it allows us to apply the important primitive of amplitude amplification~\cite{bhmt:countingj} to property testing algorithms, in analogy to Section~\ref{sssec:amplampl}. Imagine we have a test for a property $\Prop \subseteq U(d)$ which uses $q$ copies of the input unitary $U$, and such that for $U \in \Prop$ the test always accepts (it has perfect completeness), and for $U$ $\eps$-far from $\Prop$, the test accepts with probability at most $f(\eps)$. Then amplitude amplification allows us to test $\Prop$ with $O(q/\sqrt{f(\eps)})$ copies of $U$, rather than the $O(q/f(\eps))$ copies that would be required by simple repetition. For example, we will see below that this gives a square-root speed-up for testing equality of unitary operators. In the complexities we quote below, we assume that amplitude amplification has not been applied.


\subsubsection{From properties of states to properties of unitaries}
\label{sec:statesunitaries}

There is a correspondence between pure quantum states and unitary operators, which is known as (a special case of) the Choi-Jamio\l kowski isomorphism~\cite{choi75,jamiolkowski72} and will sometimes allow us to translate tests for properties of states to tests for analogous properties of unitaries. Given access to $U \in U(d)$, we first prepare the maximally entangled state of two $d$-dimensional systems
\[ \ket{\Phi} := \frac{1}{\sqrt{d}} \sum_{i=1}^d \ket{i}\ket{i} \]
and then apply $U$ to the first system. We obtain the state $\ket{U} \in (\C^d)^{\otimes 2}$ defined by
\[ \ket{U} = \frac{1}{\sqrt{d}} \sum_{i,j=1}^d U_{ji} \ket{j} \ket{i}. \]
The isomorphism is thus simply $U \leftrightarrow \ket{U}$. The state $\ket{U}$ faithfully represents the original operator~$U$; in particular, it is easy to see that $\ip{U}{V} = \hsip{U}{V}$ and hence $D(U,V) = D(\ket{U},\ket{V})$. So, if we have a tester for some property $\Prop$ of $d^2$-dimensional quantum states, by applying the test to $\ket{U}$ we obtain a tester with the same parameters for an analogous property $\Prop'$ of $d$-dimensional unitary operators.

However, one sometimes has to be careful. Imagine we have a tester which accepts states with property $\Prop$ with certainty, and accepts states which are $\eps$-far away from having property $\Prop$ with probability at most $\delta$. Then, via the Choi-Jamio\l kowski isomorphism, this translates into a tester which accepts unitary matrices with property $\Prop'$ with certainty, and accepts, with probability at most $\delta$, unitaries which are $\eps$-far away from any \emph{matrix} $M$ with $\hsip{M}{M}=1$ such that $M$ has property $\Prop'$. Therefore, in principle it could be the case that $U$ is far from any unitary matrix with property $\Prop'$, but is close to some non-unitary matrix $M$ which has property $\Prop'$. In this situation the tester might incorrectly accept. Nevertheless, in various cases of interest one can show that this situation does not arise. In particular, we have the following lemma (which generalizes similar claims in~\cite{harrow13,wang11}).

\begin{lem}
\label{lem:closeunitary}
Let $\Prop \subseteq M(d)$, and $U \in U(d)$. For $M \in \Prop$ such that $\hsip{M}{M}=1$, let $M = AV$ be a polar decomposition of $M$, with $A = \sqrt{MM^\dag}$ and $V$ unitary. Then, if $V \in \Prop$ and $D(U,M) = \eps$, $D(U,\Prop \cap U(d)) \le 2\eps$.
\end{lem}

\begin{proof}
We have
\[ \hsip{M}{V} = \frac{1}{d} \Tr(\sqrt{MM^\dag}) = \frac{1}{d}\|M\|_1 = \frac{1}{d}\max_{W\in U(d)} |\Tr(WM)| \ge \sqrt{1-\eps^2}, \]
using the definition of the trace norm and that $D(U,M)=\eps$. Thus
\[ D(U,V) \le D(U,M)+D(M,V) \le 2\eps. \]
\end{proof}

The following are some examples where one can use the Choi-Jamio\l kowski isomorphism to test properties of unitary operators:

\begin{itemize}
\item The {\bf Equality to $V$} property, where $U$ satisfies the property if $U=e^{i \theta} V$, for some $\theta$. The test creates the state $\ket{U}$ and measures in the basis $\{\proj{V},I-\proj{V}\}$. Using the analysis of the corresponding property for pure states, this property is testable with $O(1/\eps^2)$ uses of~$U$. A simple special case of this is the previously discussed {\bf Identity} property.

\item The {\bf Equality} property for pairs of unitary operators, where the pair $U,V$ satisfies the property if $U=e^{i\theta} V$, for some $\theta$. This can be tested by applying the swap test to $\ket{U}$ and $\ket{V}$; again, the analysis of the {\bf Equality} property for states goes through unchanged, implying that this property is testable with $O(1/\eps^2)$ uses of $U$ and $V$.

\item The {\bf Inverses} property, where $U,V \in U(d)$ satisfy the property if $U = e^{i\theta} V^{-1}$, for some $\theta$. The test is to create the state $\ket{UV}$ with one use of each of $U$ and $V$, then to test for equality to $\ket{\Phi}$. The probability of rejection is $D(UV,I)^2 = D(U,V^{-1})^2$, so if $D(U,V^{-1})=\eps$, the test rejects with probability $\eps^2$. Note that there is no need to have access to $U^{-1}$ or $V^{-1}$.

\item The {\bf Product} property for unitary operators, where an operator $U \in U(d^n)$ satisfies the property if $U = U_1 \otimes U_2 \otimes \dots \otimes U_n$ for some $U_1,\dots,U_n \in U(d)$. This can be tested by applying the product test described in Section~\ref{sec:productness} to $\ket{U}$~\cite{harrow13}. One also needs to show that, if $U$ is close to an operator $A \in M(d^n)$ such that $A = A_1 \otimes \dots \otimes A_n$, $U$ is in fact close to a \emph{unitary} operator of this form; this claim follows from Lemma \ref{lem:closeunitary}. The final result is that if $U$ is product the test accepts with certainty, whereas if $U$ is $\eps$-far from product, the test rejects with probability $\Theta(\eps^2)$.
\end{itemize}


\subsubsection{Membership of the Pauli and Clifford groups}
\label{sec:paulitest}

Let $B = \{B_1,\dots,B_{d^2}\}$ be a unitary operator basis for the space of linear operators on $d$ dimensions such that $B$ is orthonormal with respect to the normalized Hilbert-Schmidt inner product, i.e., $\hsip{B_i}{B_j} = \delta_{ij}$. Then the set $\ket{B_i}$ forms an orthonormal basis for $\C^{d^2}$ with respect to the standard inner product, implying that one can test membership of a unitary operator $U$ in $B$ with the following procedure, which we call the \emph{operator basis test}.
\begin{enumerate}
\item Create two copies of $\ket{U}$.
\item Measure each copy in the basis $\{\ket{B_1},\dots,\ket{B_{d^2}}\}$.
\item Accept if both measurements give the same result.
\end{enumerate}
The probability of getting outcome $i$ from each measurement is independent and equal to $|\hsip{U}{B_i}|^2$. Thus, if $U=e^{i\theta} B_i$ for some $i$, then the test will accept with certainty. On the other hand, if $\min_{V \in B} D(U,V) = \eps$, the probability of getting the same measurement outcome twice is
\[ \sum_{i=1}^{d^2} |\hsip{U}{B_i}|^4 \le \max_i |\hsip{U}{B_i}|^2 \sum_{i=1}^{d^2} |\hsip{U}{B_i}|^2 = 1 - \eps^2. \]
Therefore, by repeating the operator basis test and rejecting if any of the individual tests reject, the property of {\bf Membership in $B$} can be tested with $O(1/\eps^2)$ uses of $U$.

A natural operator basis to which this test can be applied is the set of Pauli matrices on $n$ qubits~\cite{montanaro10c,wang11}, which form a basis for the space of linear operators on $n$ qubits. This basis is orthonormal with respect to the normalized Hilbert-Schmidt inner product. We call the corresponding basis for $\C^{2^{2n}}$ obtained via the Choi-Jamio\l kowski isomorphism the \emph{Pauli basis}. The operator basis test can be immediately applied to test whether an $n$-qubit operator is proportional to an $n$-qubit Pauli matrix, or is far from any such matrix; we call this special case the \emph{Pauli test}. As pointed out in~\cite{montanaro10c}, this is a natural quantum generalization of the important classical property of linearity of Boolean functions~\cite{blr:selftest} discussed in Section~\ref{sssec:amplampl}. Given access to an oracle for $f:\{0,1\}^n \rightarrow \{0,1\}$, one can readily construct the diagonal unitary operator $U_f$ where $U_f\ket{z} = (-1)^{f(z)}\ket{z}$, and also the controlled unitary operator $\text{c-}U_f$; it is easy to see that $f$ is linear (with respect to addition mod 2) if and only if $U_f$ is a tensor product of identity and $Z$ operators. Further, if $\ell:\{0,1\}^n \rightarrow \{0,1\}$ is a Boolean function, the distance between $\text{c-}U_f$ and $\text{c-}U_\ell$ is 
\beas
D(\text{c-}U_f,\text{c-}U_\ell) &=& \sqrt{1-\left(\frac{1}{2} + \frac{1}{2^{n+1}} \sum_{z \in \{0,1\}^n} (-1)^{f(z)+\ell(z)} \right)^2}\\
&=& \sqrt{1-(1-|\{z:f(z) \neq \ell(z)\}|/2^n)^2}\\
&=& \sqrt{2d(f,\ell) - d(f,\ell)^2},
\eeas
where $d(f,\ell) := |\{x:f(x) \neq \ell(x)\}|/2^n$ is the normalized Hamming distance. This implies that the Pauli test (for the special case of testing diagonal Pauli matrices) can be used to test linearity of Boolean functions, recovering the $O(1/\eps)$ complexity of the classical tester discussed in Section~\ref{sssec:amplampl}, which can be improved to $O(1/\sqrt{\eps})$ via amplitude amplification.

The Pauli test can also be used as a subroutine in an algorithm for testing membership in the \emph{Clifford group}. The Clifford group $\mathcal{C}_n$ on $n$ qubits is the normalizer of the Pauli group $\Prop_n$, or in other words the set $\mathcal{C}_n:=\{C \in U(2^n): \forall P \in \Prop_n, C P C^{-1} \in \Prop_n \}$. The Clifford group plays an important role in many areas of quantum information theory, including quantum error-correction and simulation of quantum circuits~\cite{gottesman99,nielsen&chuang:qc}. Wang~\cite{wang11} has shown that, given access to a unitary $U$ and its inverse~$U^{-1}$, whether $U$ is a member of the Clifford group can be tested with $O(1/\eps^2)$ uses of $U$ and $U^{-1}$; this result improves a previous test of Low~\cite{low09} by removing any dependence on $n$, and can in turn be improved to $O(1/\eps)$ using amplitude amplification~\cite{bhmt:countingj}.

Wang's test is very natural: pick a Pauli matrix $P \in \Prop_n$ uniformly at random, and apply the Pauli test to the operator $U P U^{-1}$. If $U \in \mathcal{C}_n$, this test will always accept. Intuitively, if $U$ is far from any Clifford operator, then we expect that for most Pauli operators $P$, $UPU^{-1}$ will be far from being a Pauli operator, so repeating this test a constant number of times would suffice to detect this. Making this intuition precise requires some work; see~\cite{wang11} for the details.

\begin{question}
Is there an efficient test for the property of membership in the Clifford group which does not require access to $U^{-1}$?
\end{question}


\subsubsection{Testing commutativity}

Say that $U,V\in U(d)$ satisfy the {\bf Commuting} property if $UV = VU$. Assuming that we are given access to the controlled operators $\text{c-}U$ and $\text{c-}V$, consider the following tester for this property:

\begin{enumerate}
\item Create the states $\ket{\text{c-}U\text{c-}V}$, $\ket{\text{c-}V\text{c-}U}$ by applying controlled-$U$ and controlled-$V$ operations to the first half of each of two maximally entangled states.
\item Apply the swap test to these states and accept if the test accepts.
\end{enumerate}

If $U$ and $V$ commute, then $\text{c-}U$ and $\text{c-}V$ also commute, so $\ket{\text{c-}U\text{c-}V} = \ket{\text{c-}V\text{c-}U}$ and hence the swap test accepts with certainty. On the other hand, if $\|UV-VU\|_2 = \eps$, then by (\ref{eq:cu}) the test rejects with probability at least $\eps^2/8$. In order for this to be a good test for commutativity, we therefore need to show that, if $\|UV-VU\|_2 = \eps$, $U$ and $V$ are close to a pair of unitary operators $\widetilde{U}$, $\widetilde{V}$ such that $\widetilde{U}$ and $\widetilde{V}$ commute. Precisely this result has recently been shown by Glebsky~\cite{glebsky10} in the form of the following theorem, whose proof we omit.

\begin{thm}[Glebsky~\cite{glebsky10}]
Let $U,V \in U(d)$ satisfy $\|UV-VU\|_2 = \eps$. Then there exist $\widetilde{U},\widetilde{V} \in U(d)$ such that $\widetilde{U}$ and $\widetilde{V}$ commute and $\|U - \widetilde{U}\|_2 \le 30\eps^{1/9}$, $\|V - \widetilde{V}\|_2 \le 30\eps^{1/9}$.
\end{thm}

The consequence is that the above tester rejects pairs $(U,V)$ such that $U$ and $V$ are $\eps$-far from a pair of commuting matrices with probability $\Omega(\eps^{18})$. By repeating the test $\poly(1/\eps)$ times, we obtain a tester which rejects such pairs with constant probability.

\begin{question}
Is there an efficient test for commutativity which does not require access to the controlled unitaries c-$U$, c-$V$, but just uses $U$ and $V$?
\end{question}


\subsubsection{Testing quantum juntas}

Analogously to the classical case of Boolean functions $f:\{0,1\}^n \rightarrow \{0,1\}$, a unitary operation on $n$ qubits is said to be a \emph{$k$-junta} if it acts non-trivially on at most $k$ of the qubits, or in other words is of the form $U_S \otimes I_{S^c}$, where $U \in U(2^k)$ and $S$ is a $k$-subset of $[n]$. Wang~\cite{wang11} has given a tester for whether a unitary operator $U$ is a $k$-junta, which turns out to be a direct generalization of the tester of At\i c\i\ and Servedio~\cite{atici&servedio:testing} for the classical property of a Boolean function being a $k$-junta (Section~\ref{sec:juntas}). The work~\cite{montanaro10c} had previously studied a different tester for being a 1-junta (``dictator''), but did not prove correctness. Wang's tester proceeds as follows:
\begin{enumerate}
\item Set $W=\emptyset$.
\item Repeat the following procedure $T$ times, for some $T$ to be determined:
\begin{enumerate}
\item Create the state $\ket{U}$ and measure in the Pauli basis, obtaining outcome $s \in \{I,X,Y,Z\}^n$.
\item Update $W \leftarrow W \cup \{i:s_i \neq I\}$.
\item If $|W| > k$, reject.
\end{enumerate}
\item Accept.
\end{enumerate}
To show correctness of this test, it suffices to prove the following claim:
\begin{thm}[Wang \cite{wang11}]
If $U$ is $\eps$-far from any $k$-junta, and $T=\Theta(k/\eps^2)$, the above procedure accepts with probability at most $1/3$.
\end{thm}

The result originally shown by Wang~\cite{wang11} was a somewhat worse bound of $T=\Theta(k \log (k/\eps)/\eps^2)$, but the bound can be improved to $\Theta(k/\eps^2)$ via a straightforward generalization of the analysis of At\i c\i\ and Servedio~\cite{atici&servedio:testing}, as we now show (cf.\ Section~\ref{sec:juntas}). If we are given access to $U^{-1}$ as well, the bound can be improved further to $T=\Theta(k/\eps)$ via amplitude amplification.

\begin{proof}
As the Pauli matrices form a basis for the space of $n$-qubit operators, we can expand
\[ U = \sum_{s \in \{I,X,Y,Z\}^n} \widehat{U}_s \sigma_s, \]
where $\sigma_s$ is the $n$-qubit Pauli operator corresponding to the string $s$, and $\widehat{U}_s \in \C$.  Pauli matrices are orthonormal with respect to the normalized Hilbert-Schmidt inner product, implying that $\sum_{s \in \{I,X,Y,Z\}^n} |\widehat{U}_s|^2 = 1$. Assume that $U$ is $\eps$-far from any unitary operator $V$ that is a $k$-junta, and for $s \in \{I,X,Y,Z\}^n$, let $\supp(s) := \{i:s_i \neq I\}$. Then, for any subset $W \subseteq [n]$ of size at most~$k$,
\[ w_W := \sum_{s:\supp(s)\subseteq W} |\widehat{U}_s|^2 \le 1 - \eps^2/4. \]
%
To see this, assume the opposite and consider the operator $M_W = w_W^{-1/2} \sum_{s,\supp(s) \subseteq W} \widehat{U}_s \sigma_s$. Then $M_W$ is a $k$-junta, $\hsip{M_W}{M_W}=1$, and $D(U,M_W) = (1-w_W)^{1/2} < \eps/2$. Further, the unitary matrix~$V$ occurring in a polar decomposition of $M$ is also a $k$-junta. So, by Lemma \ref{lem:closeunitary}, $D(U,V) \le \eps$, contradicting that $U$ is $\eps$-far from any unitary $k$-junta.

For each measurement, the probability that a string $s$ is returned such that $\supp s \nsubseteq W$ is therefore at least $\eps^2/4$. Thus the expected number of measurements required to find $k+1$ such indices is at most $4(k+1)/\eps^2$. The theorem then follows from Markov's inequality.
\end{proof}


\subsubsection{Other properties of unitary matrices}
\label{sec:unitaryprops}

We finish this section by mentioning a few other properties of unitary matrices which have fairly straightforward testers. Say that a unitary matrix $U$ satisfies the {\bf Diagonality} property if $U_{ij}=0$ for $i\neq j$. Consider the following easy tester for this property: Apply $U$ to a uniformly random computational basis state $\ket{i}$, measure in the computational basis, and accept if and only if the outcome is $i$. Writing $U_{kk} = r_k e^{i \gamma_k}$ for $r_k\ge0$ and $0 \le \theta_k < 2\pi$, we have
\[ \max_{D \text{ diagonal}} |\hsip{U}{D}| = \frac{1}{d} \max_{D \text{ diagonal}} \left|\sum_{k=1}^d U_{kk}^* D_{kk}\right| = \frac{1}{d} \max_{\theta_k} \left|\sum_{k=1}^d r_k e^{i (\theta_k - \gamma_k)}\right| = \frac{1}{d} \left|\sum_{k=1}^d r_k\right| = \frac{1}{d} \sum_{k=1}^d |U_{kk}|. \]
On the other hand, the probability of accepting is precisely
\[ \frac{1}{d} \sum_{k=1}^d |U_{kk}|^2 \le \frac{1}{d} \max_k |U_{kk}| \sum_{k=1}^d |U_{kk}| \le \frac{1}{d} \sum_{k=1}^d |U_{kk}|. \]
Thus, if the test accepts with probability $1-\delta$, $U$ is distance at most $\sqrt{2\delta}$ from a diagonal unitary matrix $D$, implying that {\bf Diagonality} can be $\eps$-tested with $O(1/\eps^2)$ uses of~$U$.

This tester is simple, but can be applied to the following more general problem: Given a basis $B$ for $\C^d$, is every vector in $B$ an eigenvector of $U$? This is equivalent to asking whether $VUV^{-1}$ is diagonal, where $V$ is the change of basis matrix for $B$. This problem can be solved by applying the test for diagonality to $VUV^{-1}$, noting that the distance of $VUV^{-1}$ from the nearest diagonal matrix is the same as the distance of $U$ from the nearest matrix $\widetilde{U}$ such that every vector in $B$ is an eigenvector of $\widetilde{U}$. For example, this allows us to test $U$ for being a {\bf Circulant} matrix (i.e., a matrix of the form $U_{xy} = f(x-y)$ for some $f:\{0,\dots,d-1\} \rightarrow \C$, where subtraction is understood modulo $d$) as such matrices are characterized by being diagonalized by the quantum Fourier transform over~$\Z_d$.

Finally, Wang~\cite{wang11} has proven that membership of a unitary operator $U \in U(d)$ in the orthogonal group $O(d) := \{M \in M(d): MM^T = I\}$ can be $\eps$-tested with $O(1/\eps^2)$ uses of $U$. The tester is based on applying $U \otimes U$ to $\ket{\Phi}$, which produces the state $\ket{UU^T}$, then performing the measurement $\{\proj{\Phi},I-\proj{\Phi}\}$. (Recall that $\ket{\Phi} = \frac{1}{\sqrt{d}}\sum_{i=1}^d \ket{i}\ket{i}$.) If $U\in O(d)$, the test always accepts; Wang shows that if the test accepts with high probability, then $U$ is close to an orthogonal matrix.


\subsection{Properties of quantum channels}
\label{sec:channel}

Not all physical processes which occur in quantum mechanics are reversible. The mathematical framework in which the most general  physically realizable operations are studied is the formalism of \emph{quantum channels}. A quantum channel (or ``superoperator'') is a completely positive, trace-preserving linear map $\mathcal{E}: \mathcal{B}(\C^{d_{\text{in}}}) \rightarrow \mathcal{B}(\C^{d_{\text{out}}})$. Here ``completely positive'' means that the operator $\mathcal{E} \otimes \text{id}$ preserves positivity, where id is the identity map on some ancilla system of arbitrary dimension. A comprehensive introduction to the world of quantum channels is provided by lecture notes of Watrous~\cite{watrous08a}.

There has been less work on testing properties of quantum channels than the other types of properties considered above, although the problem of discriminating between quantum channels has been considered by a number of authors (e.g.~\cite{sacchi05,duan09,piani09}).


\subsubsection{A distance measure on channels}
\label{sec:channeldist}

In the context of property testing, the first task when considering quantum channels is to define a suitable measure of distance. One approach is to use the same idea as for unitary operators, and take the distance induced by the Choi-Jamio\l kowski isomorphism~\cite{choi75,jamiolkowski72}. In the case of channels, this isomorphism states that there is a bijection between the set of quantum channels $\mathcal{E}: \mathcal{B}(\C^{d_{\text{in}}}) \rightarrow \mathcal{B}(\C^{d_{\text{out}}})$ and the set of bipartite density matrices $\rho$ on a $(d_{\text{out}} \times d_{\text{in}})$-dimensional system such that applying the partial trace to the first subsystem of $\rho$ leaves the maximally mixed state $I/d_{\text{in}}$. The bijection can be explicitly given as
\[ \mathcal{E} \leftrightarrow \frac{1}{d_{\text{in}}} \sum_{i,j=1}^{d_{\text{in}}} \mathcal{E}(\ket{i}\bra{j}) \otimes \ket{i}\bra{j} =: \chi_\mathcal{E}. \]
Then one distance measure that can be put on quantum channels $\mathcal{E}$, $\mathcal{F}$ is
\[ D(\mathcal{E},\mathcal{F}) := D(\chi_\mathcal{E},\chi_\mathcal{F}). \]
As with the correspondence between unitary operators and pure states, this distance measure allows one to translate tests for properties of mixed states to properties of channels. For example, consider the property {\bf Unitarity}, where $\mathcal{E}: \mathcal{B}(\C^d) \rightarrow \mathcal{B}(\C^d)$ satisfies the property if and only if it is a unitary operator. $\mathcal{E}$ is unitary if and only if $\chi_\mathcal{E}$ is a pure state (and hence maximally entangled). In order to test this property, we can apply the test for {\bf Purity} of mixed states to $\chi_\mathcal{E}$. From the analysis of Section~\ref{sec:mixed}, we see that if the test accepts with probability $1-\delta$, there exists a pure state $\ket{\psi}$ such that $D(\chi_\mathcal{E},\proj{\psi}) = O(\delta)$. We still need to show that $\chi_\mathcal{E}$ is in fact close to a pure state which is maximally entangled. To do so, first write $\ket{\psi} = \sum_{i=1}^{d} \sqrt{\lambda_i} \ket{v_i}\ket{w_i}$ for the Schmidt decomposition of $\ket{\psi}$, 
and define the maximally entangled state $\ket{\eta} = \frac{1}{\sqrt{d}} \sum_{i=1}^d \ket{v_i}\ket{w_i}$.

Then we have the sequence of inequalities and equalities
\beas
D(\chi_\mathcal{E},\proj{\psi}) &\ge& D\left(I/d,\Tr_B(\proj{\psi}\right)) \ge 1 - F\left(I/d,\Tr_B(\proj{\psi})\right) = 1 - \frac{1}{\sqrt{d}} \sum_{i=1}^d \sqrt{\lambda_i}\\
&=& 1- |\ip{\psi}{\eta}| \ge D(\proj{\psi},\proj{\eta})^2/2.
\eeas
The first inequality holds because the trace norm does not increase under partial trace~\cite[Theorem~9.2]{nielsen&chuang:qc}, and the second is (\ref{eq:fvsd}). Therefore, if the test accepts with probability $1-\delta$,
\[ D(\chi_\mathcal{E},\proj{\eta}) \le D(\chi_\mathcal{E},\proj{\psi}) + D(\proj{\psi},\proj{\eta}) = O(\delta + \sqrt{2\delta}) = O(\sqrt{\delta}), \]
implying that {\bf Unitarity} of a quantum channel can be $\eps$-tested with $O(1/\eps^2)$ uses of the channel.


\subsubsection{Testing quantum measurements}

An important special case of quantum channels is the case of \emph{quantum measurements}. In full generality, a measurement on a $d$-dimensional quantum mechanical system is defined by a sequence of linear operators $M = (M_1,\dots,M_k)$ such that $\sum_{i=1}^k M_i^\dag M_i = I$. If $M$ is performed on the state~$\rho$, the probability of receiving outcome $i$ is $\Tr(M_i \rho M_i^\dag)$, and the resulting state of the system, given that outcome $i$ occurred, is
\[ \rho_i = \frac{M_i \rho M_i^\dag}{\Tr(M_i \rho M_i^\dag)}. \]
The quantum channel corresponding to performing the measurement $M$ and storing the outcome in a new register is the map $\mathcal{M}$ where
\[ \mathcal{M}(\rho) = \sum_{i=1}^k M_i \rho M_i^\dag \otimes \proj{i}, \]
so the Choi-Jamio\l kowski state is
\[ \chi_\mathcal{M} = \frac{1}{d} \sum_{i,j=1}^d \left(  \sum_{\ell=1}^k M_\ell \ket{i}\bra{j} M_\ell^\dag \otimes \proj{\ell}\right) \otimes \ket{i}\bra{j} \]
which, by reordering subsystems, is equivalent to
\[ \sum_{\ell=1}^k \left(\frac{1}{\sqrt{d}} \sum_{i=1}^d M_\ell \ket{i}\ket{i} \right) \left(\frac{1}{\sqrt{d}} \sum_{j=1}^d M_\ell^\dag \bra{j}\bra{j} \right) \otimes \proj{\ell} =: \sum_{\ell=1}^k \proj{\psi_M^{(\ell)}} \otimes \proj{\ell}. \]
For any two measurements $M$ and $N$ with at most $k$ outcomes, we can therefore compute the distance between the corresponding channels as
\[ D(\mathcal{M}, \mathcal{N}) = \sum_{\ell=1}^k D\left(\ket{\psi_M^{(\ell)}}, \ket{\psi_N^{(\ell)}}\right), \]
where if $M$ (resp.\ $N$) has $\ell < k$ outcomes, we set $M_i=0$ (resp.\ $N_i=0$) for $\ell < i \le k$. Observe that, using this measure of distance, we take into account the distance of the post-measurement states as well as the distance between the probability distributions corresponding to the measurement outcomes. One can explicitly calculate that, for any (potentially unnormalized) vectors $\ket{\psi}$, $\ket{\phi}$,
\[ D(\ket{\psi},\ket{\phi}) = \sqrt{\frac{1}{4}\left(\ip{\psi}{\psi} + \ip{\phi}{\phi}\right)^2 - |\ip{\psi}{\phi}|^2}, \]
which implies that
\[ D(\mathcal{M},\mathcal{N}) = \frac{1}{2} \sum_{i=1}^k \sqrt{ \left(\hsip{M_i}{M_i} + \hsip{N_i}{N_i} \right)^2 - 4|\hsip{M_i}{N_i}|^2}. \]
Recent work by Wang~\cite{wang12} has given efficient tests for a number of properties of quantum measurements, but with respect to a measure of distance which appears somewhat different to the measure $D(\cdot,\cdot)$. Given two measurements $M$ and $N$ with at most $k$ outcomes, Wang's distance measure is
\[ \Delta(M,N) := \sqrt{\frac{1}{2} \sum_{i=1}^k \hsip{M_i}{M_i} + \hsip{N_i}{N_i} -2|\hsip{M_i}{N_i}| }. \]
Wang demonstrates that $\Delta(\cdot,\cdot)$ has a number of desirable properties, including satisfying the triangle inequality and being an ``average-case'' measure of distance~\cite{wang12}. It turns out that $\Delta(\cdot,\cdot)$ is in fact closely related to $D(\cdot,\cdot)$, which we encapsulate as the following lemma.

\begin{lem}
\label{lem:dvsdelta}
Given two measurements $M$ and $N$, let $\mathcal{M}$ and $\mathcal{N}$ be the corresponding channels. Then
\[ D(\mathcal{M},\mathcal{N})/\sqrt{2} \le \Delta(M,N) \le D(\mathcal{M},\mathcal{N})^{1/2}. \]
\end{lem}

\begin{proof}
To prove the upper bound part of the lemma, it suffices to show that, for each $i$,
\[ (\hsip{M_i}{M_i} + \hsip{N_i}{N_i} -2|\hsip{M_i}{N_i}|)^2 \stackrel{?}{\le} \left(\hsip{M_i}{M_i} + \hsip{N_i}{N_i} \right)^2 - 4|\hsip{M_i}{N_i}|^2. \]
Setting $x_i = \hsip{M_i}{M_i} + \hsip{N_i}{N_i}$, $y_i=2|\hsip{M_i}{N_i}|$ and rearranging terms, we get the claimed inequality
\[ (x_i-y_i)^2 \stackrel{?}{\le} (x_i-y_i)(x_i+y_i), \]
which holds because $y_i \le x_i$ by Cauchy-Schwarz or the inequality of arithmetic and geometric means. For the lower bound, we need to show
\[ \frac{1}{2\sqrt{2}} \sum_{i=1}^k (x_i-y_i)^{1/2} (x_i+y_i)^{1/2} \stackrel{?}{\le} \sqrt{\frac{1}{2} \sum_{i=1}^k (x_i - y_i)}. \]
Indeed, by Cauchy-Schwarz,
\beas
\frac{1}{2\sqrt{2}} \sum_{i=1}^k (x_i-y_i)^{1/2} (x_i+y_i)^{1/2} &\le& \frac{1}{2\sqrt{2}} \sqrt{\sum_{i=1}^k (x_i-y_i)}\sqrt{\sum_{i=1}^k x_i+y_i}\\
&\le& \sqrt{\frac{1}{2}\sum_{i=1}^k (x_i-y_i)}\sqrt{\frac{1}{2}\sum_{i=1}^k x_i} \\
&=& \sqrt{\frac{1}{2} \sum_{i=1}^k (x_i - y_i)}
\eeas
as required, using $\sum_{i=1}^k \hsip{M_i}{M_i} = \sum_{i=1}^k \hsip{N_i}{N_i} = 1$.
\end{proof}

Lemma \ref{lem:dvsdelta} implies that Wang's results with respect to the distance measure $\Delta(\cdot,\cdot)$ can be translated into results with respect to $D(\cdot,\cdot)$. In particular, Wang~\cite{wang12} gives efficient testers for the following properties of quantum measurements:

\begin{itemize}
\item The property of being a {\bf Pauli} measurement (called ``stabilizer measurement'' in~\cite{wang12}): $M$ is a Pauli measurement if it is a two-outcome projective measurement onto the $\pm1$ eigenspaces of an $n$-qubit Pauli operator $\sigma_s$, for some $s \in \{I,X,Y,Z\}^n$. Wang shows that this property can be $\eps$-tested with $O(1/\eps^4)$ measurements.

\item The property of being an {\bf $\ell$-local} measurement of $n$ qubits, i.e., acting non-trivially on at most $\ell$ qubits. Wang gives an $\eps$-tester for this property which uses $O(\ell \log (\ell/\eps)/\eps^2)$ measurements.

\item The property of being a {\bf Permutation invariant} measurement of $n$ $d$-dimensional systems, i.e., a measurement which is unchanged when the $n$ systems are permuted arbitrarily. This property can be $\eps$-tested with $O(1/\eps^2)$ measurements.

\item Being contained within any finite set of measurements $\mathcal{S} = \{M_i\}$ with $k$ outcomes on a $d$-dimensional system. If $\Delta(M_i,M_j) \ge \gamma$ for all $i \neq j$, and we set $\delta = \min \{\gamma,\eps\}$, membership in $\mathcal{S}$ can be $\eps$-tested with $O(k^2 (\log k)/\delta^8 + (\log |S|)/\delta^2 )$ measurements.

\item {\bf Equality} of measurements, which can be $\eps$-tested with $O(k^5 (\log k) / \eps^{12})$ measurements. This is based on a more general algorithm for estimating the distance between measurements.
\end{itemize}

All of the above testers are based on constructing multiple copies of the Choi-Jamio\l kowski state corresponding to the measurement to be tested, and performing some measurements on the states. As remarked in~\cite{wang12}, it is an interesting question whether efficient testers can be designed in a setting where one is not allowed access to entanglement.

\begin{question}
Can efficient testers for the properties of unitary operators and quantum channels discussed above be designed which do not require entanglement with an ancilla system?
\end{question}

It is possible to use quantum process tomography to completely characterize any quantum channel without the use of entanglement~\cite[\S8.4.2]{nielsen&chuang:qc}, so the question is only whether the above properties can still be tested \emph{efficiently} in this setting.


\section{Quantum properties and computational complexity}
\label{sec:compcomp}

Classically, the field of property testing has had close connections to computational complexity. In this section, we briefly discuss three ways in which quantum property testing can be related to quantum computational complexity. First, we discuss how, if we change the setting in which we work, testing certain natural properties can be proven computationally hard. Second, we mention how quantum property testers can be used to prove complexity class inclusions. Finally, we consider potential connections between quantum property testing and a proposed quantum PCP conjecture.


\subsection{Computational hardness of testing quantum properties}

A different perspective from which to study the question of testing properties of quantum systems is to consider problems where, instead of being given access to a quantum object, we are given a concise \emph{classical} description of that object (for example, a quantum circuit on $n$ qubits with $\poly(n)$ gates). Our aim is to efficiently determine whether the corresponding quantum object has some property, or is far from having that property, in terms of some distance measure. The distance measure used may be quite different to the distances we discuss in the rest of the survey, leading to qualitatively different results. This type of problem turns out to be naturally addressed via the framework of computational complexity.

In particular, many problems related to testing properties of quantum circuits turn out to be $\QMA$-complete.\footnote{$\QMA$ is the quantum analog of $\mathsf{NP}$~\cite{kitaev02}; see~\cite{bookatz12} for a recent survey.} These hardness results provide an interesting counterpoint to the largely positive results obtained in the ``average-case'' scenarios considered by property testing. A prototypical example of this phenomenon is ``non-identity-check,'' which was proven to be $\QMA$-complete by Janzing et al.~\cite{janzing05}. Here the input is a quantum circuit implementing a unitary $U$, and two numbers $a$, $b$ such that $b-a \ge 1/\poly(n)$, and the problem is to distinguish between the two cases that $\min_{\theta \in \R} \|U - e^{i\theta} I\| \le a$ and $\min_{\theta \in \R} \|U - e^{i\theta} I\| \ge b$. Observe that, if we replace the operator norm with the normalized 2-norm in this definition, this problem is in $\mathsf{BQP}$ by the efficient tester for the {\bf Equality to $V$} property discussed in Section~\ref{sec:statesunitaries}.

If one generalizes to quantum circuits acting on mixed states, where each elementary gate is a quantum channel, some natural problems even become $\mathsf{PSPACE}$-complete. In particular, Rosgen and Watrous~\cite{rosgen05} showed that $\mathsf{PSPACE}$-completeness holds for the problem of testing whether two mixed-state quantum circuits are distinguishable, and it remains hard when the quantum circuits are restricted to be logarithmic depth~\cite{rosgen08}, degradable or anti-degradable~\cite{rosgen10a}. In this case, distinguishability is measured in the so-called diamond norm for quantum channels~\cite{kitaev02}; the diamond norm of an linear operator $\Phi: \mathcal{B}(\C^{d_{\text{in}}}) \rightarrow \mathcal{B}(\C^{d_{\text{out}}})$ is defined to be
\begin{equation}\label{def:diamondnorm}
\|\Phi\|_\diamond := \max_{X,\|X\|_1=1} \| (\Phi \otimes \text{id})(X) \|_1, 
\end{equation}
where id is the identity map acting on an ancilla system, which may be taken to be of dimension at most $d_{\text{in}}$. Then the Quantum Circuit Distinguishability problem is to determine, given two mixed-state quantum circuits $Q_0$, $Q_1$ and constants $a<b$, whether $\|Q_0-Q_1\|_\diamond \le a$ or $\|Q_0-Q_1\|_\diamond \ge b$. 
As with the trace distance between quantum states, $\|Q_0-Q_1\|_\diamond$ can be seen as measuring the distinguishability of $Q_0$ and $Q_1$ in a ``best-case'' scenario. This contrasts with the ``average-case'' distance measure $D(Q_0,Q_1)$ introduced in Section~\ref{sec:channeldist}.

These distinguishability problems were originally shown to be hard for the complexity class $\mathsf{QIP}$ of languages decided by quantum interactive proof systems, but this class was later proven to equal $\mathsf{PSPACE}$~\cite{jain11}. The proof technique of~\cite{rosgen05} starts by using a result of Kitaev and Watrous~\cite{kitaev00}, which states that all quantum interactive proofs can be parallelized to three rounds. A mathematical reformulation of this result is that the Close Images problem is $\mathsf{QIP}$-hard. This problem is defined as follows: given two quantum circuits $Q_0$, $Q_1$ and constants $a<b$, distinguish between the cases that there is an input $\rho$ such that $F(Q_0(\rho),Q_1(\rho))\ge b$, or that for all inputs $\rho$, $F(Q_0(\rho),Q_1(\rho))\le a$. Hardness of Quantum Circuit Distinguishability is then shown by a reduction from Close Images~\cite{rosgen05}.


\subsection{From quantum property testers to complexity class inclusions}

By contrast to the results in the previous section, work by Hayden et al.~\cite{hayden13} demonstrates that quantum property testers can be used to prove positive results (i.e., upper bounds) regarding the complexity of testing properties of quantum circuits. The problem considered by these authors is a variant of the separability-testing problem (cf.\ Sections \ref{sec:productness} and \ref{sec:mixed}). In this variant the input is the description of a mixed-state quantum circuit $Q$ on $n$ qubits, and one considers the output of the circuit as a bipartite state by dividing these qubits into two disjoint sets. The problem is to distinguish between the two cases that: (a) the output of $Q$, when applied to the input $\ket{0^n}$, is close to separable; (b) the output is far from separable. Hayden et al.~\cite{hayden13} show that this problem can be solved by a quantum interactive proof system with two messages (i.e., a message from verifier to prover, followed by a reply from prover to verifier), and hence sits in the complexity class $\mathsf{QIP}(2)$. The protocol is based on the verifier applying the permutation test discussed in Section~\ref{sec:equality}. This result is somewhat subtle in that ``close'' and ``far'' are defined asymmetrically (the former in terms of the trace distance, the latter in terms of the so-called ``1-way LOCC'' distance); see~\cite{hayden13} for details.

More recently Gutoski et al.~\cite{milner13} generalized this work: for almost every complexity class defined by quantum interactive proofs, they give a version of the separability testing problem which is complete for that class. This shows that property testing problems can be used to characterize many quantum complexity classes. For example, they use the product test of~\cite{harrow13} (see Section~\ref{sec:productness}) to show that testing whether the output of a pure-state quantum circuit is a product state is in $\mathsf{BQP}$.


\subsection{The quantum PCP conjecture}

A classic and important problem in quantum computational complexity is the \emph{local Hamiltonian problem}. Here we are given as input a Hamiltonian $H$ on $n$ qubits, described by a set of Hermitian operators $H_i$ such that $H = \sum_{i=1}^m H_i$, with each operator $H_i$ acting non-trivially on at most $k=O(1)$ qubits and satisfying $\|H_i\|=O(1)$. We are also given two real numbers $a$ and $b$ such that $b-a \ge 1/\poly(n)$. We are promised that the lowest eigenvalue of $H$ is either smaller than $a$, or larger than $b$; our task is to determine which of these is the case.

This problem was proven $\QMA$-complete for $k=5$ by Kitaev~\cite{kitaev02}, which was later improved to $k=2$ by Kempe et al.~\cite{kempe06} (the case where $k=1$ is easily seen to be in $\mathsf{P}$). One way in which this hardness result could potentially be improved is in the scaling of the gap between $b$ and $a$. Indeed, it could be the case that the local Hamiltonian problem remains $\QMA$-hard if we have the promise $b-a \ge c m$ for some constant $0<c<1$. This is (one formulation of) the quantum PCP conjecture; see a recent survey of Aharonov et al.~\cite{aharonov13} for much more on this conjecture and its implications. Classically, one version of the famous PCP Theorem states that there exist constraint satisfaction problems for which it is hard to distinguish between there existing an assignment to the variables that satisfies all of the constraints, and there being no assignment that satisfies more than a constant fraction of them; the quantum PCP conjecture would be a direct quantization of this result. One way of looking at this is as the conjecture that the local Hamiltonian problem remains hard in a ``property-testing-type'' scenario where there is a large gap between ``yes'' and ``no'' instances.

\begin{question}
Is there a quantum PCP theorem?
\end{question}

Classically, the proof of the PCP Theorem relied on efficient property testers, so it seems plausible that property testing could be useful in proving a quantum generalization. Indeed, the analysis of a \emph{classical} property tester in a quantum setting has recently been central to establishing a quantum complexity-theoretic result. $\mathsf{MIP}$ is the class of languages decided by multiple-prover interactive proof systems, which was shown to be equal to $\mathsf{NEXP}$ by Babai et al.~\cite{babai91}. Recently Ito and Vidick~\cite{ito12} have shown that the quantum generalization $\mathsf{MIP}^*$, where the provers are allowed to share entanglement, is at least as powerful: $\mathsf{MIP} \subseteq \mathsf{MIP}^*$. Their proof is based on proving soundness of the classical multilinearity test of Babai et al.~\cite{babai91} in the presence of entanglement. Another application of quantum property testing to quantum complexity is the use of the analysis of an efficient quantum property tester to prove the complexity class equality $\QMA(k)=\QMA(2)$~\cite{harrow13}, as discussed in Section~\ref{sec:productness}.

Yet another connection is explored in recent work of Aharonov and Eldar~\cite{aharonov13a} on a quantum generalization of \emph{locally testable codes} (LTCs). Classically, LTCs are codes for which the property of being a codeword can be tested efficiently by means of a few local checks; such codes played a crucial role in the original proof of the PCP Theorem~\cite{almss:pcp}. The ``qLTCs'' studied in~\cite{aharonov13a} are the kernel (zero eigenspace) of $k$-local Hamiltonians $H = \sum_i H_i$, such that containment of a state in the eigenspace can be tested with good accuracy by performing measurements corresponding to only a few of the individual $k$-local terms~$H_i$.  Aharonov and Eldar~\cite{aharonov13a} prove some surprising upper bounds on the soundness for qLTCs that are \emph{stabilizer} codes, showing that they do not exist in certain regimes of parameters where classical LTCs do exist.


\section{Conclusion}

The goal of property testing is to design efficient algorithms (``testers'') to decide whether a given object has a property or is somehow ``far'' from that property, and to determine in which cases such algorithms can exist. When the objects that need to be tested are very large, exact algorithms that are also required to work for objects that ``almost'' have the property become infeasible, and property testing is often the best we can hope for.
Classical property testing is by now a very well-developed area, but \emph{quantum} property testing is just starting out.  In this paper we surveyed what is known about this:
\begin{enumerate}
\item Quantum testers for classical properties (Section~\ref{sec:qtestclprop}).
\item Classical testers for quantum properties (Section~\ref{sec:cltestqprop}).
\item Quantum testers for quantum properties (Sections~\ref{sec:states} and~\ref{sec:dynamics}).
\end{enumerate} 
We hope the overview given here, as well as the open questions mentioned along the way, will give rise to much more research in this area. Besides the properties mentioned here, there are many other properties which have been of great interest in the classical property testing literature, and whose quantum complexity is unknown. Examples include monotonicity of Boolean functions, membership of error-correcting codes, and almost all properties of graphs. In the case of quantum properties, natural targets include testing whether a unitary operator is implemented by a small circuit, and whether a Hamiltonian is $k$-local (which would be yet another variant of junta testing).

Another very broad open question not discussed previously is to what extent one can \emph{characterize} the properties (classical or quantum) that have efficient quantum testers. This may seem a hopelessly ambitious goal; nevertheless, in the case of classical algorithms it has already been achieved in some important cases, such as graph properties~\cite{alon09} and symmetric properties of probability distributions~\cite{valiant11}. Such a characterization could have importance far beyond property testing, by shedding light on the structure of problems that have efficient quantum algorithms.


\subsection*{Acknowledgements}

We thank Scott Aaronson, Aleksandrs Belovs, Robin Blume-Kohout, Jop Bri\"et, Sourav Chakraborty, Wim van Dam, Aram Harrow, Fr\'ed\'eric Magniez, Marcelo Marchiolli, Miguel Navascu\'es, Falk Unger, Lev Vaidman, Andreas Winter and Tzyh Haur Yang for helpful comments, answers to questions and pointers to literature.  We also thank the anonymous ToC referees for their exceptionally constructive feedback.




\begin{multicols}{2}
\small

\bibliographystyle{plain}
\bibliography{proptest}

\begin{thebibliography}{100}

\bibitem{aaronson06}
Scott Aaronson.
\newblock {QMA}/qpoly $\subseteq$ {PSPACE}/poly: De-{M}erlinizing quantum
  protocols.
\newblock In {\em Proceedings of 21st Annual IEEE Conference on Computational
  Complexity (CCC)}, pages 261--273, 2006.
\newblock {\tt quant-ph/0510230}.

\bibitem{aaronson:bqpph}
Scott Aaronson.
\newblock {BQP} and the {P}olynomial {H}ierarchy.
\newblock In {\em Proceedings of 42nd Annual ACM Symposium on Theory of
  Computing (STOC)}, pages 141--150, 2010.
\newblock {\tt arXiv:0910.4698}.

\bibitem{aaronson11j}
Scott Aaronson and Andris Ambainis.
\newblock The need for structure in quantum speedups.
\newblock {\em Theory of Computing}, 10:133--166, 2014.
\newblock Earlier version in ICS'11. {\tt arXiv:0911.0996}.

\bibitem{aaronson&ambainis:forrelation}
Scott Aaronson and Andris Ambainis.
\newblock Forrelation: A problem that optimally separates quantum from
  classical computing.
\newblock In {\em Proceedings of 47th Annual ACM Symposium on Theory of
  Computing (STOC)}, pages 307--316, 2015.
\newblock {\tt arXiv:1411.5729}.

\bibitem{aaronson09}
Scott Aaronson, Salman Beigi, Andrew Drucker, Bill Fefferman, and Peter Shor.
\newblock The power of unentanglement.
\newblock {\em Theory of Computing}, 5(1):1--42, 2009.
\newblock {\tt arXiv:0804.0802}.

\bibitem{aaronson08}
Scott Aaronson and Daniel Gottesman.
\newblock Identifying stabilizer states, 2008.
\newblock {\tt http://pirsa.org/08080052/}.

\bibitem{aaronson&shi:collision}
Scott Aaronson and Yaoyun Shi.
\newblock Quantum lower bounds for the collision and the element distinctness
  problems.
\newblock {\em Journal of the ACM}, 51(4):595--605, 2004.
\newblock Earlier versions in STOC'02 and FOCS'02.

\bibitem{acin01}
Antonio Ac\'in.
\newblock Statistical distinguishability between unitary operations.
\newblock {\em Physical Review Letters}, 87(17):177901, 2001.
\newblock {\tt quant-ph/0102064}.

\bibitem{ABGMPS:qkd}
Antonio Acin, Nicolas Brunner, Nicolas Gisin, Serge Massar, Stefano Pironio,
  and Valerio Scarani.
\newblock Device-independent security of quantum cryptography against
  collective attacks.
\newblock {\em Physical Review Letters}, 98:230501, 2008.

\bibitem{aharonov13}
Dorit Aharonov, Itai Arad, and Thomas Vidick.
\newblock The quantum {PCP} conjecture.
\newblock {\em ACM SIGACT News}, 44(2):47--79, 2013.
\newblock Enhanced version at {\tt arXiv:1309.7495}.

\bibitem{aharonov13a}
Dorit Aharonov and Lior Eldar.
\newblock Quantum locally testable codes.
\newblock {\em SIAM Journal on Computing}, 44(5):1230--1262, 2015.
\newblock {\tt arXiv:1310.5664}.

\bibitem{aharonov14}
Dorit Aharonov, Aram~W. Harrow, Zeph Landau, Daniel Nagaj, Mario Szegedy, and
  Umesh Vazirani.
\newblock Local tests of global entanglement and a counterexample to the
  generalized area law.
\newblock In {\em Proceedings of 55th Annual IEEE Symposium on Foundations of
  Computer Science (FOCS)}, pages 246--255, 2014.
\newblock {\tt arXiv:1410.0951}.

\bibitem{alon09}
Noga Alon, Eldar Fischer, Ilan Newman, and Asaf Shapira.
\newblock A combinatorial characterization of the testable graph properties:
  it's all about regularity.
\newblock {\em SIAM Journal on Computing}, 39(1):143--167, 2009.

\bibitem{alon05}
Noga Alon, Tali Kaufman, Michael Krivelevich, Simon Litsyn, and Dana Ron.
\newblock Testing {R}eed-{M}uller codes.
\newblock {\em {IEEE} Transactions on Information Theory}, 51(11):4032--4038,
  2005.

\bibitem{ambainis:lowerboundsj}
Andris Ambainis.
\newblock Quantum lower bounds by quantum arguments.
\newblock {\em Journal of Computer and System Sciences}, 64(4):750--767, 2002.
\newblock Earlier version in STOC'00. {\tt quant-ph/0002066}.

\bibitem{ambainis:edj}
Andris Ambainis.
\newblock Quantum walk algorithm for element distinctness.
\newblock {\em SIAM Journal on Computing}, 37(1):210--239, 2007.
\newblock Earlier version in FOCS'04. {\tt quant-ph/0311001}.

\bibitem{abrw:juntatesting}
Andris Ambainis, Aleksandrs Belovs, Oded Regev, and Ronald~{de} Wolf.
\newblock Efficient quantum algorithms for (gapped) group testing and junta
  testing.
\newblock In {\em Proceedings of 27th ACM-SIAM Symposium on Discrete Algorithms
  (SODA)}, pages 903--922, 2016.
\newblock {\tt arXiv:1507.03126}.

\bibitem{acl:testing}
Andris Ambainis, Andrew~M. Childs, and Yi-Kai Liu.
\newblock Quantum property testing for bounded-degree graphs.
\newblock In {\em Proceedings of 15th RANDOM}, volume 6845 of {\em Lecture
  Notes in Computer Science}, pages 365--376, 2011.
\newblock {\tt arXiv:1012.3174}.

\bibitem{almss:pcp}
Sanjeev Arora, Carsten Lund, Rajeev Motwani, Madhu Sudan, and Mario Szegedy.
\newblock Proof verification and the hardness of approximation problems.
\newblock {\em Journal of the ACM}, 45(3):501--555, 1998.
\newblock Earlier version in FOCS'92.

\bibitem{atici&servedio:testing}
Alp {At\i c\i} and Rocco~A. Servedio.
\newblock Quantum algorithms for learning and testing juntas.
\newblock {\em Quantum Information Processing}, 6(5):323--348, 2009.
\newblock {\tt arXiv:0707.3479}.

\bibitem{audenaert08}
K.~M.~R. Audenaert, M.~Nussbaum, A.~Szko{\l a}, and F.~Verstraete.
\newblock Asymptotic error rates in quantum hypothesis testing.
\newblock {\em Communications in Mathematical Physics}, 279:251--283, 2008.
\newblock {\tt arXiv:0708.4282}.

\bibitem{audenaert06}
Koenraad M.~R. Audenaert.
\newblock A digest on representation theory of the symmetric group, 2006.
\newblock Available at
  \url{http://personal.rhul.ac.uk/usah/080/QITNotes_files/Irreps_v06.pdf}.

\bibitem{babai91}
L{\'a}szl{\'o} Babai, Lance Fortnow, and Carsten Lund.
\newblock Non-deterministic exponential time has two-prover interactive
  protocols.
\newblock {\em Computational Complexity}, 1:3--40, 1991.
\newblock Earlier version in FOCS'90.

\bibitem{bacon06}
David Bacon, Isaac~L. Chuang, and Aram~W. Harrow.
\newblock Efficient quantum circuits for {S}chur and {C}lebsch-{G}ordan
  transforms.
\newblock {\em Physical Review Letters}, 97(17):170502, 2006.
\newblock {\tt quant-ph/0407082}.

\bibitem{barenco97}
Adriano Barenco, Andr{\'e} Berthiaume, David Deutsch, Artur Ekert, Richard
  Jozsa, and Chiara Macchiavello.
\newblock Stabilization of quantum computations by symmetrisation.
\newblock {\em SIAM Journal on Computing}, 26(5):1541--1557, 1997.
\newblock {\tt quant-ph/9604028}.

\bibitem{barnett09}
Stephen~M. Barnett and Sarah Croke.
\newblock Quantum state discrimination.
\newblock {\em Advances in Optics and Photonics}, 1(2):238--278, 2009.
\newblock {\tt arXiv:0810.1970}.

\bibitem{bhk:crypto}
Jonathan Barrett, Lucien Hardy, and Adrian Kent.
\newblock No signaling and quantum key distribution.
\newblock {\em Physical Review Letters}, 95(1):010503, June 2005.

\bibitem{BFF+01}
Tu\v{g}kan Batu, Lance Fortnow, Eldar Fischer, Ravi Kumar, Ronitt Rubinfeld,
  and Patrick White.
\newblock Testing random variables for independence and identity.
\newblock In {\em Proceedings of 42nd Annual IEEE Symposium on Foundations of
  Computer Science (FOCS)}, pages 442--451, 2001.

\bibitem{batu13}
Tu\v{g}kan Batu, Lance Fortnow, Ronitt Rubinfeld, Warren~D. Smith, and Patrick
  White.
\newblock Testing closeness of discrete distributions.
\newblock {\em Journal of the ACM}, 60(1), 2013.
\newblock Previous version in FOCS'00. {\tt arXiv:1009.5397}.

\bibitem{beals97}
Robert Beals.
\newblock Quantum computation of {F}ourier transforms over symmetric groups.
\newblock In {\em Proceedings of 29th Annual ACM Symposium on Theory of
  Computing (STOC)}, pages 48--53, 1997.

\bibitem{bbcmw:polynomialsj}
Robert Beals, Harry Buhrman, Richard Cleve, Michele Mosca, and Ronald~{de}
  Wolf.
\newblock Quantum lower bounds by polynomials.
\newblock {\em Journal of the ACM}, 48(4):778--797, 2001.
\newblock Earlier version in FOCS'98. {\tt quant-ph/9802049}.

\bibitem{bellare96}
Mihir Bellare, Don Coppersmith, Johan H{\aa}stad, Marcos Kiwi, and Madhu Sudan.
\newblock Linearity testing in characteristic two.
\newblock {\em {IEEE} Transactions on Information Theory}, 42(6), 1996.

\bibitem{belovs:juntalearning}
Aleksandrs Belovs.
\newblock Quantum algorithms for learning symmetric juntas via adversary bound.
\newblock In {\em Proceedings of 29th Annual IEEE Conference on Computational
  Complexity (CCC)}, pages 22--31, 2014.
\newblock {\tt arXiv:1311.6777}.

\bibitem{bennett97}
Charles~H. Bennett, Ethan Bernstein, Gilles Brassard, and Umesh Vazirani.
\newblock Strengths and weaknesses of quantum computing.
\newblock {\em SIAM Journal on Computing}, 26(5):1510--1523, 1997.
\newblock {\tt quant-ph/9701001}.

\bibitem{bernstein&vazirani:qcomplexity}
Ethan Bernstein and Umesh Vazirani.
\newblock Quantum complexity theory.
\newblock {\em SIAM Journal on Computing}, 26(5):1411--1473, 1997.
\newblock Earlier version in STOC'93.

\bibitem{blais_juntas}
Eric Blais.
\newblock Testing juntas nearly optimally.
\newblock In {\em Proceedings of 41st Annual ACM Symposium on Theory of
  Computing (STOC)}, pages 151--158, 2009.

\bibitem{bbmj}
Eric Blais, Joshua Brody, and Kevin Matulef.
\newblock Property testing lower bounds via communication complexity.
\newblock {\em Computational Complexity}, 21(2):311--358, 2012.
\newblock Earlier version in Complexity'11.

\bibitem{blr:selftest}
Manuel Blum, Michael Luby, and Ronitt Rubinfeld.
\newblock Self-testing/correcting with applications to numerical problems.
\newblock {\em Journal of Computer and System Sciences}, 47(3):549--595, 1993.
\newblock Earlier version in STOC'90.

\bibitem{bookatz12}
Adam~D. Bookatz.
\newblock {QMA}-complete problems.
\newblock {\em Quantum Information and Computation}, 14(5-6):361--383, 2014.
\newblock {\tt arXiv:1212.6312}.

\bibitem{brassard&hoyer:simon}
Gilles Brassard and Peter H{\o}yer.
\newblock An exact quantum polynomial-time algorithm for {S}imon's problem.
\newblock In {\em Proceedings of the 5th Israeli Symposium on Theory of
  Computing and Systems (ISTCS'97)}, pages 12--23, 1997.
\newblock {\tt quant-ph/9704027}.

\bibitem{bhmt:countingj}
Gilles Brassard, Peter H{\o}yer, Michele Mosca, and Alain Tapp.
\newblock Quantum amplitude amplification and estimation.
\newblock In {\em Quantum Computation and Quantum Information: A Millennium
  Volume}, volume 305 of {\em AMS Contemporary Mathematics Series}, pages
  53--74. 2002.
\newblock {\tt quant-ph/0005055}.

\bibitem{bmr:maxbellviolation}
Samuel~L. Braunstein, Ady Mann, and Michael Revzen.
\newblock Maximal violation of {B}ell inequalities for mixed states.
\newblock {\em Physical Review Letters}, 68(22):3259--3261, 1992.

\bibitem{BHH10j}
Sergey Bravyi, Aram~W. Harrow, and Avinatan Hassidim.
\newblock Quantum algorithms for testing properties of distributions.
\newblock {\em {IEEE} Transactions on Information Theory}, 57(6):3971--3981,
  2011.
\newblock Earlier version in STACS'10. {\tt arXiv:0907.3920}.

\bibitem{brun04}
Todd~A. Brun.
\newblock Measuring polynomial functions of states.
\newblock {\em Quantum Information and Computation}, 4:401, 2004.
\newblock {\tt quant-ph/0401067}.

\bibitem{bcpsw:bellnonlocality}
Nicolas Brunner, Daniel Cavalcanti, Stefano Pironio, Valerio Scarani, and
  Stephanie Wehner.
\newblock Bell nonlocality.
\newblock {\em Reviews of Modern Physics}, 86:419--478, 2014.
\newblock {\tt arXiv:1303.2849}.

\bibitem{bruss99}
Dagmar Bru{\ss} and Chiara Macchiavello.
\newblock Optimal state estimation for {$d$}-dimensional quantum systems.
\newblock {\em Physical Review A}, 253(5--6):249--251, 1999.
\newblock {\tt quant-ph/9812016}.

\bibitem{buhrman01}
Harry Buhrman, Richard Cleve, John Watrous, and Ronald de~Wolf.
\newblock Quantum fingerprinting.
\newblock {\em Physical Review Letters}, 87(16):167902, 2001.
\newblock {\tt quant-ph/0102001}.

\bibitem{bfnr:qpropj}
Harry Buhrman, Lance Fortnow, Ilan Newman, and Hein R{\"o}hrig.
\newblock Quantum property testing.
\newblock {\em SIAM Journal on Computing}, 37(5):1387--1400, 2008.
\newblock Earlier version in SODA'03. {\tt quant-ph/0201117}.

\bibitem{bgmw:kparity}
Harry Buhrman, David {Garc\'{i}a-Soriano}, Arie Matsliah, and Ronald~{de} Wolf.
\newblock The non-adaptive query complexity of testing k-parities.
\newblock {\em Chicago Journal of Theoretical Computer Science}, 6, 2013.
\newblock {\tt arXiv:1209.3849}.

\bibitem{buhrman&wolf:dectreesurvey}
Harry Buhrman and Ronald~{de} Wolf.
\newblock Complexity measures and decision tree complexity: A survey.
\newblock {\em Theoretical Computer Science}, 288(1):21--43, 2002.

\bibitem{chakraborty13}
Kaushik Chakraborty and Subhamoy Maitra.
\newblock Improved quantum test for linearity of a {B}oolean function, 2013.
\newblock {\tt arXiv:1306.6195}.

\bibitem{cfmw:qproptesting}
Sourav Chakraborty, Eldar Fischer, Arie Matsliah, and Ronald~{de} Wolf.
\newblock New results on quantum property testing.
\newblock In {\em Proceedings of FSTTCS}, pages 145--156, 2010.
\newblock {\tt arXiv:1005.0523}.

\bibitem{chan13}
Siu-On Chan, Ilias Diakonikolas, Gregory Valiant, and Paul Valiant.
\newblock Optimal algorithms for testing closeness of discrete distributions.
\newblock In {\em Proceedings of 25th ACM-SIAM Symposium on Discrete Algorithms
  (SODA)}, pages 1193--1203, 2014.
\newblock {\tt arXiv:1308.3946}.

\bibitem{chefles00}
Anthony Chefles.
\newblock Quantum state discrimination.
\newblock {\em Contemporary Physics}, 41(6):401--424, 2001.
\newblock {\tt quant-ph/0010114}.

\bibitem{childs07c}
Andrew Childs, Aram~W. Harrow, and Pawel Wocjan.
\newblock Weak {F}ourier-{S}chur sampling, the hidden subgroup problem, and the
  quantum collision problem.
\newblock In {\em Proceedings of 24th Symposium on Theoretical Aspects of
  Computer Science (STACS)}, pages 598--609, 2007.
\newblock {\tt quant-ph/0609110}.

\bibitem{chockler&gutfreund}
Hana Chockler and Dan Gutfreund.
\newblock A lower bound for testing juntas.
\newblock {\em Information Processing Letters}, 90(6):301--305, 2004.

\bibitem{choi75}
Man-Duen Choi.
\newblock Completely positive linear maps on complex matrices.
\newblock {\em Linear Algebra and its Applications}, 10(3):285--290, 1975.

\bibitem{christandl06}
Matthias Christandl.
\newblock {\em The Structure of Bipartite Quantum States -- Insights from Group
  Theory and Cryptography}.
\newblock PhD thesis, University of Cambridge, 2006.
\newblock {\tt quant-ph/0604183}.

\bibitem{tsirelson80}
Boris~S. Cirel'son.
\newblock Quantum generalizations of {B}ell's inequality.
\newblock {\em Letters in Mathematical Physics}, 4(2):93--100, 1980.

\bibitem{chsh}
John~F. Clauser, Michael~A. Horne, Abner Shimony, and Richard~A. Holt.
\newblock Proposed experiment to test local hidden-variable theories.
\newblock {\em Physical Review Letters}, 23(15):880--884, 1969.

\bibitem{colbeck:phd}
Roger Colbeck.
\newblock {\em Quantum and relativistic protocols for secure multi-party
  computation}.
\newblock PhD thesis, University of Cambridge, 2006.
\newblock {\tt arXiv:0911.3814}.

\bibitem{cramer10}
Marcus Cramer, Martin~B. Plenio, Steven~T. Flammia, Rolando Somma, David Gross,
  Stephen~D. Bartlett, Olivier Landon-Cardinal, David Poulin, and Yi-Kai Liu.
\newblock Efficient quantum state tomography.
\newblock {\em Nature Communications}, 1(9):49, 2010.
\newblock {\tt arXiv:1101.4366}.

\bibitem{dmms:selftestj}
Wim~{van} Dam, Fr{\'e}d{\'e}ric Magniez, Michele Mosca, and Miklos Santha.
\newblock Self-testing of universal and fault-tolerant sets of quantum gates.
\newblock {\em SIAM Journal on Computing}, 37(2):611--629, 2007.
\newblock Earlier version in STOC'00. {\tt quant-ph/9904108}.

\bibitem{dks:sparsedisj}
Anirban Dasgupta, Ravi Kumar, and D.~Sivakumar.
\newblock Sparse and lopsided set disjointness via information theory.
\newblock In {\em Proceedings of 16th RANDOM}, pages 517--528, 2012.

\bibitem{duan09}
Runyao Duan, Yuan Feng, and Mingsheng Ying.
\newblock Perfect distinguishability of quantum operations.
\newblock {\em Physical Review Letters}, 103:210501, 2009.
\newblock {\tt arXiv:0908.0119}.

\bibitem{fis_sur}
Eldar Fischer.
\newblock The art of uninformed decisions.
\newblock {\em Bulletin of the EATCS}, 75:97, 2001.

\bibitem{fkrss:juntas}
Eldar Fischer, Guy Kindler, Dana Ron, Samuel Safra, and Alex Samorodnitsky.
\newblock Testing juntas.
\newblock {\em Journal of Computer and System Sciences}, 68(4):753--787, 2004.
\newblock Earlier version in FOCS'02.

\bibitem{flammia12}
Steven~T. Flammia, David Gross, Yi-Kai Liu, and Jens Eisert.
\newblock Quantum tomography via compressed sensing: Error bounds, sample
  complexity, and efficient estimators.
\newblock {\em New Journal of Physics}, 14:095022, 2012.
\newblock {\tt arXiv:1205.2300}.

\bibitem{flammia11}
Steven~T. Flammia and Yi-Kai Liu.
\newblock Direct fidelity estimation from few {P}auli measurements.
\newblock {\em Physical Review Letters}, 106:230501, 2011.
\newblock {\tt arXiv:1104.4695}.

\bibitem{friedl05}
Katalin Friedl, G{\'a}bor Ivanyos, and Miklos Santha.
\newblock Efficient testing of groups.
\newblock In {\em Proceedings of 37th Annual ACM Symposium on Theory of
  Computing (STOC)}, pages 157--166, 2005.

\bibitem{fmsp:testhidden}
Katalin Friedl, Fr{\'e}d{\'e}ric Magniez, Miklos Santha, and Pranab Sen.
\newblock Quantum testers for hidden group properties.
\newblock {\em Fundamenta Informaticae}, 91(2):325--340, 2009.
\newblock Earlier version in MFCS'03.

\bibitem{gharibian10}
Sevag Gharibian.
\newblock Strong {NP}-hardness of the quantum separability problem.
\newblock {\em Quantum Information and Computation}, 10(3{\&}4):343--360, 2010.
\newblock {\tt arXiv:0810.4507}.

\bibitem{glebsky10}
Lev Glebsky.
\newblock Almost commuting matrices with respect to normalized
  {H}ilbert-{S}chmidt norm, 2010.
\newblock {\tt arXiv:1002.3082}.

\bibitem{goldreich:prop}
Oded Goldreich, editor.
\newblock {\em Property Testing: Current Research and Surveys}, volume 6390 of
  {\em Lecture Notes in Computer Science}.
\newblock Springer, 2010.

\bibitem{ggr}
Oded Goldreich, Shafi Goldwasser, and Dana Ron.
\newblock Property testing and its connection to learning and approximation.
\newblock {\em Journal of the ACM}, 45(4):653--750, 1998.

\bibitem{gr}
Oded Goldreich and Dana Ron.
\newblock On testing expansion in bounded-degree graphs.
\newblock {\em Electronic Colloquium on Computational Complexity (ECCC)},
  7(20), 2000.

\bibitem{gr:boundeddegree}
Oded Goldreich and Dana Ron.
\newblock Property testing in bounded degree graphs.
\newblock {\em Algorithmica}, 32(2):302--343, 2002.

\bibitem{gottesman99}
Daniel Gottesman.
\newblock {\em Stabilizer Codes and Quantum Error Correction}.
\newblock PhD thesis, Caltech, 1999.
\newblock {\tt quant-ph/9705052}.

\bibitem{ghz}
Daniel~M. Greenberger, Michael~A. Horne, and Anton Zeilinger.
\newblock Going beyond {B}ell's theorem.
\newblock In M.~Kafatos, editor, {\em Bell's Theorem, Quantum Theory, and
  Conceptions of the Universe}, pages 69--72. Kluwer, 1989.
\newblock {\tt arXiv:0712.0921}.

\bibitem{gross09b}
David Gross, Yi-Kai Liu, Steven~T. Flammia, Stephen Becker, and Jens Eisert.
\newblock Quantum state tomography via compressed sensing.
\newblock {\em Physical Review Letters}, 105:150401, 2010.
\newblock {\tt arXiv:0909.3304}.

\bibitem{grover:search}
Lov~K. Grover.
\newblock A fast quantum mechanical algorithm for database search.
\newblock In {\em Proceedings of 28th Annual ACM Symposium on Theory of
  Computing (STOC)}, pages 212--219, 1996.
\newblock {\tt quant-ph/9605043}.

\bibitem{guhne09}
Otfried G{\"{u}}hne and Geza Toth.
\newblock Entanglement detection.
\newblock {\em Physics Reports}, 471(1), 2009.
\newblock {\tt arXiv:0811.2803}.

\bibitem{gurvits03}
Leonid Gurvits.
\newblock Classical deterministic complexity of {E}dmonds' problem and quantum
  entanglement.
\newblock In {\em Proceedings of 35th Annual ACM Symposium on Theory of
  Computing (STOC)}, pages 10--19, 2003.
\newblock {\tt quant-ph/0303055}.

\bibitem{milner13}
Gus Gutoski, Patrick Hayden, Kevin Milner, and Mark~M. Wilde.
\newblock Quantum interactive proofs and the complexity of separability
  testing.
\newblock {\em Theory of Computing}, 11(3):59--103, 2015.
\newblock {\tt arXiv:1308.5788}.

\bibitem{haffner05}
H.~H{\"a}ffner, W.~H{\"a}nsel, C.~Roos, J.~Benhelm, D.~{Chek-al-kar},
  M.~Chwalla, T.~K{\"o}rber, U.~Rapol, M.~Riebe, P.~Schmidt, C.~Becher,
  O.~G{\"u}hne, W.~D{\"u}r, and R.~Blatt.
\newblock Scalable multiparticle entanglement of trapped ions.
\newblock {\em Nature}, 438:643--646, 2005.
\newblock {\tt quant-ph/0603217}.

\bibitem{hales:phd}
Lisa Hales.
\newblock {\em The Quantum Fourier Transform and Extensions of the Abelian
  Hidden Subgroup Problem}.
\newblock PhD thesis, University of California, Berkeley, 2002.
\newblock {\tt quant-ph/0212002}.

\bibitem{hales&hallgren:improvedfourier}
Lisa Hales and Sean Hallgren.
\newblock An improved quantum {F}ourier transform algorithm and applications.
\newblock In {\em Proceedings of 41st Annual IEEE Symposium on Foundations of
  Computer Science (FOCS)}, pages 515--525, 2000.

\bibitem{harrow05}
Aram~W. Harrow.
\newblock {\em Applications of coherent classical communication and the {S}chur
  transform to quantum information theory}.
\newblock PhD thesis, Massachusetts Institute of Technology, 2005.
\newblock {\tt quant-ph/0512255}.

\bibitem{harrow13}
Aram~W. Harrow and Ashley Montanaro.
\newblock Testing product states, quantum {M}erlin-{A}rthur games and tensor
  optimization.
\newblock {\em Journal of the ACM}, 60(1), 2013.
\newblock Earlier version in FOCS'10. {\tt arXiv:1001.0017}.

\bibitem{hayden06}
Patrick Hayden, Debbie~W. Leung, and Andreas Winter.
\newblock Aspects of generic entanglement.
\newblock {\em Communications in Mathematical Physics}, 265(1):95--117, 2006.
\newblock {\tt quant-ph/0407049}.

\bibitem{hayden13}
Patrick Hayden, Kevin Milner, and Mark~M. Wilde.
\newblock Two-message quantum interactive proofs and the quantum separability
  problem.
\newblock In {\em Proceedings of 28th Annual IEEE Conference on Computational
  Complexity (CCC)}, pages 156--167, 2013.
\newblock {\tt arXiv:1211.6120}.

\bibitem{helstrom76}
Carl~W. Helstrom.
\newblock {\em Quantum detection and estimation theory}.
\newblock Academic Press, New York, 1976.

\bibitem{hillery11}
Mark Hillery and Erika Andersson.
\newblock Quantum tests for the linearity and permutation invariance of
  {B}oolean functions.
\newblock {\em Physical Review A}, 84:062329, 2011.
\newblock {\tt arXiv:1106.4831}.

\bibitem{holevo73}
Alexander~S. Holevo.
\newblock Bounds for the quantity of information transmitted by a quantum
  communication channel.
\newblock {\em Problemy Peredachi Informatsii}, 9(3):3--11, 1973.
\newblock English translation {\em Problems of Information Transmission}, vol.
  9, pp. 177-183, 1973.

\bibitem{hlw:expander}
Shlomo Hoory, Nathan Linial, and Avi Wigderson.
\newblock Expander graphs and their applications.
\newblock {\em Bulletin of the AMS}, 43:439--561, 2006.

\bibitem{horodecki09}
Ryszard Horodecki, Pawe{\l{ }} Horodecki, Micha{\l{ }} Horodecki, and Karol
  Horodecki.
\newblock Quantum entanglement.
\newblock {\em Reviews of Modern Physics}, 81:865--942, 2009.
\newblock {\tt quant-ph/0702225}.

\bibitem{hls:madv}
Peter H{\o}yer, Troy Lee, and Robert {\v{S}}palek.
\newblock Negative weights make adversaries stronger.
\newblock In {\em Proceedings of 39th Annual ACM Symposium on Theory of
  Computing (STOC)}, pages 526--535, 2007.
\newblock {\tt quant-ph/0611054}.

\bibitem{inui&legall:testingj}
Yoshifumi Inui and Francois Le~Gall.
\newblock Quantum property testing of group solvability.
\newblock {\em Algorithmica}, 59(1):35--47, 2011.
\newblock Earlier version in LATIN'08. {\tt arXiv:0712.3829}.

\bibitem{ito12}
Tsuyoshi Ito and Thomas Vidick.
\newblock A multi-prover interactive proof for {NEXP} sound against entangled
  provers.
\newblock In {\em Proceedings of 53rd Annual IEEE Symposium on Foundations of
  Computer Science (FOCS)}, pages 243--252, 2012.
\newblock {\tt arXiv:1207.0550}.

\bibitem{jain11}
Rahul Jain, Zhengfeng Ji, Sarvagya Upadhyay, and John Watrous.
\newblock {QIP=PSPACE}.
\newblock {\em Journal of the ACM}, 58(6):30, 2011.
\newblock Earlier version in STOC'10. {\tt arXiv:0907.4737}.

\bibitem{jamiolkowski72}
A.~Jamio{\l }kowski.
\newblock Linear transformations which preserve trace and positive
  semidefiniteness of operators.
\newblock {\em Reports on Mathematical Physics}, 3(4):275--278, 1972.

\bibitem{janzing05}
Dominik Janzing, Pawel Wocjan, and Thomas Beth.
\newblock Non-identity check is {QMA}-complete.
\newblock {\em International Journal of Quantum Information}, 3:463, 2005.
\newblock {\tt quant-ph/0305050}.

\bibitem{kada08}
Masaru Kada, Harumichi Nishimura, and Tomoyuki Yamakami.
\newblock The efficiency of quantum identity testing of multiple states.
\newblock {\em Journal of Physics A: Mathematical and Theoretical}, 41:395309,
  2008.
\newblock {\tt arXiv:0809.2037}.

\bibitem{ks:disj}
Bala Kalyanasundaram and Georg Schnitger.
\newblock The probabilistic communication complexity of set intersection.
\newblock {\em SIAM Journal on Discrete Mathematics}, 5(4):545--557, 1992.
\newblock Earlier version in Structures'87.

\bibitem{kane09}
Daniel~M. Kane and Samuel~A. Kutin.
\newblock Quantum interpolation of polynomials.
\newblock {\em Quantum Information and Computation}, 1(1\&2):95--103, 2011.
\newblock {\tt arXiv:0909.5683}.

\bibitem{kempe06}
Julia Kempe, Alexei Kitaev, and Oded Regev.
\newblock The complexity of the local {H}amiltonian problem.
\newblock {\em SIAM Journal on Computing}, 35(5):1070--1097, 2006.
\newblock {\tt quant-ph/0406180}.

\bibitem{keyl01}
M.~Keyl and R.~F. Werner.
\newblock Estimating the spectrum of a density operator.
\newblock {\em Physical Review A}, 64(5):052311, 2001.
\newblock {\tt quant-ph/0102027}.

\bibitem{kitaev02}
A.~Yu. Kitaev, A.~H. Shen, and M.~N. Vyalyi.
\newblock {\em Classical and Quantum Computation}, volume~47 of {\em Graduate
  Studies in Mathematics}.
\newblock AMS, 2002.

\bibitem{kitaev00}
Alexei Kitaev and John Watrous.
\newblock Parallelization, amplification, and exponential time simulation of
  quantum interactive proof systems.
\newblock In {\em Proceedings of 32nd Annual ACM Symposium on Theory of
  Computing (STOC)}, pages 608--617, 2000.

\bibitem{kobayashi03}
Hirotada Kobayashi, Keiji Matsumoto, and Tomoyuki Yamakami.
\newblock Quantum {M}erlin-{A}rthur proof systems: Are multiple {M}erlins more
  helpful to {A}rthur?
\newblock {\em Chicago Journal of Theoretical Computer Science}, 2009.
\newblock Earlier version in ISAAC'03. {\tt quant-ph/0306051}.

\bibitem{knp:simon}
Pascal Koiran, Vincent Nesme, and Natacha Portier.
\newblock A quantum lower bound for the query complexity of {S}imon's problem.
\newblock In {\em Proceedings of 32nd International Colloquium on Automata,
  Languages and Programming (ICALP)}, volume 3580 of {\em Lecture Notes in
  Computer Science}, pages 1287--1298, 2005.
\newblock {\tt quant-ph/0501060}.

\bibitem{krauthgamer03}
Robert Krauthgamer and Ori Sasson.
\newblock Property testing of data dimensionality.
\newblock In {\em Proceedings of 14th ACM-SIAM Symposium on Discrete Algorithms
  (SODA)}, pages 18--27, 2003.

\bibitem{kushilevitz&nisan:cc}
Eyal Kushilevitz and Noam Nisan.
\newblock {\em Communication Complexity}.
\newblock Cambridge University Press, 1997.

\bibitem{lachish&newman:periodicity}
Oded Lachish and Ilan Newman.
\newblock Testing periodicity.
\newblock {\em Algorithmica}, 60(2):401--420, 2011.
\newblock Earlier version in RANDOM'05.

\bibitem{legall11j}
Fran\c{c}ois {Le Gall} and Yuichi Yoshida.
\newblock Property testing for cyclic groups and beyond.
\newblock {\em Journal of Combinatorial Optimization}, 26(4):636--654, 2013.
\newblock Earlier version in COCOON'11. {\tt arXiv:1105.1842}.

\bibitem{lmrss:stateconv}
Troy Lee, Rajan Mittal, Ben~W. Reichardt, Robert {\v{S}}palek, and Mario
  Szegedy.
\newblock Quantum query complexity of state conversion.
\newblock In {\em Proceedings of 52nd Annual IEEE Symposium on Foundations of
  Computer Science (FOCS)}, pages 344--353, 2011.
\newblock {\tt arXiv:1011.3020}.

\bibitem{low09}
Richard~A. Low.
\newblock Learning and testing algorithms for the {C}lifford group.
\newblock {\em Physical Review A}, 80:052314, 2009.
\newblock {\tt arXiv:0907.2833}.

\bibitem{macwilliams83}
F.~J. MacWilliams and N.~J.~A. Sloane.
\newblock {\em The Theory of Error-Correcting Codes}.
\newblock North-Holland, Amsterdam, 1983.

\bibitem{mmmo:selftesting}
Fr{\'e}d{\'e}ric Magniez, Dominic Mayers, Michele Mosca, and Harold Ollivier.
\newblock Self-testing of quantum circuits.
\newblock In {\em Proceedings of 33rd International Colloquium on Automata,
  Languages and Programming (ICALP)}, volume 4051 of {\em Lecture Notes in
  Computer Science}, pages 72--83, 2006.
\newblock {\tt quant-ph/0512111}.

\bibitem{majewski09}
Krzysztof Majewski and Nicholas Pippenger.
\newblock Attribute estimation and testing quasi-symmetry.
\newblock {\em Information Processing Letters}, 109(4):233--237, 2007.
\newblock {\tt arXiv:0708.2105}.

\bibitem{mayers&yao:impapp}
Dominic Mayers and Andrew C-C. Yao.
\newblock Quantum cryptography with imperfect apparatus.
\newblock In {\em Proceedings of 39th Annual IEEE Symposium on Foundations of
  Computer Science (FOCS)}, pages 503--509, 1998.
\newblock {\tt quant-ph/9809039}.

\bibitem{mayers&yao:selftesting}
Dominic Mayers and Andrew C-C. Yao.
\newblock Self testing quantum apparatus.
\newblock {\em Quantum Information and Computation}, 4(4):273--286, 2004.
\newblock {\tt quant-ph/0307205}.

\bibitem{mckague:selftestgraph}
Matthew {McKague}.
\newblock Self-testing graph states.
\newblock In {\em Proceedings of 6th Conference on Theory of Quantum
  Computation, Communication, and Cryptography (TQC)}, pages 104--120, 2011.
\newblock {\tt arXiv:1010.1989}.

\bibitem{mys:robustselftest}
Matthew {McKague}, Tzyh~Haur Yang, and Valerio Scarani.
\newblock Robust self-testing of the singlet.
\newblock {\em Journal of Physics A: Mathematical and Theoretical}, 45(455304),
  2012.
\newblock {\tt arXiv:1203.2976}.

\bibitem{miller&shi:selftestingxor}
Carl~A. Miller and Yaoyun Shi.
\newblock Optimal robust quantum self-testing by binary nonlocal {XOR} games.
\newblock In {\em Proceedings of 8th Conference on Theory of Quantum
  Computation, Communication, and Cryptography (TQC)}, 2013.
\newblock {\tt arXiv:1207.1819}.

\bibitem{mintert05}
Florian Mintert, Marek Ku{\'s}, and Andreas Buchleitner.
\newblock {Concurrence of mixed multipartite quantum states}.
\newblock {\em Physical Review Letters}, 95(26):260502, 2005.
\newblock quant-ph/0411127.

\bibitem{montanaro09c}
Ashley Montanaro.
\newblock Symmetric functions of qubits in an unknown basis.
\newblock {\em Physical Review A}, 79(6):062316, 2009.
\newblock {\tt arXiv:0903.5466}.

\bibitem{montanaro10c}
Ashley Montanaro and Tobias Osborne.
\newblock Quantum boolean functions.
\newblock {\em Chicago Journal of Theoretical Computer Science}, 1, 2010.
\newblock {\tt arXiv:0810.2435}.

\bibitem{nielsen00b}
Michael~A. Nielsen.
\newblock Continuity bounds for entanglement.
\newblock {\em Physical Review A}, 61(6):064301, 2000.
\newblock {\tt quant-ph/9908086}.

\bibitem{nielsen&chuang:qc}
Michael~A. Nielsen and Isaac~L. Chuang.
\newblock {\em Quantum Computation and Quantum Information}.
\newblock Cambridge University Press, 2000.

\bibitem{odonnell:analysis}
Ryan {O'Donnell}.
\newblock {\em Analysis of Boolean Functions}.
\newblock Cambridge University Press, 2014.

\bibitem{odonnell&wright:spectrum}
Ryan {O'Donnell} and John Wright.
\newblock Quantum spectrum testing.
\newblock In {\em Proceedings of 47th Annual ACM Symposium on Theory of
  Computing (STOC)}, pages 529--538, 2015.
\newblock {\tt arXiv:1501.05028}.

\bibitem{ogawa02}
Tomohiro Ogawa and Hiroshi Nagaoka.
\newblock A new proof of the channel coding theorem via hypothesis testing in
  quantum information theory.
\newblock In {\em Proceedings of 2002 {IEEE} International Symposium on
  Information Theory}, page~73, 2002.
\newblock {\tt quant-ph/0208139}.

\bibitem{paris04}
Matteo Paris and Jaroslav {\v R}eh{\'a}{\v c}ek, editors.
\newblock {\em Quantum State Estimation}, volume 649 of {\em Lecture Notes in
  Physics}.
\newblock Springer-Verlag, 2004.

\bibitem{perezgarcia07}
David P{\'e}rez-Garc{\'i}a, Frank Verstraete, Michael~M. Wolf, and Juan~Ignacio
  Cirac.
\newblock Matrix product state representations.
\newblock {\em Quantum Information and Computation}, 7:401, 2007.
\newblock {\tt quant-ph/0608197}.

\bibitem{piani09}
Marco Piani and John Watrous.
\newblock All entangled states are useful for channel discrimination.
\newblock {\em Physical Review Letters}, 102:250501, 2009.
\newblock {\tt arXiv:0901.2118}.

\bibitem{popescu&rohrlich:generic}
Sandu Popescu and Daniel Rohrlich.
\newblock Which states violate {B}ell's inequality maximally?
\newblock {\em Physics Letters A}, 166:411--414, 1992.

\bibitem{raussendorf03}
Robert Raussendorf, Dan Browne, and Hans Briegel.
\newblock Measurement-based quantum computation with cluster states.
\newblock {\em Physical Review A}, 68:022312, 2003.
\newblock {\tt quant-ph/0301052}.

\bibitem{razborov:disj}
Alexander~A. Razborov.
\newblock On the distributional complexity of disjointness.
\newblock {\em Theoretical Computer Science}, 106(2):385--390, 1992.

\bibitem{reichardt:tight}
Ben~W. Reichardt.
\newblock Span programs and quantum query complexity: The general adversary
  bound is nearly tight for every {B}oolean function.
\newblock In {\em Proceedings of 50th Annual IEEE Symposium on Foundations of
  Computer Science (FOCS)}, pages 544--551, 2009.
\newblock {\tt arXiv:0904.2759}.

\bibitem{ruv:leash}
Ben~W. Reichardt, Falk Unger, and Umesh Vazirani.
\newblock Classical command of quantum systems.
\newblock {\em Nature}, 496:456--460, 2013.
\newblock {\tt arXiv:1209:0448} and {\tt arXiv:1209:0449}.

\bibitem{ron_sur}
Dana Ron.
\newblock Property testing: A learning theory perspective.
\newblock {\em Foundations and Trends in Machine Learning}, 1(3):307--402,
  2008.

\bibitem{rosgen08}
William Rosgen.
\newblock Distinguishing short quantum computations.
\newblock In {\em Proceedings of 25th Symposium on Theoretical Aspects of
  Computer Science (STACS)}, pages 597--608, 2008.
\newblock {\tt arXiv:0712.2595}.

\bibitem{rosgen10a}
William Rosgen.
\newblock Computational distinguishability of degradable and antidegradable
  channels.
\newblock {\em Quantum Information and Computation}, 10:735--746, 2010.
\newblock {\tt arXiv:0911.2109}.

\bibitem{rosgen05}
William Rosgen and John Watrous.
\newblock On the hardness of distinguishing mixed-state quantum computations.
\newblock In {\em Proceedings of 20th Annual IEEE Conference on Computational
  Complexity (CCC)}, pages 344--354, 2005.
\newblock {\tt cs/0407056}.

\bibitem{sacchi05}
Massimiliano~F. Sacchi.
\newblock Optimal discrimination of quantum operations.
\newblock {\em Physical Review A}, 71:062340, 2005.
\newblock {\tt quant-ph/0505183}.

\bibitem{santha:qrwsurvey}
Miklos Santha.
\newblock Quantum walk based search algorithms.
\newblock In {\em Proceedings of 5th conference on Theory and Applications of
  Models of Computation (TAMC)}, pages 31--46, 2008.
\newblock {\tt arXiv:0808.0059}.

\bibitem{shor:factoring}
Peter~W. Shor.
\newblock Polynomial-time algorithms for prime factorization and discrete
  logarithms on a quantum computer.
\newblock {\em SIAM Journal on Computing}, 26(5):1484--1509, 1997.
\newblock Earlier version in FOCS'94. {\tt quant-ph/9508027}.

\bibitem{dasilva11}
Marcus P.~da Silva, Olivier Landon-Cardinal, and David Poulin.
\newblock Practical characterization of quantum devices without tomography.
\newblock {\em Physical Review Letters}, 107:210404, 2011.
\newblock {\tt arXiv:1104.3835}.

\bibitem{simon:power}
Daniel~R. Simon.
\newblock On the power of quantum computation.
\newblock {\em SIAM Journal on Computing}, 26(5):1474--1483, 1997.
\newblock Earlier version in FOCS'94.

\bibitem{spalek&szegedy:adversaryj}
Robert {\v S}palek and Mario Szegedy.
\newblock All quantum adversary methods are equivalent.
\newblock {\em Theory of Computing}, 2(1):1--18, 2006.
\newblock Earlier version in ICALP'05, quant-ph/0409116.

\bibitem{summers&werner:maxviolation}
Stephen~J. Summers and Reinhard Werner.
\newblock Maximal violation of {B}ell's inequalities is generic in quantum
  field theory.
\newblock {\em Communications in Mathematical Physics}, 110(2):247--259, 1987.

\bibitem{tsirelson:some}
Boris~S. Tsirelson.
\newblock Some results and problems on quantum {B}ell-type inequalities.
\newblock {\em Hadronic Journal Supplement}, 8:329--345, 1993.
\newblock Available at {\tt http://www.tau.ac.il/\~{
  }tsirel/download/hadron.pdf}.

\bibitem{valiant11}
Paul Valiant.
\newblock Testing symmetric properties of distributions.
\newblock {\em SIAM Journal on Computing}, 40(6):1927--1968, 2011.

\bibitem{vazirani&vidick:dice}
Umesh Vazirani and Thomas Vidick.
\newblock Certifiable quantum dice: or, true random number generation secure
  against quantum adversaries.
\newblock In {\em Proceedings of 44th Annual ACM Symposium on Theory of
  Computing (STOC)}, pages 61--76, 2012.
\newblock {\tt arXiv:1111.6054}.

\bibitem{vazirani&vidick:qkd}
Umesh Vazirani and Thomas Vidick.
\newblock Fully device-independent quantum key distribution.
\newblock {\em Physical Review Letters}, 113:140501, 2014.
\newblock {\tt arXiv:1210.1810}.

\bibitem{wang11}
Guoming Wang.
\newblock Property testing of unitary operators.
\newblock {\em Physical Review A}, 84:052328, 2011.
\newblock {\tt arXiv:1110.1133}.

\bibitem{wang12}
Guoming Wang.
\newblock Property testing of quantum measurements, 2012.
\newblock {\tt arXiv:1205.0828}.

\bibitem{watrous08a}
John Watrous.
\newblock \emph{Theory of Quantum Information} lecture notes, 2008.
\newblock \url{http://www.cs.uwaterloo.ca/~watrous/quant-info/}.

\bibitem{winter99}
Andreas Winter.
\newblock Coding theorem and strong converse for quantum channels.
\newblock {\em {IEEE} Transactions on Information Theory}, 45(7):2481--2485,
  1999.

\bibitem{wolf:fouriersurvey}
Ronald~{de} Wolf.
\newblock A brief introduction to {F}ourier analysis on the {B}oolean cube.
\newblock {\em Theory of Computing}, 2008.
\newblock ToC Library, Graduate Surveys 1.

\bibitem{yang&navascues}
Tzyh~Haur Yang and Miguel Navascu\'es.
\newblock Robust self testing of unknown quantum systems into any entangled
  two-qubit states.
\newblock {\em Physical Review A}, 87(5):050102(R), 2013.
\newblock {\tt arXiv:1210.4409}.

\bibitem{yvbsn:openingblackbox}
Tzyh~Haur Yang, Tam{\'a}s V\'ertesi, Jean-Daniel Bancal, Valerio Scarani, and
  Miguel Navascu\'es.
\newblock Opening the black box: how to estimate physical properties from
  non-local correlations, 26 Jul 2013.
\newblock {\tt arXiv:1307.7053}.

\bibitem{yao:unified}
Andrew C-C. Yao.
\newblock Probabilistic computations: Toward a unified measure of complexity.
\newblock In {\em Proceedings of 18th Annual IEEE Symposium on Foundations of
  Computer Science (FOCS)}, pages 222--227, 1977.

\bibitem{yao:distributive}
Andrew C-C. Yao.
\newblock Some complexity questions related to distributive computing.
\newblock In {\em Proceedings of 11th Annual ACM Symposium on Theory of
  Computing (STOC)}, pages 209--213, 1979.

\end{thebibliography}

\end{multicols}

\end{document}